\documentclass[a4paper]{llncs}

\pdfoutput=1
\usepackage[utf8]{inputenc}
\usepackage[T1]{fontenc}

\usepackage[usenames,dvipsnames]{xcolor}
\usepackage{amssymb,amsmath,amsfonts}
\usepackage{tikz,pgffor}
\usetikzlibrary{arrows}
\usetikzlibrary{shapes}
\usetikzlibrary{calc}
\usetikzlibrary{automata}
\usepackage{txfonts}
\usepackage{wrapfig}
\usepackage{cutwin}
\usepackage{stackrel}
\usepackage{enumitem}
\usepackage{eurosym}

\usepackage[linesnumbered,ruled,noend]{algorithm2e}

\usepackage{microtype}

\newcommand{\theoremlike}[2]{\par\medskip\penalty-250\refstepcounter{theorem}{{\bfseries\noindent#2
\ref{#1}.}}}

\newcommand{\thmhelperpre}[2]{\theoremlike{#1}{#2}}
\newcommand{\thmhelperpost}{\par\medskip}

\newenvironment{reflemma}[1]{\thmhelperpre{#1}{Lemma}}{\thmhelperpost}

\newenvironment{refproposition}[1]{\thmhelperpre{#1}{Proposition}}{
\thmhelperpost }

\newcommand{\appref}[1]{Appendix~\ref{#1}}

\newcommand{\myspacebig}{\vspace*{-1em}}
\newcommand{\myspacesmall}{\vspace*{-.6em}}
 \advance\textwidth3.4mm
 \advance\hoffset-1.7mm
 \advance\textheight1.6mm
 \advance\voffset-.8mm

\newcommand{\cylinderSet}{R^{s_0,\ldots,s_n}_{I_0,\ldots,I_{n-1}}}

\newcommand{\Nset}{\mathbb{N}}
\newcommand{\Nseto}{\Nset_0}

\newcommand{\Rset}{\mathbb{R}}
\newcommand{\Rsetp}{\mathbb{R}_{>0}}
\newcommand{\Rsetpo}{\mathbb{R}_{\ge 0}}

\renewcommand{\vec}[1]{\mathbf{#1}}
\newcommand{\de}[1]{\mathit{d#1}}
\newcommand{\eps}{\varepsilon}
\newcommand{\dist}{\mathcal{D}}

\newcommand{\probm}{\mathrm{Pr}}

\newcommand{\expred}{E}
\newcommand{\expected}{\mathbb{E}}

\newcommand{\constFactor}{N}
\newcommand{\size}[1]{||#1||}

\newcommand{\fdC}{C}

\newcommand{\states}{S}
\newcommand{\sta}{s}

\newcommand{\goalStates}{G}

\newcommand{\initstate}{\sta_{in}}

\newcommand{\gen}{Q}
\newcommand{\prob}{\mathrm{P}}

\newcommand{\allact}{\states_{\!\mathrm{fd}}}

\newcommand{\paramspace}{D}
\newcommand{\paramspacePO}{D(\dmin,\dmax,\equiv)}
\newcommand{\paramspacePOShort}{D}

\newcommand{\trans}{\mathrm{F}}

\newcommand{\delays}{d}
\newcommand{\timeouts}{\mathbf{d}}

\newcommand{\nodel}{\infty}

\newcommand{\nxt}{\mathrm{next}}

\newcommand{\prtr}[4]{\probm_{\rightarrow}(#1,#2,#3;#4)}

\newcommand{\allruns}{\Omega}
\newcommand{\sigmafield}{\mathcal{F}}

\newcommand{\rateRew}{\mathcal{R}}
\newcommand{\impRew}{\mathcal{I}}
\newcommand{\impRewExp}{\mathcal{I}_\prob}
\newcommand{\impRewFix}{\mathcal{I}_\trans}
\newcommand{\costFdC}{\mathit{Cost}}

\newcommand{\expectedRewExp}{\mathcal{J}_\gen}
\newcommand{\expectedRewFix}{\mathcal{J}_\trans}

\newcommand{\dmin}{\underline{d}}
\newcommand{\dmax}{\overline{d}}

\newcommand{\sink}{K}
\newcommand{\sinkValue}{V}

\renewcommand{\psi}{\text{Pois}}

\newcommand{\algbelow}{\vspace{-0em}}
\newcommand{\algabove}{\vspace{-0em}}

\newcommand{\computedOutcome}{x}
\newcommand{\computedValue}{x}

\newcommand{\minPst}{minP}
\newcommand{\maxRew}{maxR}
\newcommand{\minRew}{minR}

\newcommand{\first}[2]{\Diamond_{#1}^{#2}}

\newcommand{\Value}[1]{\mathit{Val}\left[#1\right]}
\newcommand{\CostBound}[1]{\mathit{Bound}\left[#1\right]}

\newcommand{\contMdp}{\mathcal{M}}

\newcommand{\discMdp}{\contMdp_\eps}

\newcommand{\mactsDisc}{\macts_\eps}

\newcommand{\guess}{\timeouts}
\newcommand{\guessMdp}{\contMdp_\guess}
\newcommand{\mtranGuess}{\mtran_\guess}

\newcommand{\mcostGuess}{\mcost_\guess}

\newcommand{\matchmdp}{\mdp^\star}

\newcommand{\matchtimeouts}{\timeouts^\star}

\newcommand{\otherError}{\gamma}

\newcommand{\brambory}{S'}

\newcommand{\todoc}[2]{
}

\newcommand{\mverts}{V}
\newcommand{\mlocs}{\mverts}
\newcommand{\vtx}{v}

\newcommand{\macts}{\mathit{Act}}
\newcommand{\mtran}{T}

\newcommand{\mcost}{\text{\euro}} 
\newcommand{\msteps}{\#} 
\newcommand{\mcostd}{A}

\newcommand{\mdp}{\mathcal{M}}
\newcommand{\hist}{w}

\newcommand{\Runs}{\mathit{Run}}

\newcommand{\run}{\rho}
\newcommand{\initvertex}{{\vtx_{in}}}

\newcommand{\ctrun}{\omega}

\newcommand{\entrysta}{\sta}

\newcommand{\revisit}[1]{\#^\goalStates_{#1}}

\newcommand{\discconst}{D_1}
\newcommand{\cutconst}{D_2}
\newcommand{\maxValue}[1]{\overline{\mathit{Val}}\left[#1\right]}
\newcommand{\maxvalconst}{M}
\newcommand{\POconst}{B}

\begin{document}

\title{
Optimizing Performance of Continuous-Time Stochastic Systems using Timeout Synthesis\thanks{
The research leading to these results has received funding from the People Programme (Marie Curie Actions) of the European Union's Seventh Framework Programme (FP7/2007-2013) under REA grant agreement $\text{n}^\circ$ [291734].
This work is partly supported by the German Research Council (DFG) as part of the Transregional Collaborative Research Center AVACS (SFB/TR 14), by the EU 7th Framework Programme under grant agreement no. 295261 (MEALS) and 318490 (SENSATION), by the Czech Science Foundation, grant No.~15-17564S, and by the CAS/SAFEA International Partnership Program for Creative Research Teams.
	}}

\author{
Tom\'a\v{s} Br\'azdil\inst{1}
\and \v{L}ubo\v{s} Koren\v{c}iak\inst{1}
\and Jan Kr\v{c}\'al\inst{2}
\and Petr Novotn\'{y}\inst{3}
\and Vojt\v{e}ch~{\v{R}}eh\'ak\inst{1}}

\institute{
	Faculty of Informatics, Masaryk University, Brno, Czech Republic
		\texttt{\{brazdil,\,korenciak,\,\,rehak\}\!@fi.muni.cz}
	\and 
	Saarland University -- Computer Science, Saarbr\"ucken, Germany \\
		\texttt{krcal@cs.uni-saarland.de}
	\and
	IST Austria, Klosterneuburg, Austria \\
		\texttt{petr.novotny@ist.ac.at}
}

\maketitle

\begin{abstract}
We consider parametric version of fixed-delay continuous-time Markov chains (or equivalently deterministic and stochastic Petri nets, DSPN)
where fixed-delay transitions are specified by parameters, rather than concrete values. Our goal is to synthesize values of these parameters that, for a given cost function, minimize 
expected total cost incurred before reaching a given set of target states.
We show that under mild assumptions, optimal values of parameters can be effectively approximated using translation to a Markov decision process (MDP) whose actions correspond to discretized values of these parameters. 
To this end we identify and overcome several interesting phenomena arising in systems with fixed delays.

\end{abstract}

\section{Introduction}
\label{sec-intro}

Continuous-time Markov chains (CTMC) are a fundamental model of stochastic
systems with discrete state-spaces that evolve in continuous-time. Several
higher level modelling formalisms, such as stochastic Petri nets and
stochastic process algebras, use CTMC as their semantics.  As such, CTMC
have been applied in performance and dependability analysis in various
contexts ranging from aircraft communication protocols (see,
e.g.~\cite{Tiassou-phd}) to models of biochemical systems (see,
e.g.~\cite{CTMC-biochem}).

There are several equivalent definitions of CTMC
(see,~e.g.~\cite{Feller:book-I,Norris:book}). We may 
define
a~(uniformized, finite-state) CTMC to consist of a finite set of states $S$
coupled with a common rate $\lambda$ and a stochastic matrix $\prob\in
\Rsetpo^{S\times S}$ specifying probabilities of transitions between 
states. An~execution starts in a given initial state. In every step, the CTMC waits
for a duration that is selected randomly according to the exponential
distribution with the rate
$\lambda$, 
and then moves to a state $s'$ randomly chosen with probability
$\prob(s,s')$.

The practical interpretation of the above semantics is that in every state
the system waits for an~event to occur and then reacts by changing its
state.  A typical example is a model of a simple queue to which new
customers come in random intervals and are also served in random intervals.
However, in practice, events are usually not exponentially distributed, and,
in fact, their distributions may be quite far from being exponential. To
deal with such events, phase-type approximation technique~\cite{Neuts:book}
is usually applied. Unfortunately, as already noted in~\cite{Neuts:book},
some distributions cannot be efficiently fit with phase-type
approximation. In particular, degenerate distributions of events with fixed
delays, i.e., events that occur after a fixed amount of time with
probability 1, form a distinguished example of this phenomenon (for more
details see~\cite{nas-EPEW}).  However, as events with fixed delays play a
crucial role in many systems, especially in communication
protocols~\cite{TTC-book}, time-driven real-time
scheduling~\cite{scheduling-in-real-time}, etc., they should be handled
faithfully in modelling and analysis.

Inspired by deterministic and stochastic Petri nets~\cite{marsan1987petri}
and delayed CTMC~\cite{guet2012delayed} with at most one non-exponential 
transition enabled in any time, we study fixed-delay CTMC (fdCTMC),
the CTMC extended with {\em fixed-delay transitions}. More concretely, we
specify a set of states $\allact\subseteq S$ where fixed-delay transitions
are enabled and add a stochastic matrix $\trans\in \Rsetpo^{\allact\times
  S}$ specifying probabilities of fixed-delay transitions between states. In
addition, we consider a {\em delay function} $\timeouts:\allact\rightarrow
\Rsetp$.
The semantics can be intuitively described as follows. Imagine a~CTMC
extended with an alarm clock. At the beginning of an execution, the alarm
clock is turned off and the process behaves as the original CTMC. Whenever a
state $s$ of $\allact$ is visited and the alarm clock is off at the time, it
is turned on and set to ring after $\timeouts(s)$ time units. Subsequently,
the process keeps behaving as the original CTMC until either a state of
$S\smallsetminus \allact$ is visited (in which case the alarm clock is
turned off), or the alarm clock rings in a state $s'$ of $\allact$. In the
latter case, a fixed-delay transition takes place, which means that the
process changes the state randomly according to the distribution
$\trans(s',\cdot)$, and the alarm clock is either turned off or newly set
(when entering a state of $\allact$).

In most practical applications mentioned above, fixed-delay transitions are
determined by the design of the system and often strongly influence
performance of the system. Indeed, both timeouts in network protocols as
well as scheduling intervals in real-time systems directly influence
performance of the respective systems and their manual setting usually
requires considerable effort and expertise.  This motivates us to consider
the fixed-time delays $\timeouts(s)$ as {\em free parameters} of the model,
and develop techniques for their optimization with respect to a given
performance measure.

\begin{example} \label{ex:protocols}
We demonstrate the concept on two different models of sending \emph{one} segment of data in the \emph{alternating bit protocol}. 
In the protocol, each segment of data is retransmitted until an acknowledgement is received. 
The delay between retransmissions has impact on throughput of the protocol as well as on network congestion.
In the simpler model below on the left, the data is sent in state
$init$. The exp-delay transitions, drawn as solid arrows, model message loss
(with probability $0.2$) and delivery (with probability $0.8$). 
For simplicity we use rate 1 and omit self loops of exponential transitions in all examples. The
fixed-delay transitions, drawn as dashed arrows, cause the data to be
retransmitted. Note that whenever the data is retransmitted, the previous
message with the data is canceled in this model.

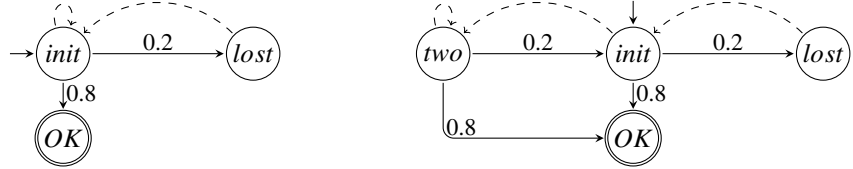
\myspacesmall
  \begin{center}
  \begin{tikzpicture}[outer sep=0.1em, xscale=1, yscale=1]

      \tikzstyle{fixed}=[dashed,->]; \tikzstyle{fixed label}=[font=\small];
      \tikzstyle{exp}=[->,rounded corners,,>=stealth]; \tikzstyle{exp rate}=[font=\small];
      \tikzstyle{loc}=[draw,circle, minimum size=2em,inner sep=0.1em];
      \tikzstyle{accepting}+=[outer sep=0.1em]; \tikzstyle{loc
        cost}=[draw,rectangle,inner sep=0.07em,above=6, minimum
      width=0.8em,minimum height=0.8em,fill=white,font=\footnotesize];
      \tikzstyle{trans cost}=[draw,rectangle,minimum width=0.8em,minimum
      height=0.8em,solid,inner sep=0.07em,fill=white,font=\footnotesize];
      \tikzstyle{prob}=[inner sep=0.03em, auto,font=\footnotesize];

\begin{scope}[yshift=-6cm]

      \node[loc] (s) at (0,0) {${init}$}; 
      \node[loc] (u) at (2.5,0) {${lost}$}; 
      \node[loc, accepting] (t) at (0,-1.12) {${OK}$};

      \path[->,>=stealth] ($(s)+(-0.7,0)$) edge (s);

\path[exp] (s) edge node[prob] {0.2} 
(u); 

\path[exp] (s) edge node[prob,pos=0.4] {0.8}
(t);

\path[bend right=45,fixed] (u) edge 
(s);
\path[loop above,fixed,looseness=5] (s) edge 
(s);

\end{scope}
\begin{scope}[yshift=-6cm,xshift=7.5cm]

      \node[loc] (i2) at (0,0) {${init}$}; 
      \node[loc] (l2) at (2.5,0) {${lost}$}; 
      \node[loc] (c2) at (-2.5,0) {${two}$}; 
      \node[loc, accepting] (t2) at (0,-1.12) {${OK}$};

      \path[->,>=stealth] ($(i2)+(-0,.7)$) edge (i2);

\path[exp] (i2) edge node[prob] {0.2} 
(l2); 

\path[exp] (i2) edge node[prob,pos=0.4] {0.8}
(t2);
\path[exp] (c2) edge node[prob] {0.2} 
(i2);

\draw[exp] (c2) |- node[prob,pos=0.4] {0.8}
(t2);

\path[bend right=45,fixed] (l2) edge 
(i2);
\path[bend right=45,fixed] (i2) edge 
(c2);
\path[loop above,fixed,looseness=5] (c2) edge 
(c2);

\end{scope}

\end{tikzpicture}
  \end{center} 	
\myspacesmall

\noindent
The more faithful model on the right models up to two messages with the data
segment being delivered concurrently. 
For choosing an optimal delay between retransmissions, we need to formalize how to express performance of the protocol.
\end{example}

To express performance properties, we use standard cost (or reward)
structures (see,~e.g.~\cite{Puterman:book}) that assign numerical rewards to
states and transitions. More precisely, we consider the following three cost
functions: $\rateRew: \states \to \Rsetpo$, which assigns a cost rate
$\rateRew(s)$ to every state $s$ so that the cost $\rateRew(s)$ is paid for
every unit of time spent in the state $s$, and functions $\impRewExp,
\impRewFix: \states \times \states \to \Rsetpo$ that assign to each
exp-delay and fixed-delay transition, respectively, the cost that is
immediately paid when the transition is taken. 
Note that $\rateRew$ is
usually used to express time spent in individual states, while the other two
cost functions are used to quantify the difficulty of dealing with events
corresponding to transitions.  The performance measure itself is the {\em
  expected total cost incurred before reaching a given set of states $G$
  starting in a given initial state $\initstate$}. For this moment, let us
denote this measure by $\expred_{\timeouts}$, stressing the fact that it depends
on the delay function $\timeouts$ which is the only variable quantity in our
optimization task:
\begin{problem}[Cost optimization]\label{prob:optim}
For a subspace of delay functions $\paramspace \subseteq (\Rsetp)^{\allact}$ and a given approximation error $\eps > 0$, compute a delay function $\timeouts \in \paramspace$ that is \emph{$\eps$-optimal within $\paramspace$}, i.e.
$$
\left\lvert \;
\inf_{\timeouts'\in\paramspace} \expred_{\timeouts'} - \expred_{\timeouts}
\; \right\rvert
\;\; < \;\; \eps.$$
\end{problem}

\myspacebig
\paragraph{Example~\ref{ex:protocols} (cont.)}
We can model the expected cost of sending one data segment in our examples as follows.
To take into account the expected time of data delivery, we set the cost rate of each state to, e.g., $1$.
To take into account the expected number of retransmissions, we set the cost of each fixed-delay transition, e.g., to $3$. The cost of each exp-delay transition is set to $0$.
Now the goal for the left model is to find a delay $\timeouts(init)$
optimizing the expected total cost incurred before reaching the state
$OK$. Note that $\timeouts$ is never set in the state $lost$.
The goal is the same for the model on the right where $\timeouts$ is set also
in the state $two$. Note that it makes no sense to synthesize different
delays $\timeouts(init)$ and $\timeouts(two)$ as the states $init$ and $two$
are indistinguishable in the implementation of the protocol. Therefore, we
need to require that the synthesised delay function satisfies
$\timeouts(init) = \timeouts(two)$.
\vspace{0.2em}

\paragraph{Our contribution:} We consider fixed-delay CTMC as a natural extension of CTMC suitable for algorithmic synthesis of fixed timeouts. Upon this model, we investigate algorithmic complexity of the cost optimization problem. This is, to the best of our knowledge, the most general attempt at fully automatic synthesis of timeouts in continuous-time stochastic systems. 
We provide algorithms for solving the following two special cases of the cost optimization problem  under the assumption that the reward rate $\rateRew(\sta)$ is {\em positive} in every state $s$:
\begin{enumerate}
\item \textbf{Unconstrained optimization} where we demand $\paramspace =
  (\Rsetp)^{\allact}$, i.e. the set of all delay functions. We solve this problem by reduction to a finite Markov decision process (MDP) whose actions correspond to {\em discretized} (i.e. rounded onto a~finite mesh) values of delays in the individual states, and then apply standard polynomial time algorithms for synthesis of the delays (note that a brute force search through a "discretized" subset of $D$ would be exponentially worse).
   The most non-trivial part is to prove that the delays may be
   discretized. We show that a na\"ive rounding of a near-optimal delay
   function may cause arbitrarily high {\em absolute} error. Our solution,
   based on rather non-trivial insights into the structure of fdCTMCs, avoids
   this obstacle by identifying "safe" delay functions that may be rounded
   with an error bounded (exponentially) in the size of the system. This
   leads to an exponential time algorithm for solving the cost optimization
   problem.
  
\item \textbf{Bounded optimization under partial observation} where we introduce bounds $\dmin,\dmax>0$ together with an equivalence relation $\equiv$ on $\allact$ and demand $D$ to be the set of all delay functions $\timeouts$ satisfying the following conditions:
\begin{itemize}
\item $\dmin \leq \timeouts(s) \leq \dmax$ for all $s\in \allact$,
\item $\timeouts(s)=\timeouts(s')$ whenever $s\equiv s'$.
\end{itemize}
Like in the Example~\ref{ex:protocols}, the equivalence $\equiv$ can be
used to hide information about detailed internal structure of states which
is often needed in practical applications. 
In this paper, we show that
the bounded optimization under partial observation can be solved in time doubly exponential in $\dmax$ and exponential in all other parameters.

We also consider the
corresponding approximate threshold variant: For a~given $\computedValue$ decide whether
$\inf_{\timeouts'\in\paramspace} \expred_{\timeouts'}>\computedValue+\varepsilon$, or
$\inf_{\timeouts'\in\paramspace} \expred_{\timeouts'}<\computedValue-\varepsilon$ (for
$\inf_{\timeouts'\in\paramspace} \expred_{\timeouts'}\in
[\computedValue-\varepsilon,\computedValue+\varepsilon]$ an arbitrary answer may be given). We show
that this bounded optimization problem is NP-hard, thus a polynomial time solution of the bounded optimization under partial observation is unlikely.

The assumption that all delays are between fixed thresholds $\dmin$ and
$\dmax$ is crucial in our approach. As we discuss in
Section~\ref{sec:results-multi}, without this assumption the optimization
under partial observation becomes much trickier and we leave its solution
for future work.
\end{enumerate}

\myspacebig
\paragraph{Related work.}

Various forms of continuous-time stochastic processes with fixed-delay transitions have already been studied, see e.g. \cite{marsan1987petri,DBLP:dblp_conf/apn/ChoiKT93,ACD:model-checking-real-time,DBLP:conf/concur/BrazdilKKR11,Chirstel_TA_CTMC}. In particular, as noted above, our definition of fdCTMC is closely related to the original definition of deterministic and stochastic Petri nets~\cite{marsan1987petri}. 
Papers on verification of continuous-time systems with timed automata (TA) specifications~\cite{CHKM:CTMC-quant-timed,Chirstel_TA_CTMC,BBBBG:One-clock-almost-sure} are also related to our work as the constraints in timed automata resemble fixed-delay transitions.
None of these works, however, considers synthesis of fixed-delays (or other parameters).

Parameter synthesis techniques have been developed for several models, such as parametric timed automata~\cite{Alur-TA-params}, parametric one-counter automata~\cite{hkow-concur09}, parametric Markov models~\cite{DBLP:journals/sttt/HahnHZ11}, etc.
In continuous-time stochastic systems, \cite{Katoen-param-synt,CTMC-biochem} study synthesis of rates in CTMC which is a problem orthogonal to timeouts.
Furthermore, optimal control of continuous-time (Semi)-Markov decision processes~\cite{TBR-in_CTMDP-Neuhausser,Prob_reach__param_DTMC,BKKKR:GSMP-games-TA,DBLP:journals/iandc/BrazdilFKKK13}
can be viewed as synthesis of \emph{discrete} parameters in continuous-time systems. 

The problem of synthesizing timeouts as \emph{continuous}  parameters has been studied
in variety of engineering contexts such as vehicle communication systems 
\cite{cars-timeout-synthesis} and avionic subsystems
\cite{ARINC_629-analysis,Tiassou-phd}. 
To the best of our knowledge, no generic framework for synthesis of timeouts in stochastic continuous-time systems has been developed so far. In theoretical literature, only simpler cases have been addressed. For instance
\cite{CRV:real-time-testing-synthesis,WDQ:real-time-test-execution} consider a finite test case, a sequence of input and output actions, and synthesize times for input actions that maximize the probability of executing this acyclic sequence. Allowing cycles in fdCTMC makes the timeout synthesis problem much more demanding, e.g., due to potentially unbounded number of stochastic events between timeouts.
Instead of static timeouts, \cite{KNSS:PTA-scheduler-synthesis,JT:PTMDP-learning} consider synthesis of ``dynamic'' timeouts where the delay is chosen based on the history of the execution so far. Consequently, the delay can be changed \emph{while it is elapsing} whenever stochastic events occur. This makes it much simpler to solve and also adequate for a different application domain.

Section~\ref{sec-prelims} introduces fixed-delay CTMC and cost structures. Section~\ref{sec:results-single} and Section~\ref{sec:results-multi} are devoted to unconstrained optimization and bounded optimization under partial observation, respectively. Due to space constraints, full proofs are in~\cite{BKKNR:new-arxiv}.

\section{Preliminaries}
\label{sec-prelims}

We use $\Nseto$, $\Rsetpo$, and $\Rsetp$ to denote the set of all
non-negative integers, non-negative real numbers, and positive real numbers,
respectively. Furthermore, for a countable set $A$, we denote by $\dist(A)$ the set of discrete probability distributions over $A$, i.e. functions $\mu: A \to \Rsetpo$ such that $\sum_{a\in A} \mu(a) = 1$. Encoding size of an object $O$ is denoted by $\size{O}$.

\begin{definition}
  A \emph{fixed-delay CTMC structure} (fdCTMC structure) $\fdC$ is a tuple
  $(\states, \lambda, \prob, \allact, \trans, \initstate)$ where
  \begin{itemize}
  \item $\states$ is a finite set of states,
  \item $\lambda \in \Rsetp$ is a (common) rate of exp-delay transitions,
  \item $\prob: \states \times \states \to \Rsetpo$ is a stochastic matrix specifying probabilities of exp-delay transitions,
  \item $\allact \subseteq \states$ is a set of states where fixed-delay transitions are enabled,
  \item $\trans: \allact \times \states \to \Rsetpo$ is a stochastic matrix specifying probabilities of fixed-delay transitions, and
  \item $\initstate \in \states$ is an initial state.
  \end{itemize}
A {\em fixed-delay CTMC} (fdCTMC) is a pair
$\fdC(\timeouts)=(\fdC,\timeouts)$ where $\fdC$ is a fdCTMC structure and
$\timeouts: \allact \to \Rsetp$ is a {\em delay function} which to every
state where fixed-delay transitions are enabled  assigns a~positive
delay.
\end{definition}

A \emph{configuration} of a fdCTMC is a pair
$(\sta,\delays)$ where $\sta\in\states$ is the current state and $\delays \in \Rsetp \cup \{\nodel\}$ is the remaining time before a fixed-delay transition takes place. We assume that $d=\infty$ iff  $\sta\not\in \allact$. To simplify notation, we similarly extend any delay function $\timeouts$ to all states $\states$ by assuming $\timeouts(\sta) = \infty$ iff $\sta \not\in\allact$.

An execution of $\fdC(\timeouts)$ starts in the configuration $(\sta_0,\delays_0)$ with $\sta_0 = \initstate$ and $\delays_0 = \timeouts(\initstate)$. In every step, assuming that the current configuration is $(\sta_i,\delays_i)$, the fdCTMC waits for some time $t_i$ and then moves to a~next configuration $(s_{i+1},d_{i+1})$ determined as follows:
\begin{itemize}
\item First, a waiting time $t_{exp}$ for exp-delay transitions from $s_i$ is chosen randomly according to the exponential distribution with the rate $\lambda$.
\item Then 
	\begin{itemize}
	\item If $t_{exp}<d_i$, then an exp-delay transition occurs, which
          means that
          $t_i=t_{exp}$, $s_{i+1}$ is chosen randomly with probability $P(s_i,s_{i+1})$, and $d_{i+1}$ is determined by
			$$\delays_{i+1} = \begin{cases}
	\delays_i - t_{exp}& \text{if $\sta_{i+1} \in \allact$ and $\sta_i\in \allact$ \small{(previous delay remains)},} \\
	\timeouts(\sta_{i+1}) & \text{if $\sta_{i+1} \not\in \allact$ or $\sta_i \not\in\allact$ \small{(delay is newly set or disabled)}.} 
	\end{cases}
	$$
	\item If $t_{exp}\geq d_i$, then a fixed-delay transition occurs,
          which means that 
	$t_i=d_i$, $s_{i+1}$ is chosen randomly with probability
        $F(s_i,s_{i+1})$, and $d_{i+1}=\timeouts(\sta_{i+1})$.
	\end{itemize}
\end{itemize}

This way, the execution of a fdCTMC forms a \emph{run},
an alternating sequence of configurations and times $(\sta_0,\delays_0)t_0 (\sta_1,\delays_1) t_1 \cdots$. The probability measure $\probm_{\fdC(\timeouts)}$ over runs of $\fdC(\timeouts)$ is formally defined in~\appref{subsec:prob-spaceFdCTMC}.

\paragraph{Total cost before reaching a goal} 
To allow formalization of performance properties, we enrich the model in a standard way (see,~e.g.~\cite{Puterman:book}) with costs (or rewards).
A \emph{cost structure} over a fdCTMC structure $\fdC$ with state space $\states$ is a
tuple $\costFdC = (\goalStates, \rateRew, \impRewExp, \impRewFix)$ where $\goalStates\subseteq S$ is a set of goal states, $\rateRew:
\states \to \Rsetpo$ assigns a cost rate to every state, and $\impRewExp,
\impRewFix: \states \times \states \to \Rsetpo$ assign an impulse cost to
every exp-delay and fixed-delay transition, respectively.
Slightly abusing the notation, we denote by $\costFdC$ also the random variable assigning to each run $\ctrun = (\sta_0,\delays_0) t_0 \cdots$ the \emph{total cost before reaching $\goalStates$} (in at least one transition), given by
$$
\costFdC(\ctrun)
= \begin{cases}
\sum_{i=0}^{n-1} \left(
t_i\cdot\rateRew(s_i) +
\impRew_i(\ctrun)
\right)
 & \text{for minimal $n>0$ such that $s_n\in\goalStates$,} \\
\infty & \text{if there is no such $n$,}
\end{cases}
$$ 
where $\impRew_i(\ctrun)$ equals $\impRewExp(s_i,s_{i+1})$ for an exp-delay transition, i.e. when $t_i < \delays_i$, and
equals $\impRewFix(s_i,s_{i+1})$
for a fixed-delay transition, i.e. when $t_i = \delays_i$.

We denote the expectation of $\costFdC$ with respect to $\probm_{\fdC(\timeouts)}$ simply by $\expred_{\fdC(\timeouts)}$, or by $\expred_{\fdC(\timeouts)}^\costFdC$ when $\costFdC$ is not clear from context. 
Our aim is to (approximatively) minimize the expected cost, i.e. to find a delay function $\timeouts$ such that $\expred_{\fdC(\timeouts)} \leq \Value{\fdC} + \eps$ where $\Value{\fdC}$ denotes the optimal cost $\inf_{\timeouts'} \expred_{\fdC(\timeouts')}$.

\paragraph{Non-parametric analysis}
Due to \cite{DBLP:journals/pe/Lindemann93}, we can easily analyze a fdCTMC where the delay function is \emph{fixed}. Hence, both the expected total cost before reaching $\goalStates$ and the reaching probabilities of states in $\goalStates$ can be efficiently approximated. 
\begin{proposition}\label{prop:compute-cost}
There is a polynomial-time algorithm that for a given fdCTMC
$\fdC(\timeouts)$, cost structure $\costFdC$ with goal states $\goalStates$,
and an approximation error $\eps > 0$ computes $\computedOutcome \in
\Rsetp \cup \{\infty\}$ and $p_s \in \Rsetp$, for each $s\in \goalStates$, such that 
$$\left\lvert \; \expred_{\fdC(\timeouts)} - \computedOutcome \;
\right\rvert  < \eps \quad \text{ and } \qquad
\left\lvert \probm_{\fdC(\timeouts)}(\first{G}{s}) - p_s \right\rvert < \eps$$
where $\first{G}{s}$ is the set of runs that reach $s$ as the first state of $\goalStates$ (after at least one transition).
\end{proposition}

\paragraph{Markov decision processes.} 
In Section~\ref{sec:results-single} we use a reduction of fdCTMC to discrete-time Markov decision processes (DTMDP,~see~e.g.~\cite{Puterman:book}) with uncountable space of actions.

\begin{definition}
	A \emph{DTMDP} is a tuple $\mdp = (\mlocs,\macts,\mtran, \initvertex, V')$, where $\mlocs$ is a finite set of \emph{vertices}, $\macts$ is a (possibly uncountable) set of \emph{actions}, $\mtran\colon \mlocs \times \macts \rightarrow \dist(\mlocs)\cup\{\bot\}$ is a \emph{transition function}, $\initvertex\in\mlocs$ is an initial vertex, and $V'\subseteq \mlocs$ is a set of {\em goal vertices}.
\end{definition}
\noindent
An action $a$ is \emph{enabled} in a vertex $\vtx$ if $\mtran(\vtx,a)\neq \bot$.
A~{\em strategy} is a function $\sigma\colon \mlocs \rightarrow \macts$
which assigns to every vertex $\vtx$ an action enabled in $\vtx$. The
behaviour of $\mdp$ with a fixed strategy $\sigma$ can be intuitively
described as follows: A run starts in the vertex $\initvertex$. In every step, assuming that the current vertex is $v$, the process moves to a new vertex $v'$ with probability $T(v,\sigma(v))(v')$.
Every strategy $\sigma$ uniquely determines a probability measure $\probm_{\mdp(\sigma)}$ on the set of \emph{runs}, i.e. infinite alternating sequences of vertices and actions $\vtx_0 a_1 \vtx_1 a_2 \vtx_2\cdots \in (\mlocs\cdot \macts)^{\omega}$; see~\appref{subsec:prob-spaceDTMDP} for details.

Analogously to fdCTMC, we can endow a DTMDP with a {\em cost function} $\mcostd\colon\mlocs\times\macts \rightarrow \Rsetpo$. We then define for each run $\vtx_0 a_1 \vtx_1 a_2 \dots$, the \emph{total cost incurred before reaching $V'$} as 
$\sum_{i=0}^{n-1} \mcostd(\vtx_{i},a_{i+1})$ if there is a minimal $n > 0$ such that $\vtx_n \in V'$, and as $\infty$ otherwise. 
The expectation of this cost w.r.t. $\probm_{\mdp(\sigma)}$ is similarly denoted by 
$\expred_{\mdp(\sigma)}$ or by $\expred_{\mdp(\sigma)}^\mcostd$ if the cost function is not clear from context.

Given $\varepsilon\geq 0$, we say that a strategy $\sigma$ is {\em $\varepsilon$-optimal}
in $\mdp$ 
if $\expred_{\mdp(\sigma)}\leq 
\Value{\mdp}
+\varepsilon$ where $\Value{\mdp}=\inf_{\sigma'} \expred_{\mdp(\sigma')}$; we call it {\em optimal} if it is $0$-optimal. 
For any $s\in S$, let us denote by $\mdp[s]$ the DTMDP obtained from $\mdp$ by replacing the initial state by $s$.
We call a strategy \emph{globally ($\eps$-)optimal} if it is ($\eps$-)optimal in $\mdp[s]$ for every $s\in S$.
Sometimes, we restrict to a subset $D$ of all strategies and denote by $\Value{\mdp,D}$ the restricted infimum $\inf_{\sigma' \in D} \expred_{\mdp(\sigma')}$.

\section{Unconstrained Optimization}
\label{sec:results-single}
\begin{theorem}\label{thm:unconstrained}
There is an algorithm that given a fdCTMC structure $\fdC$, a cost structure $\costFdC$ with $\rateRew(\sta) > 0$ for all $\sta\in\states$, and $\eps>0$ computes in exponential time a delay function $\timeouts$ with
$$
\left\lvert 
\expred_{\fdC(\timeouts)}
\ - \ 
\inf_{\timeouts'}\ 
\expred_{\fdC(\timeouts')}
\; \right\rvert
\;\; < \;\; \eps.$$
\end{theorem}
The rest of this section is devoted to a proof of
Theorem~\ref{thm:unconstrained}, which consists of two parts. First, we
reduce the optimization problem in the fdCTMC to an optimization problem in
a \emph{discrete time} Markov decision process (DTMDP) with uncountably many
actions. Second, we present the actual approximation algorithm based on a straightforward discretization of the space of actions of the DTMDP. However, the proof of its error bound is actually quite intricate.
The time complexity is exponential because the discretized DTMDP needs exponential size in the worst case. For complexity with respect to various parameters see ~\appref{subs:mesh-error}.

For the rest of this section we fix a fdCTMC structure $\fdC = (\states,
\lambda, \prob, \allact, \trans, \initstate)$, a cost structure $\costFdC = (\goalStates, \rateRew,
\impRewExp, \impRewFix)$, and $\eps > 0$. 
We assume that $\Value{\fdC} < \infty$. The opposite case can be easily detected
by fixing an arbitrary delay function $\timeouts$ and finding out whether $\expred_{\fdC(\timeouts)} = \infty$ by Proposition~\ref{prop:compute-cost}.
This is equivalent to $\Value{\fdC} = \infty$ by the following observation.
\begin{lemma}
\label{lem:finite}
	For any delay functions $\timeouts, \timeouts'$ we have $\expred_{\fdC(\timeouts)} = \infty$ if and only if $\expred_{\fdC(\timeouts')} = \infty$.
\end{lemma}

To further simplify our presentation, we assume that each state $s$ of $\allact$ directly encodes whether the delay needs to be reset upon entering $s$. 
Formally, we assume $\allact = S^{reset}\uplus S^{keep}$ where $s' \in S^{reset}$ if $P(s,s') > 0$ for some $s\in S\setminus \allact$ or if $F(s,s') > 0$ for some $s\in \allact$; and $s' \in S^{keep}$ if $P(s,s') > 0$ for some $s \in \allact$.
We furthermore assume that $\initstate \in S^{reset}$ if $\initstate \in \allact$.
Note that each fdCTMC structure can be easily transformed to satisfy this assumption in polynomial time by duplication of states $\allact$, see, e.g.,  Example~\ref{ex:mdp-creation}.

\subsection{Reduction to DTMDP $\contMdp$ with Uncountable Space of Actions.} \label{subsec:reduction-DTMDP}
We reduce the problem into a \emph{discrete-time} problem by capturing the evolution of the fdCTMC only at \emph{discrete moments} when a transition is taken after which the fixed-delay is (a) newly set, or (b) switched off, or (c) irrelevant as the goal set is reached. 
This happens exactly when one of the states of $\brambory = S^{reset}\cup \left(S\smallsetminus \allact\right)\cup G$ is reached.
We define a DTMDP $\mdp=(\brambory,\macts,\mtran,\initstate,\goalStates)$ with a cost function $\mcost$: 
\begin{itemize}
\item $\macts := \Rset_{>0}\cup\{\infty\}$; where actions $\Rset_{>0}$ are enabled in $\sta \in S^{reset}$ and action $\infty$ is enabled in $\sta \in \states\setminus\allact$.
\item Let $s\in \brambory$ and $d$ be an action of $\mdp$. Intuitively, we define $\mtran(s,d)$ and $\mcost(s,d)$ to summarize the behaviour of the fdCTMC starting in the configuration $(s,d)$ until the first moment when $\brambory$ is reached again. 

Formally, let $\fdC[s](d)$ denote a fdCTMC obtained from $\fdC$ by changing initial state to $s$ and fixing a delay function that assigns $d$ to $s$ (and arbitrary values elsewhere).
We define $\mcost(\entrysta,\delays)$ as the cost accumulated before reaching another state of $\brambory$ and $\mtran(\sta,\delays)(\sta')$ as the probability that $s'$ is the first such a reached state of $\brambory$. That is,
\[
\mcost(s,\delays)=\expred_{\fdC[s](d)}^{\costFdC[\brambory]}
\qquad \text{and} \qquad
\mtran(s,\delays)(s')=\probm_{\fdC[s](d)}(\first{\brambory}{s'})
\]
where $\costFdC[\brambory]$ is obtained from $\costFdC$ by changing the set of goal states to $\brambory$. Note that the definition is correct as it does not depend on the delay function apart from its value $d$ in the initial state $s$.
\end{itemize}

\begin{example} \label{ex:mdp-creation}
	Let us illustrate the construction on the fdCTMC from Section~\ref{sec-intro}. The model, depicted on the left is modified to satisfy the assumption $\allact = S^{reset}\uplus S^{keep}$: we duplicate the state $init$ into another state $one \in S^{keep}$; the states from $S^{reset}$ are then depicted in the top row. As in Section~\ref{sec-intro}, we assign cost rate $1$ to all states and impulse cost $3$ to every fixed-delay transition (and zero to exp-delay transitions).

	  
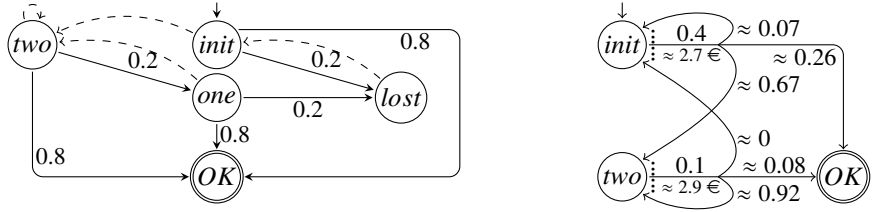
\begin{center}
	  	\begin{tikzpicture}[outer sep=0.1em, xscale=0.97, yscale=0.7]

	  	\tikzstyle{fixed}=[dashed,->]; \tikzstyle{fixed label}=[font=\small];
      	\tikzstyle{exp}=[->,rounded corners,,>=stealth]; \tikzstyle{exp rate}=[font=\small];
	  	\tikzstyle{loc}=[draw,circle, minimum size=1.8em,inner sep=0.1em];
	  	\tikzstyle{accepting}+=[outer sep=0.1em]; \tikzstyle{loc cost} = [draw,rectangle,inner sep=0.07em,above right=9, minimum
	  	width=0.8em,minimum height=0.8em,fill=white,font=\footnotesize];
	  	\tikzstyle{trans cost}=[draw,rectangle,minimum width=0.8em,minimum
	  	height=0.8em,solid,inner sep=0.07em,fill=white,font=\footnotesize,auto];
	  	\tikzstyle{prob}=[inner sep=0.03em, auto,font=\footnotesize];
	  	\tikzstyle{mcost}=[inner sep=0.03em, auto,font=\scriptsize];
	  	
	  	\begin{scope}[]
	  	
	  	\node[loc] (i2) at (0,0) {${init}$}; 
	  	\node[loc] (i3) at (0,-1) {${one}$}; 
	  	\node[loc] (l2) at (2.5,-1) {${lost}$}; 
	  	\node[loc] (c2) at (-2.5,0) {${two}$}; 
	  	\node[loc, accepting] (t2) at (0,-2.5) {${OK}$};

	  	\path[->,>=stealth] ($(i2)+(-0,.8)$) edge (i2);
	  	
	  	\path[exp] (i2) edge node[prob] {0.2} 
	  	(l2); 
	  	
	  	\draw[exp] (i2.35) -| node[prob,pos=0.4,swap] {0.8}
	  		+(3,-1) |-
	  	(t2);
	  	\path[exp] (c2) edge node[prob] {0.2} 
	  	(i3); 
	  	
	  	\path[exp] (i3) edge node[prob,pos=0.4] {0.8} 
	  	(t2);
	  	\path[exp] (i3) edge node[prob,swap] {0.2} 
	  	(l2);

	  	\draw[exp] (c2) |- node[prob,pos=0.4] {0.8}
	  	(t2);
	  	
	  	\path[bend right=30,fixed] (l2) edge 
	  	(i2);
	  	\path[bend right=35,fixed] (i2) edge 
	  	(c2);
	  	\path[bend right=30,fixed] (i3) edge 
	  	(c2);
	  	\path[loop above,fixed,looseness=5] (c2) edge 
	  	(c2);
	  	
	  	\end{scope}
	  	
      	\tikzstyle{exp}=[->,rounded corners]; \tikzstyle{exp rate}=[font=\small];
	  	\begin{scope}[xshift=8.5cm]
	  	
	  	\node[loc] (i) at (-3.0,0) {${init}$}; 
	  	\node[loc] (c) at (-3.0,-2.5) {${two}$}; 
	  	\node[loc, accepting] (t) at (0,-2.5) {${OK}$};
	  	
	  	\path[->] ($(i)+(-0,.8)$) edge (i);

	  	\draw[exp] (c) -- node[prob,pos=0.75] {$\approx 0.08$}
		  	 node[prob,pos=0.25] {$0.1$}
	  	     node[mcost,pos=0.25,swap] {$\approx 2.9$ \euro} 
	  	(t);
	  	
	  	\draw[exp] (i) -| node[prob,pos=0.4,below] {$\approx 0.26$}
	  		 node[prob,pos=0.11] {$0.4$}
	  		 node[mcost,pos=0.11,swap] {$\approx 2.7$ \euro} 
	  	(t);
	  	
	  	\path[exp] ($(i)+(1.3,0)$) edge[out=40,in=40,looseness=2.0] node[prob,pos=0.22,right] {$\approx 0.07$} 
	  	(i);
	  	\path[exp] ($(i)+(1.3,0)$) edge[out=320,in=40] node[prob,pos=0.35,right] {$\approx 0.67$} 
	  	(c);
	  	
	  	\path[exp] ($(c)+(1.3,0)$) edge[out=320,in=320,looseness=2.0] node[prob,pos=0.22,right] {$\approx 0.92$} 
	  	(c);
	  	\path[exp] ($(c)+(1.3,0)$) edge[out=40,in=320] node[prob,pos=0.35,right] {$\approx 0$} 
	  	(i);

	  	\node (ia) at ($(i)+(0.42,0)$) {$\cdot$};
	  	\node (ib) at ($(i)+(0.42,0.1)$) {$\cdot$};
	  	\node (ic) at ($(i)+(0.42,0.2)$) {$\cdot$};

	  	\node (id) at ($(i)+(0.42,-0.12)$) {$\cdot$};
	  	\node (ie) at ($(i)+(0.42,-0.22)$) {$\cdot$};
	  	\node (if) at ($(i)+(0.42,-0.32)$) {$\cdot$};

	  	\node (ca) at ($(c)+(0.42,0)$) {$\cdot$};
	  	\node (cb) at ($(c)+(0.42,0.1)$) {$\cdot$};
	  	\node (cc) at ($(c)+(0.42,0.2)$) {$\cdot$};

	  	\node (cd) at ($(c)+(0.42,-0.12)$) {$\cdot$};
	  	\node (ce) at ($(c)+(0.42,-0.22)$) {$\cdot$};
	  	\node (cf) at ($(c)+(0.42,-0.32)$) {$\cdot$};
	  	
	  	\end{scope}
	  	\end{tikzpicture}
	  \end{center} 
	  On the right, there is an excerpt of the DTMDP $\contMdp$ and of the cost function $\mcost$. For each non-goal state, we depict only one action out of uncountably many: for state $two$ it is action $0.1$ with cost $\approx 2.9$, for state $init$ it is action $0.4$ with cost $\approx 2.7$. The costs are computed in PRISM.
\qed
\end{example}

\noindent
Note that there is a one-to-one correspondence between the delay functions in $\fdC$ and strategies in $\mdp$.
Thus we use $\timeouts, \timeouts', \ldots$ to denote strategies in $\contMdp$.
Finally, let us state correctness of the reduction.

\begin{proposition}
\label{prop:mdp-formulation}
For any delay function $\timeouts$ it holds 
$\expred_{\fdC(\timeouts)}=\expred_{\contMdp(\timeouts)}$. Hence, 
$$\Value{\fdC} = \Value{\contMdp}.$$
\end{proposition}

\noindent
In particular, 
in order to solve the optimization problem for $\fdC$ it suffices to find an $\varepsilon$-optimal strategy (i.e., a delay function) $\timeouts$ in $\contMdp$.

\subsection{Discretization of the Uncountable MDP $\contMdp$}

Since the MDP $\contMdp$ has uncountably many actions, it is not directly suitable for algorithmic solutions. 
We proceed in two steps. In the first and technically demanding step, we show that we can restrict to actions on a finite mesh. Second, we argue that we can also approximate all transition probabilities and costs by rational numbers of small bit length.

\subsubsection{Restricting to a Finite Mesh.}
For positive reals $\delta > 0$ and $\dmax > 0$, we 
define a subset of delay functions 
$\paramspace(\delta,\dmax) = \{\timeouts \mid \forall s\in\brambory \; \exists k\in \Nset: \timeouts(s) = k\delta \leq \dmax\}$. Here, all delays are multiples of $\delta$ bounded by $\dmax$.

We need to argue that for the fixed $\eps>0$ there are some appropriate values $\delta$ and $\dmax$ such that $\paramspace(\delta,\dmax)$ contains an $\eps$-optimal delay function. A na\"{\i}ve approach would be to take any, say $\eps/2$-optimal 
delay function $\timeouts$, round it to closest delay function $\matchtimeouts \in \paramspace(\delta,\dmax)$ on the mesh, and show that the expected costs of these two functions do not differ by more than $\eps/2$. However, this approach does not work as shown by the following example.

\begin{example}
\label{ex:non-stability}
	Let us fix the fdCTMC structure $\fdC$ on the left (with cost rates in small boxes and zero impulse costs).
	An excerpt of $\contMdp$ and $\mcost$ is shown on the right (where only a few actions are depicted).
	
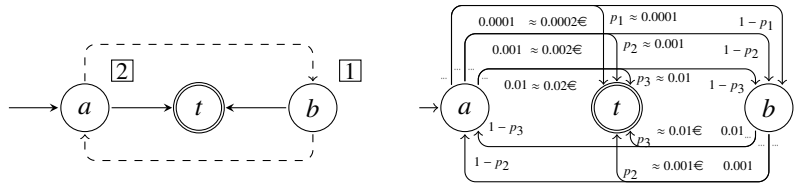
\begin{center}
		\begin{tikzpicture}[outer sep=0.1em, xscale=1, yscale=0.7]
		\tikzstyle{fixed}=[dashed,->];
		\tikzstyle{fixed label}=[font=\small];
		\tikzstyle{exp}=[->,>=stealth];
		\tikzstyle{exp rate}=[font=\small];
		\tikzstyle{loc}=[draw,circle, minimum size=1.8em,inner sep=0.1em];
		\tikzstyle{accepting}+=[outer sep=0.1em];
		\tikzstyle{loc cost} = [draw,rectangle,inner sep=0.07em,above right=9, minimum
		width=0.8em,minimum height=0.8em,fill=white,font=\footnotesize];
		\tikzstyle{action}=[font=\tiny];
		
		\begin{scope}
			\node[loc] (s) at (0,0) {$a$};
			\node[loc, accepting] (t) at (1.5,0) {$t$};
			\node[loc] (u) at (3,0) {$b$};
			
			\node[loc cost] at (s) {$2$};
			\node[loc cost] (cost) at (u) {$1$};
			
			\path[->,>=stealth] ($(s)+(-1,0)$) edge (s);
			
			\draw[fixed,rounded corners] (s) |- +(1,1.2) -| (u);
			\draw[fixed,rounded corners] (u) |- +(-1,-1) -| (s);
			
			\path[exp] (s) edge (t);
			\path[exp] (u) edge (t);
		\end{scope}

		\tikzstyle{exp}=[->];
		
		\begin{scope}[xshift=5cm]
			\node[loc] (s) at (0,0) {$a$};
			\node[loc, accepting] (t) at (2,0) {$t$};
			\node[loc] (u) at (4,0) {$b$};
			
			\path[->] ($(s)+(-.6,0)$) edge (s);
			
			\draw[exp,rounded corners] (s.60) |- 
				node[action,below,pos=0.75] {$0.01$} 
				node[action,below,pos=1] {$\approx 0.02 \text{\euro}$}
				+(1,0.3) -| 
					node[action,below,pos=.43] {$1-p_3$}
			(u.120);
			\draw[exp,rounded corners] (s.90) |- 
				node[action,below,pos=0.78] {$0.001$} 
				node[action,below,pos=1.1] {$\approx 0.002 \text{\euro}$}
			+(1,0.9) -| 
				node[action,below,pos=.44] {$1-p_2$}
			(u.90);
			\draw[exp,rounded corners] (s.120) |- 
				node[action,below,pos=0.8] {$0.0001$} 
				node[action,below,pos=1.2] {$\approx 0.0002 \text{\euro}$}
			+(1,1.5) -| 
				node[action,below,pos=.45] {$1-p_1$}
			(u.60);
			
			\draw[exp,rounded corners] (s.120) |- 
			+(1,1.5) -| 
				node[action,right=-2pt,pos=0.58] {$p_1 \approx 0.0001$}
			(t.120);
			\draw[exp,rounded corners] (s.90) |- 
			+(1,0.9) -| 
				node[action,right=-2pt,pos=0.63] {$p_2 \approx 0.001$}
			(t.90);
			\draw[exp,rounded corners] (s.60) |- 
			+(1,0.3) -| 
				node[action,right=-2pt,pos=0.85] {$p_3 \approx 0.01$}
			(t.60);

			\foreach \x in {40,75,105,140}{
				\node[above,font=\tiny,xscale=0.4] at (s.\x) {$\dots$};
			}
			\foreach \x in {40,75,105}{
				\node[below,font=\tiny,xscale=0.4] at (u.180+\x) {$\dots$};
			}
			
			\draw[exp,rounded corners] (u.240) |- 
				node[action,above,pos=0.65] {$0.01$} 
				node[action,above,pos=1] {$\approx 0.01 \text{\euro}$}
			+(-1,-0.3) -| 
				node[action,above,pos=.43] {$1-p_3$}
			(s.300);
			
			\draw[exp,rounded corners] (u.270) |- 
				node[action,above,pos=0.7] {$0.001$} 
				node[action,above,pos=1.1] {$\approx 0.001 \text{\euro}$}
			+(-1,-0.9) -| 
				node[action,above,pos=.44] {$1-p_2$}
			(s.270);
			
			\draw[exp,rounded corners] (u.270) |- 
			+(-1,-0.9) -| 
				node[action,right=-2pt,pos=0.58] {$p_2$}
			(t.270);
			
			\draw[exp,rounded corners] (u.240) |- 
			+(-1,-0.3) -| 
				node[action,right=-2pt,pos=0.63] {$p_3$}
			(t.300);
			
		\end{scope}
		\end{tikzpicture}
	\end{center}
	First, we point out that $\Value{\fdC} = 1$ as one can make sure that nearly all time before reaching $t$ is spent in the state $b$ that has a lower cost rate $1$.
	Indeed, this is achieved by setting a very short delay in $a$ and a long delay in $b$.

We claim that for any $\delta > 0$ there is a near-optimal delay function $\timeouts$ such that rounding its components to the nearest integer multiples of $\delta$ yields a large error independent of $\delta$. Indeed, it suffices to take a function $\timeouts$ with $\timeouts(b)=\delta$ and $\timeouts(a)$ an arbitrary number significantly smaller than $\timeouts(b)$, say $\timeouts(a)=0.01\cdot \timeouts(b)$. The error produced by the rounding can then be close to 0.5. For instance, given $\delta=0.01$ we take a function with $\timeouts(a)=0.0001$ and $\timeouts(b)=0.01$, whose rounding to the closest delay function on the finite mesh yields a constant function $\matchtimeouts=(0.01,0.01)$. Then $\expred_{\fdC(\timeouts)} \approx 1.01$ and $\expred_{\fdC(\matchtimeouts)} \approx 1.5$, even though the rounding does not change any transition probability or cost by more than $0.02$! 

	The reason why the delay function $\timeouts$ is so sensitive to small perturbations is that it makes a very large number of steps before reaching $t$ (around $200$ on average) and thus the small one-step errors caused by a perturbation accumulate into a large global error. 
	The number of steps of an $\eps$-optimal delay functions is not bounded, in general. By multiplying both $\timeouts(a)$ and $\timeouts(b)$ by the same infinitesimally small factors we obtain an $\eps$-optimal delay functions that make an arbitrarily high expected number of steps before reaching $t$. \qed
\end{example}

A crucial observation is that we do not have to show that the ``na\"{\i}ve'' rounding works for \emph{every} near-optimal delay function. To prove that $\paramspace(\delta,\dmax)$ contains an $\eps$-optimal function, it suffices to show that there is \emph{some} $\eps/2$-optimal function whose rounding yields error at most~$\eps/2$. Proving the existence of such well-behaved functions forms the technical core of our discretization process which is formalized below.

We start by formalizing the concept of ``small perturbations''.
We say that a delay function $\matchtimeouts$ is \emph{$\alpha$-bounded} by a delay function $\timeouts$
if for all states $s,t \in \brambory$ we have:
\vspace{-0.5\baselineskip}
\begin{enumerate} 
	\item  $|\mtran(s,\matchtimeouts(s))(t)-\mtran(s,\timeouts(s))(t)|\leq \alpha$ and
	\item $\mcost(s,\matchtimeouts(s))-\mcost(s,\timeouts(s))\leq \alpha$;
\end{enumerate}
and furthermore, $\mtran(s,\timeouts(s))(t)=0$ iff $\mtran(s,\matchtimeouts(s))(t)=0$, i.e. the qualitative transition structure is preserved. (Note that $\matchtimeouts$ may incur much smaller one-step costs than $\timeouts$, but not significantly higher).

Using standard techniques of numerical analysis, we express the increase in accumulated cost caused by a bounded perturbation as a function of the worst-case (among all possible initial states) expected cost and expected number of steps before reaching the target. The number of steps is essential as discussed in Example~\ref{ex:non-stability} and can be easily expressed by a cost function $\#$ that assigns $1$ to every action in every state. To express the worst-case expectation of some cost function $\$$, we denote $\CostBound{\$,\timeouts} := \max_{s\in\brambory} \expred_{\contMdp[s](\timeouts)}^{\$}$.

\begin{lemma}
	\label{lem:perturbation-error}
	Let $\alpha\in[0,1]$ and let, $\timeouts'$ be a delay function that is $\alpha$-bounded by another delay function $\timeouts$. If $\alpha \leq \frac{1}{2\cdot\CostBound{\mcost,\timeouts}\cdot|\brambory|}$, then
$$\expred_{\contMdp(\timeouts')}   
\; \leq \; 
\expred_{\contMdp(\timeouts)} + 2\cdot \alpha\cdot \CostBound{\msteps,\timeouts} \cdot (1 + \CostBound{\mcost,\timeouts}\cdot |\brambory|).$$ 
\end{lemma}

 The next lemma shows how to set the parameters $\delta$ and $\dmax$ to make the finite mesh  $\paramspace(\delta,\dmax)$ ``dense'' enough, i.e. to ensure that for any $\timeouts$, $\paramspace(\delta,\dmax)$ contains a delay function that is $\alpha$-bounded by $\timeouts$.

\begin{lemma}
\label{lem:mesh-error}
There are positive numbers $\discconst,\cutconst\in \exp(\size{\fdC}^{\mathcal{O}(1)})$ computable in time polynomial in $\size{\fdC}$ such that the following holds for any $\alpha\in[0,1]$ and any delay function $\timeouts$: If we put $$\delta \; := \; \alpha/\discconst \quad \text{and} \quad \dmax \; := \; |\log(\alpha)|\cdot \cutconst \cdot\CostBound{\mcost,\timeouts},$$ then $\paramspace(\delta,\dmax)$ contains a delay function which is $\alpha$-bounded by $\timeouts$.
\end{lemma}
\begin{proof}[Sketch]
Computing the value of $\delta$ is easy as the derivatives of the probabilities and costs are bounded from above by the rate $\lambda$ and the maximal cost rate, respectively.
For $\dmax$ we need additional technical observations,
 see~\cite{BKKNR:new-arxiv} for further details.
\end{proof}

Unfortunately, as shown in Example~\ref{ex:non-stability}, the value $\CostBound{\msteps,\timeouts}$ can be arbitrarily high, even for near-optimal functions $\timeouts$. 
Hence, we cannot use Lemma~\ref{lem:perturbation-error} right away to show that a delay function in $\paramspace(\delta,\dmax)$ that is $\alpha$-bounded by some near-optimal $\timeouts$ is also near-optimal.
The crucial insight is that for any $\eps' > 0$ there are (globally) $\eps'$-optimal delay functions that use number of steps that is \emph{proportional} to their expected cost. 

\begin{lemma}
	\label{lem:bounded-step-existence}
There is a positive number $\constFactor\in \exp(\size{\fdC}^{\mathcal{O}(1)})$ computable in time polynomial in $\size{\fdC}$ such that the following holds:
	for any $\eps'>0$, there is a globally $\eps'/2$-optimal delay function $\timeouts'$ with 
	\begin{equation}
	\label{eq:step-bound}
	\CostBound{\msteps,\timeouts'} \; \leq \; \frac{\CostBound{\mcost,\timeouts'}}{\eps'}\cdot \constFactor.
	\end{equation}
\end{lemma}
\begin{proof}[Sketch]
After proving the existence of globally near-optimal strategies, we suitably define the number $\constFactor$
and take an arbitrary globally $\eps''$-optimal delay function $\timeouts''$, where $\eps''<<\eps'$. If this function \emph{does not} satisfy \eqref{eq:step-bound}, we conclude that it must induce the following pathological behaviour in $\fdC$: the system stays for a long time in a component of its state space such that a) fixed-delay transitions are active in each state of the component, each such transition within the component having zero impulse cost; and b) function $\timeouts''$ assigns very small (in a well-defined sense) delays to all states of the component. We call such a component a \emph{bad sink}. Intuitively, inside a bad sink the system rapidly performs one fixed-delay transition after another, incurring only a tiny cost between two successive transitions. This allows the delay function to perform many steps while staying  $\eps''$-optimal. (In Example~\ref{ex:non-stability}, $\{a,b\}$ would be a bad sink for $\timeouts$, as with high probability the cycle on these two states is completed every $0.0101$ units of time, with cost $0.0102$ incurred per cycle.)

To obtain a globally $\eps'$-optimal delay function satisfying~\eqref{eq:step-bound}, we carefully modify $\timeouts''$ so as to remove all bad sinks. This is done by selecting a suitable state in each bad sink and ``inflating'' its delay to a sufficiently high threshold. Choosing the right state and threshold is a rather delicate process, since an improper choice might significantly increase the incurred cost. Also note that Lemma~\ref{lem:perturbation-error} cannot be used to bound the increase in cost caused by the modification, as we do not know the value of $\CostBound{\msteps,\timeouts''}$. Instead, we utilize non-trivial insights into the structure of $\fdC$ and $\mdp$. \qed
\end{proof}

By using these proportional delay functions, we reduce the perturbation error of Lemma~\ref{lem:perturbation-error} only to a function of $\CostBound{\mcost,\timeouts}$. 
Combining this with Lemma~\ref{lem:mesh-error}, we obtain that the delay functions in $\paramspace(\delta,\dmax)$ approximate all the proportional delay functions $\timeouts$ of Lemma~\ref{lem:bounded-step-existence}, and thus $\Value{\contMdp,\paramspace(\delta,\dmax)}$ approximates $\Value{\contMdp}$.
The parameters $\delta,\dmax$ depend on $\eps$ and $\CostBound{\mcost,\timeouts}$ of any such $\timeouts$ from Lemma~\ref{lem:bounded-step-existence}. As these delay functions are \emph{globally} $\eps$-optimal, all such $\CostBound{\mcost,\timeouts}$ can be $\eps$-approximated by $\maxValue{\contMdp} := \max_{s \in S'}\Value{\contMdp[s]}$.

\begin{proposition}
	\label{prop:mesh-delays-ok}

For $\constFactor$ from Lemma~\ref{lem:bounded-step-existence}, $\discconst$
and $\cutconst$ from Lemma~\ref{lem:mesh-error}, it holds that
$$\left\lvert \; \Value{\contMdp} - \Value{\contMdp,\paramspace(\delta,\dmax)} \; \right\rvert 
\;\; \leq \;\; \frac{\eps}{2}$$
\[ \textrm{where }~~
\delta \; :=\; \frac{\alpha}{\discconst}, ~~ \dmax\; := \; |\log(\alpha)|\cdot \cutconst\cdot (\maxValue{\contMdp}+\eps), ~~ \alpha \; := \; \frac{\eps^2}{64\constFactor\cdot|\states'|\cdot (1+\maxValue{\contMdp})^2}.
\]
\end{proposition}

\subsubsection{Bounding $\maxValue{\contMdp}$}

In Proposition~\ref{prop:mesh-delays-ok}, the allowed perturbation $\alpha$ and hence the fineness of the mesh $\delta$ needed to obtain the required precision depend on the bound 
$\maxValue{\contMdp}$.
We first provide the following theoretical worst-case bound.

\begin{lemma}
	\label{lem:bound-value}
	There is a number $\maxvalconst\in \exp(\size{\fdC}^{\mathcal{O}(1)})$ computable in time polynomial in $\size{\fdC}$ such that $\maxValue{\contMdp} \leq \maxvalconst$. 
\end{lemma}
In practice, one can obtain better bounds by computing $\max_{s\in S} \expred_{\fdC[s](\timeouts)}$ for an arbitrary $\timeouts$ as $\max_{s\in S} \expred_{\fdC[s](\timeouts)} \geq \max_{s\in S} \inf_{\timeouts'} \expred_{\fdC[s](\timeouts')} = \maxValue{\contMdp}$. One can set $\timeouts$ by some heuristics (e.g. to the constant function $1/\lambda$) or randomly. One can even use the minimum from a series of such computations.
We believe that in most cases, this yields a significant improvement. For instance, for the $3$-state model from Section~\ref{sec-intro}, we get a bound $\max_{s\in S} \expred_{\fdC[s](1/\lambda)} \approx 4.3$ instead of the theoretical bound $\maxValue{\contMdp} \approx 55000$.

\myspacesmall
\subsubsection{Representing the Finite Mesh}
Since one-step costs and probabilities produced by delay functions in $\paramspace(\delta,\dmax)$ may be irrational, we need to approximate them by rational numbers.
So let us fix $\delta$ and $\dmax$ from Proposition~\ref{prop:mesh-delays-ok}. 
For any $\kappa>0$ we define DTMDP $\mdp_{\kappa}=(\states',\macts_{\kappa},\mtran_{\kappa}, \initstate,G)$ with a cost function $\mcost_{\kappa}$ where 
\begin{itemize}
	\item the strategies are exactly delay functions from $D(\delta,\dmax)$, i.e. $\macts_{\kappa} =\{k\delta\mid k\in\Nset, \delta \leq k\delta\leq \dmax\}\cup \{\infty\}$ where again $\infty$ is enabled in $\sta \in \states'\setminus\allact$ and the rest is enabled in $\sta \in \allact$; and 
	\item for all $(s,\timeouts)\in S'\times \macts_{\kappa}$ the transition probabilities in $\mtran_{\kappa}(s,\timeouts)$ and costs in $\mcost_{\kappa}(s,\timeouts)$ are obtained by rounding the corresponding numbers in $\mtran(s,\timeouts)$ and $\mcost(s,\timeouts)$ up to the closest multiple of $\kappa$.\footnote{More precisely, all but the largest probability in $\mtran(s,\timeouts)$ are rounded up, the largest probability is suitably rounded down so that the resulting vector adds up to 1. 
}
	For more details and pseudo-codes see Appendix~\ref{app:algs}.
\end{itemize}

\begin{proposition}
\label{prop:rounding-error}
Let $\eps>0$ and fix $\kappa=(\eps\cdot\delta\cdot \minRew)/({2\cdot |\states'|\cdot(1+\maxValue{\contMdp})^2}),$ where $\minRew$ is a minimal cost rate in $\fdC$.
	\label{prop:mdp-formulation-2}
	Then it holds
	$$ \left\lvert \; \Value{\contMdp,\paramspace(\delta,\dmax)} - \Value{\mdp_{\kappa}} \; \right\rvert 
	\;\; \leq \;\; \frac{\eps}{2}.$$
\end{proposition}
\begin{proof}[Sketch]
We use similar technique as in Lemma~\ref{lem:perturbation-error}, taking advantage of the fact that  probabilities and costs of each action are changed by at most $\kappa$ by the rounding. \qed  
\end{proof}

\subsubsection{The Algorithm for Theorem~\ref{thm:unconstrained}}
First the discretization step $\delta$, maximal delay $\dmax$, and rounding error $\kappa$ are computed. Then the discretized DTMDP $\mdp_{\kappa}$  is constructed according to the above-mentioned finite mesh representation. Finally the globally optimal delay function from $\mdp_{\kappa}$ is chosen using standard polynomial algorithms for finite MDPs~\cite{Puterman:book,EWY:RSG-Positive-Rewards}. From Propositions~\ref{prop:mesh-delays-ok} and~\ref{prop:mdp-formulation-2} it follows that this delay function is $\eps$-optimal in $\mdp$, and thus also in $\fdC$ (Proposition~\ref{prop:mdp-formulation}). 

The size of $\mdp_{\kappa}$ (and its construction time) can be stated in terms of a polynomial in $\size{\fdC}$, $\maxValue{\contMdp}$, $1/\delta$, $\dmax$, and $1/\kappa$. Examining the definitions of these parameters in Propositions~\ref{prop:mesh-delays-ok} and~\ref{prop:mdp-formulation-2}, as well as the bound on $\maxValue{\contMdp}$ from Lemma~\ref{lem:bound-value}, we conclude that the size of $\mdp_{\kappa}$ and the overall running time of our algorithm are exponential in $\size{\fdC}$ and polynomial in $1/\eps$. The pseudo-code of the whole algorithm is given in \appref{app:algs}.

\vspace{-0.5\baselineskip}
\section{Bounded Optimization Under Partial Observation}
\label{sec:results-multi}

In this section, we address the cost optimization problem for delay functions chosen under partial observation. For an 
equivalence relation $\equiv$ on $\allact$ specifying observations, and $\dmin,\dmax > 0$, we define $\paramspace(\dmin,\dmax,\equiv) = \{\timeouts \mid \forall s,s': \dmin \leq \timeouts(s) \leq \dmax, s\equiv s' \Rightarrow \timeouts(s)=\timeouts(s')\}$.

\begin{theorem}\label{thm:multi-entry-approximation}
There is an algorithm that for a fdCTMC structure $\fdC$, a cost structure $\costFdC$ with $\rateRew(\sta) > 0$ for all $\sta\in\states$, 
an equivalence relation $\equiv$ on $\allact$, $\dmin,\dmax > 0$, and $\eps>0$ computes in time exponential in $\size{\fdC}$, $\size{\dmin}$, and $\dmax$ a delay function $\timeouts$ such that
$$
\left\lvert 
\inf_{\timeouts' \in \paramspacePO} \expred_{\fdC(\timeouts')}
- \expred_{\fdC(\timeouts)}
\; \right\rvert
\;\; < \;\; \eps.$$
\end{theorem}

\noindent
Also, one cannot hope for polynomial complexity as the corresponding threshold problem is NP-hard, even if we restrict to instances where $\dmax$ is of magnitude polynomial in $\size{\fdC}$.

\begin{theorem}\label{thm:np-complete}
For a fdCTMC structure $\fdC$, a cost structure $\costFdC$ with $\rateRew(\sta) > 0$ for all $\sta\in\states$, 
an equivalence relation $\equiv$ on $\allact$, $\dmin,\dmax > 0$, $\eps>0$, and $\computedValue \in\Rsetpo$, it is NP-hard to decide
$$ \text{whether} \qquad 
\inf_{\timeouts \in \paramspacePO} \expred_{\fdC(\timeouts)}
 \;>\; \computedValue + \eps
\qquad \text{or} \qquad
\inf_{\timeouts \in \paramspacePO} 
\expred_{\fdC(\timeouts)}
 \;<\; \computedValue - \eps
$$
(if the optimal cost lies in the interval $[\computedValue-\eps,\computedValue+\eps]$, an arbitrary answer may be given). The problem remains NP-hard even if $\dmax$ is given in unary encoding.
\end{theorem} 
For $\dmax$ given in unary we get a matching upper bound.

\begin{theorem}
\label{thm:PO-unary-in-NP}
The approximate threshold problem of Theorem~\ref{thm:np-complete} is in NP, provided that $\dmax$ is given in unary.
\end{theorem}

We leave the task of settling the exact complexity of the general problem (where $\dmax$ is given in binary) to future work. 

For the rest of this section we fix a fdCTMC structure $\fdC = (\states, \lambda, \prob, \allact, \trans, \initstate)$, a cost structure $\costFdC = (\goalStates, \rateRew, \impRewExp, \impRewFix)$, $\eps > 0$, and an equivalence relation $\equiv$ on $\allact$, $\dmin,\dmax > 0$.
We simply write $\paramspacePOShort$ instead of $\paramspacePO$ and again assume that $\Value{\fdC,\paramspacePOShort} < \infty$.

\myspacesmall
\subsection{Approximation Algorithm}

In this Section, we address Theorem~\ref{thm:multi-entry-approximation}. First observe, that the MDP $\contMdp$ introduced in Section~\ref{sec:results-single} can be due to Proposition~\ref{prop:mdp-formulation} also applied in the bounded partial-observation setting. Indeed,
$\expred_{\fdC(\timeouts)} =  \expred_{\contMdp(\timeouts)}$ for each $\timeouts \in\paramspacePOShort$
and thus, $\Value{\fdC,\paramspacePOShort} = \Value{\contMdp,\paramspacePOShort}$ (where analogously $\Value{\fdC,\paramspacePOShort}$ denotes $\inf_{\timeouts\in \paramspacePOShort} \expred_{\fdC(\timeouts)}$).
Furthermore, by fixing a mesh $\delta$ and a round-off error $\kappa$, we define a finite DTMDP $\matchmdp_D$ where
\vspace{-0.25\baselineskip} 
\begin{itemize}
	\item actions are restricted to a finite mesh of multiples of $\delta$ within the bounds $\dmin$ and $\dmax$; and  
	\item probabilities and costs are rounded to multiples of $\kappa$ as in Section~\ref{sec:results-single}.
\end{itemize}
To show that $\matchmdp_D$ suitably approximates $\contMdp$ we use similar techniques as in Section~\ref{sec:results-single}. However, thanks to the constraints $\dmin$ and $\dmax$ we can show that for \emph{every} delay function $\timeouts\in \paramspacePOShort$ the values $\CostBound{\msteps,\timeouts}$ and $\CostBound{\mcost,\timeouts}$, which feature in Lemma~\ref{lem:perturbation-error}, are bounded by a function of $\size{\fdC}$, $\dmin$ and $\dmax$ (in particular, the bound is independent of $\timeouts$). This substantially simplifies the analysis. We state just the final result.

\begin{proposition}\label{prop:multi-entry-aprox}
There is a number $\POconst \in \exp((\size{\fdC}\cdot\size{\dmin}\cdot \dmax)^{\mathcal{O}(1)})$ such that for $\delta= \eps/\POconst$ and $\kappa=(\eps\cdot \delta)/\POconst$ it holds	$\left\lvert\Value{\contMdp,\paramspacePOShort}-\Value{\matchmdp_D} \right\rvert < \varepsilon.$
\end{proposition}

\noindent
The proof of Theorem~\ref{thm:multi-entry-approximation} is finished by the following algorithm.
\begin{itemize}
\item For $\delta$ and $\kappa$ from Proposition~\ref{prop:multi-entry-aprox}, the algorithm first constructs (in the same fashion as in Section~\ref{sec:results-single}) in 2-exponential time the MDP $\matchmdp_D$. 
\item Then it finds an optimal strategy $\timeouts$ (which also satisfies $\lvert \expred_{\fdC(\timeouts)} - \inf_{\timeouts'}\expred_{\fdC(\timeouts')}\rvert < \eps$) by computing $\expred_{\discMdp'(\timeouts)}$ for every (MD) strategy $\timeouts$ of $\discMdp'$ in the set $\paramspacePOShort$. 
\end{itemize}
\noindent
The algorithm runs in 2-EXPTIME because there are $\leq |\mactsDisc|^{|\states|}$ strategies which is exponential in $\size{\fdC}$, $\size{\dmin}$, and $\dmax$ as $|\mactsDisc|$ is exponential in these parameters. The correctness follows from Propositions~\ref{prop:mdp-formulation},~\ref{prop:multi-entry-aprox}, proving Theorem~\ref{thm:multi-entry-approximation}.

\vspace{-\baselineskip}
\subsubsection{Challenges of Unbounded Optimization}
The proof of Proposition~\ref{prop:multi-entry-aprox} is simpler than the techniques from Section~\ref{sec:results-single} because we work with the compact space bounded by $\dmin$ and $\dmax$.
This restriction is not easy to lift; the techniques from Section~\ref{sec:results-single} cannot be easily adapted to \emph{unbounded} optimization under partial observation.

\setlength{\intextsep}{3pt}
\begin{wrapfigure}{r}{0\textwidth}
	\begin{tikzpicture}[outer sep=0.1em, xscale=0.9, yscale=0.57]
	
	\tikzstyle{fixed}=[dashed,->];
	\tikzstyle{fixed label}=[font=\small];
	\tikzstyle{exp}=[->,>=stealth];
	\tikzstyle{exp rate}=[font=\small];
	\tikzstyle{loc}=[draw,circle, minimum size=1.5em,inner sep=0.1em];
	\tikzstyle{accepting}+=[outer sep=0.1em];
	\tikzstyle{loc cost}=[draw,rectangle,inner sep=0.07em,above=6, minimum width=0.8em,minimum height=0.8em,fill=white,font=\footnotesize];
	
	\node[loc] (s) at (0,0) {$a$};
	\node[loc, accepting] (g) at (1.2,1.3) {$t$};
	\node[loc] (u) at (2.4,0) {$b$};
	\node[loc] (t) at (1.2,0) {$0$};
	\node[loc] (t1) at (1.2,-1.3) {$1$};
	\node[loc] (t2) at (1.2,-2.6) {$2$};

	\path[->,>=stealth] ($(s)+(-.6,0)$) edge (s);
	
	\draw[fixed,rounded corners] (s) |- +(1, 2) -| (u);
	\draw[fixed,rounded corners] (u) |- +(-1,-3.4) -| (s);
	
	\path[fixed] (t) edge (g);
	\path[exp] (t) edge (t1);
	\path[fixed] (t1) edge (t2);
	\path[fixed,bend left=15] (t2) edge (s);
	
	\path[exp,loop right] (t1) edge (t1);
	\path[exp,bend right=35] (t2) edge (t1);
	
	\path[exp] (s) edge (t);
	\path[exp] (u) edge (t);
	\end{tikzpicture}
	\vspace{-10pt}
\end{wrapfigure}
The reason is that local adaptation of the delay function (heavily applied in the proof of Lemma~\ref{lem:bounded-step-existence}) is not possible as the delays are not independent. 
Consider on the right a variant of Example~\ref{ex:non-stability} with components $a$ and $b$ being switched by fixed-delay transitions. All states have cost rate $1$ and all transitions have cost $0$; furthermore, all states are in one class of equivalence of $\equiv$.
If in state $a$ or $b$ more than one exp-delay transition is taken before a fixed-delay transition, a long detour via state $1$ is taken. In order to avoid it and to optimize the cost, one needs to set the one common delay as close as possible to $0$. 
Contrarily, in order to decrease the expected number of visits from $a$ to $b$ from $a$ before reaching $t$ which is crucial for the error bound, one needs to increase the delay.

\subsection{Complexity of the Threshold Problem} 

We now turn our attention to Theorem~\ref{thm:np-complete}. 
We show the hardness by reduction from SAT. Let $\varphi = \varphi_1 \land \cdots \land \varphi_n$ be a propositional formula in conjunctive normal form (CNF) with $\varphi_i = (l_{i,1} \lor \cdots \lor l_{i,k_i})$ for each $1 \leq i \leq n$ and with the total number of literals $k = \sum_{i=1}^n k_i$. As depicted in the following figure, the fdCTMC structure $\fdC_\varphi$ is composed of $n$ components (one per clause), depicted by rectangles. The component of each clause is formed by a cycle of sub-components (one per literal) connected by fixed-delay transitions. Positive literals are modelled differently from negative literals. 

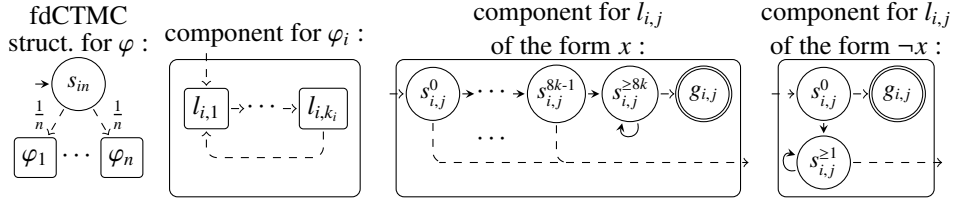
\vspace{-\baselineskip}
\begin{center}
\begin{tikzpicture}[outer sep=0.1em, xscale=0.57, yscale=0.65]

\tikzstyle{fixed}=[dashed,->];
\tikzstyle{fixed label}=[font=\small];
\tikzstyle{exp}=[->,>=stealth];
\tikzstyle{exp rate}=[font=\small];
\tikzstyle{title}=[];

\tikzstyle{loc}=[draw,circle, minimum size=2em,inner sep=0.1em,font=\small];
\tikzstyle{gadget}=[draw,rectangle, rounded corners=1.5, minimum size=1.5em,inner sep=0.3em];
\tikzstyle{biggadget}=[rounded corners=3];

\tikzstyle{accepting}+=[outer sep=0.1em];
\tikzstyle{loc cost}=[draw,rectangle,inner sep=0.07em,above=6, minimum width=0.8em,minimum height=0.8em,fill=white,font=\footnotesize];
\tikzstyle{trans cost}=[draw,rectangle,minimum width=0.8em,minimum height=0.8em,solid,inner sep=0.07em,fill=white,font=\footnotesize];
\tikzstyle{prob}=[inner sep=0.03em,pos=0.25, auto,font=\footnotesize];

\begin{scope}[xshift=2cm]
\node[title] at (2,1.9) {fdCTMC};
\node[title] at (2,1.3) {struct. for $\varphi:$};

\node[loc] (phi_init) at (2,0.5) {$\initstate$};
\path[->,>=stealth] ($(phi_init)+(-1,0)$) edge (phi_init);

\node[gadget] (phi1) at (1,-1) {$\varphi_1$};
\node at (2,-1) {$\cdots$};
\node[gadget] (phin) at (3,-1) {$\varphi_n$};

\path[fixed] (phi_init) edge  node[prob,swap,pos=0.9] {$\frac{1}{n}$} (phi1);
\path[fixed] (phi_init) edge  node[prob,pos=0.9] {$\frac{1}{n}$} (phin);
\end{scope}

\begin{scope}[xshift=4.8cm]
\node[title] at (3.5,1.5) {component for $\varphi_i:$};

\node[gadget] (phii_1) at (2.15,0) {$l_{i,1}$};
\path[fixed] ($(phii_1)+(0,+1.13)$) edge (phii_1);

\node[inner sep=0] (phii_3) at (3.47,0) {$\cdots$};
\node[gadget] (phii_k) at (4.85,0) {$l_{i,k_i}$};

\path[fixed] (phii_1) edge  (phii_3);
\path[fixed] (phii_3) edge  (phii_k);
\draw[fixed,rounded corners] (phii_k) |- (4,-1) -| (phii_1);
\draw[biggadget] (1.3,1) rectangle (5.7,-1.8);
\end{scope}

\begin{scope}[xshift=10.03cm]
\node[title] at (5.275,1.9) {component for $l_{i,j}$};
\node[title] at (5.275,1.3) {of the form $x:$};

\node[loc] (s_0) at (2.15,0.3) {$s_{i,j}^0$};
\path[fixed] ($(s_0)+(-1,0)$) edge (s_0);

\node[inner sep=0] (s_2) at (3.55,0.3) {$\cdots$};
\node[loc] (s_3) at (5,0.3) {$s_{i,j}^{8k\text{-}1}$};
\node[loc] (s_k) at (6.7,0.3) {$s_{i,j}^{\geq 8k}$};

\node[loc,accepting] (s_acc) at (8.4,0.3) {$g_{i,j}$};
\node (exit) at (9.7,-1.1) {};

\path[exp] (s_0) edge  (s_2);
\path[exp] (s_2) edge  (s_3);
\path[exp] (s_3) edge  (s_k);
\path[exp,loop below,looseness=2.5] (s_k) edge  (s_k);
\path[fixed] (s_k) edge  (s_acc);

\draw[fixed,rounded corners] (s_0) |- +(0.4,-1.4) -- (exit);
\node[inner sep=0] (s_2) at (3.55,-0.65) {$\cdots$};
\draw[fixed,rounded corners,-] (s_3) |- +(0.4,-1.4);

\draw[biggadget] (1.3,1) rectangle (9.25,-1.8);
\end{scope}

\begin{scope}[xshift=19.2cm]
\node[title] at (2.75,1.9) {component for $l_{i,j}$};
\node[title] at (2.75,1.3) {of the form $\neg x:$};

\node[loc] (s_0) at (2,0.3) {$s_{i,j}^0$};
\path[fixed] ($(s_0)+(-1.2,0)$) edge (s_0);

\node[loc] (s_1) at (2,-1.1) {$s_{i,j}^{\geq1}$};

\node[loc,accepting] (s_acc) at (3.7,0.3) {$g_{i,j}$};
\node (exit) at (5,-1.1) {};

\path[exp] (s_0) edge  (s_1);
\path[exp,loop left,looseness=3] (s_1) edge  (s_1);
\path[fixed] (s_0) edge  (s_acc);
\path[fixed] (s_1) edge  (exit);

\draw[biggadget] (0.95,1) rectangle (4.55,-1.8);
\end{scope}
\end{tikzpicture}
\end{center}

The cost structure $\costFdC_\varphi$ assigns rate cost $1$ to every state, and impulse cost $0$ to every transition; the goal states are depicted by double circles and exp-delay transitions are depicted with heavier heads  to distinguish from the dashed fixed-delay transitions. 
We require $s_{i,j}^0 \equiv s_{i',j'}^0$ iff the literals $l_{i,j}$ and $l_{i',j'}$ have the same variable. Furthermore, let $\paramspacePOShort$ denote $\paramspace(0.01,16k,\equiv)$. Note that $\dmax=16k$ is linear in $\size{\varphi}$ and thus it can be encoded in unary. We obtain the following:

\begin{proposition}\label{prop:np-hardness}

For a formula $\varphi$ in CNF with $k$ literals, $\fdC_\varphi$ and $\costFdC_\varphi$ are constructed in time polynomial in $k$ and, furthermore,
$$\Value{\fdC_\varphi,\paramspacePOShort} < 17 k^2 \;\; \text{if $\varphi$ is satisfiable \;\;\; and} \;\;\;\; \Value{\fdC_\varphi,\paramspacePOShort} > 17k^2 + 1 \text{, otherwise.}$$
\end{proposition}
\begin{proof}[Sketch]
If $\varphi$ is satisfiable, let $\nu$ denote the satisfying truth assignment. We set $\timeouts(\initstate) := \dmin$ and $\timeouts(s_{i,j}^0) := \dmax$ if $\nu(X) = 1$ and $\timeouts(s_{i,j}^0) := \dmin$ if $\nu(X) = 1$, where $X$ is the variable of the literal $l_{i,j}$ and arbitrarily in other states.
In $\fdC_\varphi(\timeouts)$ in the component for any TRUE literal $l_{i,j}$, i.e. with $\nu(l_{i,1}) = 1$, the goal is reached from $s_{i,j}^0$ with probability $>0.99$ before leaving the component. Indeed, the probability to take no exponential transition within time $0.01$ is $>0.99$ and the probability to take at least $8k$ exponential transitions within time $16k$ is $>0.99$ for any $k\in\Nset$.
As each clause $\varphi_i$ has at most $k$ literals and at least one TRUE literal, the expected cost incurred in the component for $\varphi_i$ is at most $(16k\cdot k)/0.99 < 17k^2$.

The other implication is based on the observation that there is no delay length guaranteeing high probability of going directly into the target from both components for a positive and a negative literal. \qed
\end{proof}

The reduction proves NP-hardness as it remains to set $\computedValue := 2k^2 + \frac{1}{2}$ and $\eps := \frac{1}{2}$.

\myspacesmall
\vspace{-0.75\baselineskip}
\subsubsection{NP Membership for Unary $\dmax$} To prove Theorem~\ref{thm:PO-unary-in-NP} we give an algorithm which, for a given
approximate threshold $\computedValue>0$, consists of
\begin{itemize}
\item \emph{first} guessing the delay function $\timeouts$ of $\matchmdp_\paramspacePOShort$ that is in the set $\paramspacePOShort$ such that $\expred_{\matchmdp_\paramspacePOShort(\timeouts)} < \computedValue$,
\item \emph{then} constructing just the fragment $\guessMdp$ of $\matchmdp_\paramspacePOShort$ used by the guessed function $\timeouts$. Here $\guessMdp = (\states', \{\infty\},\mtranGuess, G, \mcostGuess)$ where the transition probabilities and costs coincide with $\contMdp$ for states in $\states'\setminus\allact$ and in any state $s\in\allact$ are defined by $\mtranGuess(s,\infty) = \mtran^\star(s,\timeouts(s))$ and $\mcostGuess(s,\infty) = \mcost^\star(s,\timeouts(s))$ (here $\mtran^\star(s,\timeouts(s))$ and $\mcost^\star(s,\timeouts(s))$ are as in~$\matchmdp_\paramspacePOShort$).
\item Last, for $\sigma: s \mapsto \infty$, the algorithm computes $y = \expred_{\guessMdp(\sigma)}$ by standard methods and \emph{accepts} iff $y < \computedValue$.
\end{itemize}

Note that when $\dmax$ is encoded in unary, both $\timeouts$ and $\guessMdp$ are of bit size that is polynomial in the size of the input. Hence, $\timeouts$ and $\guessMdp$ can be constructed in non-deterministic polynomial time (although the whole $\matchmdp_\paramspacePOShort$ is of exponential size in this unary case). The expected total cost $\computedOutcome$ in $\guessMdp(\sigma)$ that has polynomial size can be also computed in polynomial time. The correctness of the algorithm  easily follows from Proposition~\ref{prop:multi-entry-aprox}; for an explicit proof see~\appref{subsec:correctness}.

\vspace{-0.5\baselineskip}
\section{Conclusions
}\label{sec:concl}
\vspace{-0.5\baselineskip}

In this paper, we introduced the problem of synthesizing timeouts for fixed-delay CTMC. We study two variants of this problem, show that they are effectively solvable, and obtain provable worst-case complexity bounds.
First, for \emph{unconstrained optimization}, we present an approximation algorithm based on a reduction to a discrete-time Markov decision process and a standard optimization algorithm for this model. 
Second, we approximate the case of \emph{bounded optimization under partial observation} also by a MDP.
However, a restriction of the class of strategies twists it basically into a partial-observation MDP (where only memoryless deterministic strategies are considered). We give a 2-exponential approximation algorithm (which becomes exponential if one of the constraints is given in unary) and show that the corresponding decision problem is NP-hard. 

The correctness of our algorithms stems from non-trivial insights into the behaviour of fdCTMC that we deem to be interesting in their own right. Hence, we believe that techniques presented in this paper lay the ground for further development of performance optimization via timeout synthesis.

\vspace{-0.6\baselineskip}

\bibliographystyle{plain}
\bibliography{str-short,concur}
\newpage
\appendix
\section{Algorithms}\label{app:algs}

In this section, we present pseudo-codes for algorithms.

\algabove
\begin{procedure}[h]
\SetAlgoLined
\DontPrintSemicolon
\SetKwInOut{Input}{input}\SetKwInOut{Output}{output}
\SetKwData{n}{n}\SetKwData{f}{f}\SetKwData{g}{g}
\SetKwData{Low}{l}\SetKwData{x}{x}
\Input{fdCTMC structure $\fdC$, cost structure $\costFdC$, and $\delta, \kappa$, $\ell$, $u > 0$.}
\Output{A DTMDP $\mdp=(\states',\macts,\mtran, \initstate,G)$ with a cost function $\mcost$.}
\BlankLine
Construct $\states'$, $\initstate$, and $G$ as specified in Subsection~\ref{subsec:reduction-DTMDP}\;
$\macts := \{k\delta \mid k \in \Nset \text{ and } \ell < k\delta \leq u\} \cup \{\infty\}$\;
\BlankLine
\For{$s \in \states' \cap \allact$ and $\delays \in \macts \setminus \{\infty\}$}{
Compute distribution $\mtran(s,\delays)$ and cost in $\mcost(s,\delays)$ using $\fdC[s](\delays)$ and $\kappa$ \;
}
\BlankLine
\lFor{$s\in\states'\setminus\allact$}{
	$\mtran(s,\infty):= \prob(s,\cdot)$;
	$\mcost(s,\infty) := \rateRew(\sta) + \sum_{\sta' \in\states'} \prob(\sta,\sta')\cdot \impRewExp (\sta,\sta')$\;
}
\BlankLine
set $\mtran$ to $\bot$ and $\mcost(s,a)$ to $0$ for all remaining $s\in\states'$ and $a\in\macts$\;
\caption{Discretize($\fdC$, $\costFdC$, $\delta$, $\kappa$, $\ell$, $u$)}
\label{alg:discretize}
\end{procedure}
\algbelow

\algabove
\begin{algorithm}[h]
\SetAlgoLined
\DontPrintSemicolon
\SetKwInOut{Input}{input}\SetKwInOut{Output}{output}
\SetKwData{n}{n}\SetKwData{f}{f}\SetKwData{g}{g}
\SetKwData{Low}{l}\SetKwData{x}{x}

\Input{fdCTMC structure $\fdC$, cost structure $\costFdC$ with $\rateRew(\sta) > 0$ for all $\sta\in\states$, and $\eps>0$.}
\Output{The delay function $\timeouts$ that is $\eps$-optimal in $\fdC$ and $\costFdC$.}
\BlankLine
Compute $\maxValue{\contMdp}$, $\delta$, $\kappa$, and $\dmax$\;
Construct $\mdp_{\kappa}$ and $\mcost_{\kappa}$ by procedure Discretize($\fdC$, $\costFdC$, $\delta$, $\kappa$, $\delta$, $\dmax$) \;
$\timeouts := $ optimal strategy  in $\mdp_{\kappa}$ and $\mcost_{\kappa}$ \tcp*{e.g. alg.\cite{Puterman:book} in PTIME in $|\mdp_{\kappa}|$ and $|\mcost_{\kappa}|$}
\caption{Unconstrained optimization}
\label{alg:single-entry}
\end{algorithm}
\algbelow

\algabove
\begin{algorithm}[h]
\SetAlgoLined
\DontPrintSemicolon
\SetKwInOut{Input}{input}\SetKwInOut{Output}{output}
\SetKwData{n}{n}\SetKwData{f}{f}\SetKwData{g}{g}
\SetKwData{Low}{l}\SetKwData{x}{x}
\Input{fdCTMC structure $\fdC$, cost structure $\costFdC$ with $\rateRew(\sta) > 0$ for all $\sta\in\states$,  equivalence relation $\equiv$ on $\allact$, $\dmin,\dmax > 0$, and $\eps>0$.}
\Output{The delay function $\timeouts$ that is $\eps$-optimal in $\fdC$, $\costFdC$, $\dmin$,$\dmax$, and $\equiv$.}
\BlankLine
Compute $\maxValue{\contMdp}$, $\delta$, and $\kappa$\;
Construct $\matchmdp_D$ and $\mcost_D$ by procedure Discretize($\fdC$, $\costFdC$, $\delta$, $\kappa$, $\dmin$, $\dmax$) \;
$\timeouts := $ optimal strategy  in $\matchmdp_D$, $\mcost_D$, and $\equiv$ \tcp*{standard algorithms for POMDP}
\caption{Bounded optimization under partial observation}
\label{alg:multiple-entry}
\end{algorithm}
\algbelow

Let us discuss in closer detail, how the transition probabilities and costs are approximated.

\subsubsection{Approximating Transition Probabilities and Costs}
Note that for $\sta \in \states' \setminus\allact$ and $\delays = \infty$, the approximation is easy as these states behave as in CTMC and so the probabilities in $\mtran(\sta,\delays)$ and cost  $\mcost(\sta,\delays)$ are already expressed as rational numbers.

Let us fix a state $\sta \in \states' \cap \allact$,  $\delays \in \macts \setminus \{\infty\}$, and error bound $\kappa \in \Rsetp$.  
First, we compute all probabilities in $\mtran(\sta,\delays)$ and cost $\mcost(\sta,\delays)$ up to an absolute error of $\kappa/2$ employing transient analysis of $\fdC[\sta](\delays)$.  
Then we round up (to the closest multiple of $\kappa$) cost and  all probabilities except the highest one, that is accordingly rounded down so that the sum of probabilities is $1$.

Before the transient analysis we need to modify $\fdC[\sta](\delays)$ so that all states $\states' \setminus \{\sta\}$ are (1)~absorbing w.r.t. both exponential and fixed-delay transitions and (2)~incur neither rate nor impulse costs. If there is a fixed-delay transition leading to $\sta$, then we create an absorbing copy of $\sta$ where we redirect all such transitions. 
After this transformation we obtain a (subordinated) continuous-time Markov chain $\overline{\fdC}$, and $\mtran(\sta,\delays)$ is a transient distribution in time $\delays$ and $\mcost(\sta,\delays)$ is expected accumulated cost up to time $\delays$ in this $\overline{\fdC}$. 
We denote the modified transition matrices by $\overline{\prob}$ and $\overline{\trans}$ and the modified cost structure by $(\overline{\rateRew}, \overline{\impRewExp}, \overline{\impRewFix})$.
Similarly to the method of uniformization~\cite{Jensen:uniformization} or methods in \cite{DBLP:journals/pe/Lindemann93}, it is easy to show that 
\begin{lemma}\label{lem:trans_cost_form}
For any state $s \in\states' \cap \allact$, and $\delays \in \Rsetp$, 
\begin{align*}
\mtran(s,\delays) & = \sum_{i=0}^\infty \psi_{\lambda\delays}(i) \cdot \left( \vec{1}_s \cdot \overline{\prob}^i \right) \cdot \overline{\trans} \\
\mcost(s,\delays) & =  \sum_{i=0}^\infty \psi_{\lambda\delays}(i) \left(
\sum_{j=0}^{i-1} \left(\vec{1}_s \cdot \overline{\prob}^j\right) \cdot \left(\frac{\delays\cdot\overline{\rateRew}}{i+1} + \overline{\expectedRewExp} \right) \;+\;
\left(\vec{1}_s \cdot \overline{\prob}^i\right) \cdot \left(\frac{\delays\cdot\overline{\rateRew}}{i+1} + \overline{\expectedRewFix}\right)
 \right)
\end{align*}
where $\vec{1}_s$ denotes the unit vector of state $\sta$, $\psi_{\lambda\delays}$ denotes the probability mass function of the Poisson distribution with parameter $\lambda\delays$, and $\overline{\expectedRewExp},\overline{\expectedRewFix}: \states \to \Rsetpo$ assign to each state the expected impulse reward of the next exponential or fixed-delay transition, i.e. $\overline{\expectedRewExp}(s) = \sum_{s'}\overline{\prob}(s,s')\cdot\overline{\impRewExp}(s,s')$ and 
$\overline{\expectedRewFix}(s) = \sum_{s'}\overline{\trans}(s,s')\cdot\overline{\impRewFix}(s,s')$. 
\end{lemma}
Indeed, the distribution over states after $j$ exponential transitions is $\vec{1}_s \cdot \overline{\prob}^j$; the probability that exactly $i$ exponential transitions occur before time $\delays$ is $\psi_{\lambda\delays}(i)$; and under the condition that $i$ transitions occur, the expected time spent in each of the $i+1$ visited states is $\delays/(i+1)$~\cite{DBLP:journals/pe/Lindemann93}.

Finally, note that for any fixed error $\otherError > 0$, we can easily compute a bound $I$ such that $\sum_{i=I}^\infty \psi_{\lambda\delays}(i) < \otherError/2$. Hence the $\mtran(\sta,\delays)$ can be approximated up to $\otherError/2$ by truncating the infinite sums at $I$. Similarly there is a bound $J$ for approximation of $\mcost(\sta,\delays)$ since $\psi_{\lambda\delays}$ decreases exponentially while the costs increase linearly for increasing $i$.

\section{Proofs for Section~\ref{sec-prelims}} \label{sec:app-prelims}

\subsection{Definition of probability space for fdCTMC}\label{subsec:prob-spaceFdCTMC}

For a run $(\sta_0,\delays_0)t_0 (\sta_1,\delays_1) t_1 \cdots$ and for each $i\in\Nseto$ we define function $\nxt(\sta_i,\delays_i;t_i)$ as follows:
\begin{itemize}
\item If $t_i < \delays_i$ (an exp-delay transition occurs), $$\nxt(\sta_i,\delays_i;t_i) = \begin{cases}
\delays_i - t_i & \text{if $\sta_{i+1} \in \allact$ and $\sta_i\in \allact$,} \\
\timeouts(\sta_{i+1}) & \text{if $\sta_{i+1} \in \allact$ and $\sta_i \not\in\allact$,} \\
\nodel & \text{if $\sta_{i+1} \not\in\allact$,}
\end{cases}
$$
\item If $t_i = \delays_i$ (a fixed-delay transition occurs),  
\[
\nxt(\sta_i,\delays_i;t_i) = \begin{cases}
\timeouts(\sta_{i+1}) & \text{if } \sta_{i+1} \in \allact
\text{,} \\
\nodel & \text{if }\sta_{i+1} \not\in\allact
\text{.}
\end{cases}
\]
\end{itemize}

We define the probability measure $\probm_{\fdC(\timeouts)}$ over the measurable space
$(\allruns, \sigmafield)$ where $\allruns$ is the set of all runs initiated in $\initstate$ and $\sigmafield$ is a $\sigma$-field over $\allruns$ generated by the set all of \emph{cylinder sets} of the form $\cylinderSet=\{(\sta'_0,\delays_0) t_0 \cdots \in \allruns \mid
   \forall i \leq n: \sta'_i = \sta_i, \forall i < n: t_i \in I_i \}$, here $\sta_0,\ldots,\sta_n\in S$  and $I_0,\ldots,I_{n-1}\subseteq \Rsetpo$ are intervals. 
Given such a cylinder set $\cylinderSet$, we define the~probability that a run of $\fdC(\timeouts)$ belongs to $\cylinderSet$ by
\[
   \probm_{\fdC(\timeouts)}^{\initstate}\left[\cylinderSet\right] :=  \mathbf{1}_{\sta_0=\initstate}
   \cdot \int_{t_0 \in I_0} \cdots \int_{t_{n-1} \in I_{n-1}}
   \prtr{\sta_0}{\delays_0}{\sta_1}{\de{t_0}} \cdot \cdots \prtr{\sta_{n-1}}{\delays_{n-1}}{\sta_n}{\de{t_{n-1}}}
\]
where $d_0=\timeouts(\sta_0)$ if $\sta_0 \in \allact$, and $\delays_0 =
\nodel$ otherwise; $\delays_{i+1} = \nxt(\sta_i,\delays_i,t_i)$ for each
$i\in\Nseto$; and the probability measure
$\prtr{\sta}{\delays}{\sta'}{\cdot}$ stands for the probability of moving to
$s'$ in a given time interval when staying in $(\sta,\delays)$,
i.e. for each $s,s' \in \states$, $\delays \in \Rsetp \cup \{\nodel\}$, and
an interval $[a,b]\subseteq \Rsetpo$
 $$
   \prtr{\sta}{\delays}{\sta'}{[a,b]} \; := \;
     \underbrace{\mathbf{1}_{a<\delays} \cdot \prob(\sta,\sta') \cdot (1 - e^{-\lambda(\min\{b,d\}-a)})}_{\text{exp-delay transition in the interval $[a,\min\{b,d\}]$}} \ + \ 
     \underbrace{\mathbf{1}_{\delays \in [a,b]} \cdot \trans(\sta,\sta')
       \cdot  e^{- \lambda d}.}_{\text{fixed-delay transition after time
         $d\in[a,b]$}}
   $$ 
The probability measure $\probm_{\fdC(\timeouts)}^{\initstate}$ then extends uniquely to all sets of $\sigmafield$.

\subsection{Proof of Proposition~\ref{prop:compute-cost}}\label{subsec:prob-spaceDTMC}

\begin{refproposition}{prop:compute-cost}
There is a polynomial-time algorithm that for a given fdCTMC
$\fdC(\timeouts)$, cost structure $\costFdC$ with goal states $\goalStates$,
and an approximation error $\eps > 0$ computes $\computedOutcome \in
\Rsetp \cup \{\infty\}$ and $p_s \in \Rsetp$, for each $s\in \goalStates$, such that 
$$\left\lvert \; \expred_{\fdC(\timeouts)} - \computedOutcome \;
\right\rvert  < \eps \quad \text{ and } \qquad
\left\lvert \probm_{\fdC(\timeouts)}(\first{G}{s}) - p_s \right\rvert < \eps$$
where $\first{G}{s}$ is the set of runs that reach $s$ as the first state of $\goalStates$ (after at least one transition).
\end{refproposition}
\begin{proof}
This proposition follows directly from results presented in \cite{DBLP:journals/pe/Lindemann93}. For details how the transition probabilities and costs of embedded DTMC are computed, please refer to Appendix~\ref{app:algs}. 
\end{proof}

\subsection{Definition of probability space for DTMDP}\label{subsec:prob-spaceDTMDP}
Every strategy $\sigma$ uniquely determines a probability space $(\Runs_\mdp,\sigmafield_\mdp,\probm_{\mdp(\sigma)})$, where $\Runs_\mdp$ is the set of all \emph{runs}, i.e. infinite alternating sequences of vertices and actions $\vtx_0 a_1 \vtx_1 a_2 \vtx_2\cdots \in (\mlocs\cdot \macts)^{\omega}$; $\sigmafield_\mdp$ is the sigma-field generated by all sets of runs of the form $\{\run \mid \text{$\run$ has prefix $\hist$}\}$ for all \emph{finite paths} $\hist$, i.e. prefixes of runs ending with a vertex; and $\probm_{\mdp(\sigma)}$ is the unique probability measure such that for every finite path $\hist = \vtx_0a_1\vtx_1 \dots a_n \vtx_n$ it holds $$\probm^{\initvertex}_{\mdp(\sigma)}(\{\run \mid \text{$\run$ has prefix $\hist$}\})= \begin{cases} 0 & \text{if } \sigma(\vtx_{i-1})\neq a_i \text{ for some } 1 \leq i \leq n,\\ \prod_{i=1}^{n} \mtran(\vtx_{i-1},a_i)(\vtx_i) & \text{otherwise}. \end{cases}$$
\newcommand{\fdtext}{\textrm{ft}}
\newcommand{\Sink}{K}
\newcommand{\yetUnknownStepsBound}{XXX}
\newcommand{\timeoutsStable}{\timeouts^*}
\newcommand{\timeoutsInfl}{\timeouts'}
\newcommand{\sinkVal}[2]{\overline{V}^{#1}_{#2}}
\newcommand{\unorm}[1]{||#1||_{\infty}}
\newcommand{\A}{\mathcal{A}}
\newcommand{\transientni}{\brambory'}

\newcommand{\delayBound}{d}
\newcommand{\lowerDelayBound}{\underline{\delayBound}}
\newcommand{\pmf}{f}
\newcommand{\poiss}{Po}
\newcommand{\statesM}{S'}
\newcommand{\e}{e}
\newcommand{\expect}{E}
\newcommand{\tempCost}{c}
\newcommand{\jumpChain}{X}
\newcommand{\timeChain}{Y}
\newcommand{\hit}{H}
\newcommand{\costFdCA}{\costFdC^{-}}
\newcommand{\costFdCB}{\costFdC^{--}}
\newcommand{\costFdCC}{\overline{\costFdC}}
\newcommand{\revisitVar}{\#}
\newcommand{\revisitSink}{\#^{\statesM\setminus \Sink}}
\newcommand{\goodStatesTemp}{L}
\newcommand{\upperSinkValueBound}{\lambda \delayBound \cdot \pmf_{\poiss(\lambda \delayBound)}(\geq 1) \cdot \Big (\Value{\contMdp} + \eps + \maxRew \cdot \Big (2 +\frac{2}{\lambda} \Big) \Big)}
\newcommand{\impulseCosts}{\alpha}

\newcommand{\tempSet}{U}

\section{Proofs for Section~\ref{sec:results-single}}

\subsection{Necessary definitions}
We denote $\Nseto$ the natural numbers with zero, i.e. $\Nset \cup \{0\}$.

\begin{definition}
Let $\fdC = (\states, \lambda, \prob, \allact, \trans, \initstate)$ be a fdCTMC structure and $\costFdC = (\goalStates, \rateRew, \impRewExp, \impRewFix)$ be the cost structure over $\fdC$. We define the \emph{minimal cost, maximal cost} and \emph{minimal branching probability} as
\begin{itemize}
	\item $\maxRew = \max \{ \rateRew(\sta) \mid \sta \in \states \} \cup \big \{ \impRew_i(\sta,\sta') >0 \mid \sta,\sta' \in \states \text{ and } i \in \{\prob, \trans \} \big \} $
	\item $\minRew = \min \{ \rateRew(\sta) \mid \sta \in \states \} \cup \big \{ \impRew_i(\sta,\sta') >0 \mid \sta,\sta' \in \states \text{ and } i \in \{\prob, \trans \} \big \} $, and
	\item $\minPst = \min \{ \prob(\sta,\sta') >0 \mid \sta,\sta' \in \states \} \cup \{ \trans(\sta,\sta') >0 \mid \sta,\sta' \in \states \} $, respectively.
\end{itemize}
\end{definition}
We will use these constants also for DTMDPs in this appendix. All DTMDPs in appendix were created of some given fdCTMC $\fdC$. Thus $\minRew, \maxRew$ and $\minPst$ are defined for all DTMDPs according the source fdCTMC.  

Let $(\Omega , \mathcal{F}, \mathcal{P})$ be arbitrary probability space. We define conditional probability $\mathcal{P} (A,B)$ for some events $A,B \in \mathcal{F}$ as 
$\mathcal{P} (A,B) = \frac{\mathcal{P}(A \cap B)}{\mathcal{P}(B)}$. Moreover let $G$ be random variable. The \emph{expected value} for discrete random variable $G$ denoted as $\expected[G]$ is $\expected[G] = \sum_{g \in Im(G)} g \cdot \mathcal{P}(G=g)$, where $Im$ is function that returns and image of a function. Finally the \emph{conditional expected value} of $G$ conditioned by an event $A \in \mathcal{F}$ as $\expected[G\mid A] = \sum_{g \in Im(G)} g \cdot \mathcal{P}(G=g \mid A)$. Standard definition is also used to define $\expected$ for continuous random variables. 

\begin{definition}
For a DTMDP $\contMdp$ and set $\tempSet \subseteq \statesM$ we define the series of random variables $(\hit^{\tempSet })_n$, where for each $n \in \Nset$ the $\hit^{\tempSet}_n:(\statesM\cdot \macts)^{\omega} \to \Nseto$ gives the jump of the $n$th visit of $\tempSet$ before reaching $\goalStates$ (stopping time), i.e.
\begin{align*}
\hit^{\tempSet}_n(\sta_0 a_1 \sta_1 a_2 \sta_2\cdots)  
&=\begin{cases}
i_n &\text{if }i_n < l,\\
l &\text{if }i_n \geq l,\\
\infty & \text{if $i_n = l = \infty$,}
\end{cases}
\end{align*}
where $l\geq0$ is minimal such that it holds $\sta_l \in \goalStates$ and $i_n\geq0$ is minimal such that there are $ i_1, \ldots i_{n-1}$ such that for each $j,k \leq n$ it holds that $\sta_{i_j} \in \tempSet$ and if $j\neq k $ then $ i_j \neq i_k $. If no such $i_n$ or $l$ exist we define them to be equal to $\infty$. 
\end{definition}

\begin{definition}
Let $\mdp$ be a DTMDP. We define a \emph{jump chain} for $\mdp$ as a series of random variables $(\jumpChain^{\mdp})_n$ where for all $n \in \Nseto$, $\jumpChain^{\mdp}_n:(\statesM\cdot \macts)^{\omega} \to \statesM$ and 
$$\jumpChain^{\mdp}_n(\sta_0 a_1 \sta_1 a_2 \sta_2\cdots)= \sta_n.$$
We define variable $\jumpChain^{\mdp}_{\infty}$ as $\jumpChain^{\mdp}_{\infty}: (\statesM\cdot \macts)^{\omega} \to \{\perp \}$.
\end{definition}

\begin{definition}
Let $\mdp$ be a DTMDP. For sets $U,\overline{U} \subseteq \statesM$ we define random variable $\revisitVar^{\overline{U}}_U:(\statesM\cdot \macts)^{\omega} \to \Nseto$ giving each run \emph{number of visits of states from $U$ before reaching any state from $\overline{U}$}, i.e. 
$$\revisitVar^{\overline{U}}_{U} (\omega)= \sum_{i=0}^{\hit_1^{\overline{U}}} 1_{\jumpChain_{i}^{\mdp} \in U},$$
where $1$ is indicator function.
\end{definition}

Let $\mdp$ be a DTMDP and $\timeouts$ be a delay function. For each $i \in \Nseto$ and $\sta \in \statesM$ (abusing notation) we define "delay function" $\timeouts[s/\delayBound/i]$ that assigns to every state $\sta' \in \statesM \setminus \{\sta\}$ delay $\timeouts(\sta')$, but it assigns to state $\sta$ delay $\delayBound$ until it is reached $i$ times and delay $\timeouts(\sta)$ afterwards. Moreover we define $\timeouts[s/\delayBound]$ as $\timeouts[s/\delayBound/1]$.

\subsection{Proof of Lemma~\ref{lem:finite}}
\begin{reflemma}{lem:finite}
	For any delay functions $\timeouts, \timeouts'$ we have $\expred_{\fdC(\timeouts)} = \infty$ if and only if $\expred_{\fdC(\timeouts')} = \infty$.
\end{reflemma}
\begin{proof}
Note that the only way how $\expred_{\fdC(\timeouts)} = \infty$ for delay function $\timeouts$ is to reach bottom strongly connected component (BSCC) in $\fdC(\timeouts)$ with positive probability. Runs that reach such BSCC stay there forever and incur infinite cost. Observe that for each $\sta \in \states$ and any delay function $\timeouts$ it holds that $0< \timeouts(\sta) < \infty$. Thus there is positive probability of reaching BSCC in $\fdC(\timeouts)$ for some delay function $\timeouts$ if only if there is positive probability to reach exactly same BSCC in $\fdC(\timeouts')$ also using any other delay function $\timeouts'$.
\end{proof}

\subsection{Proof of Proposition~\ref{prop:mdp-formulation}}

\begin{refproposition}{prop:mdp-formulation}
	For any delay function $\timeouts$ it holds 
	$\expred_{\fdC(\timeouts)}=\expred_{\contMdp(\timeouts)}$. Hence, 
	$$\Value{\fdC} = \Value{\contMdp}.$$
\end{refproposition}

\begin{proof}
	
	Let us fix $\timeouts$ and let $R$ and $R'$ denote the set of runs in $\contMdp(\timeouts)$ and $\fdC(\timeouts)$, respectively.
	Note that for each $n\in\Nset$, both $R$ and $R'$ can be partitioned into countable collections of sets $\mathcal{R}_n = \{ R_{s_0 \cdots s_n} \mid s_0,\ldots,s_n \in\states \}$ and $\mathcal{R}'_n = \{ R'_{s_0 \cdots s_n} \mid s_0,\ldots,s_n \in\states \}$ where
	\begin{align*}
		R_{s_0 \cdots s_n} &:=
		\{s'_0 s'_1 \cdots \mid \forall i \leq n : s'_i = s_i\} \\
		R'_{s_0 \cdots s_n} &:=
		\{(s'_0,d_0)t_0 \cdots \mid s'_0 = s_0, \exists i_1 < \cdots < i_n: \\
		& \qquad\qquad
		\forall i>0 \text{ such that } i = i_k \text{ for some $k$}: s'_{i} = s_k, s'_{i} \not\in\allact \lor d_{i-1} = t_{i-1} \lor s'_{i} \in \goalStates, \\ 
		& \qquad\qquad
		\forall i>0 \text{ such that } i \neq i_k \text{ for any $k$}: s_i \in \allact \setminus \goalStates, d_{i-1} \neq t_{i-1} \}
	\end{align*}
	
	Intuitively, the indices $s_0,\cdots,s_n$ in the latter definition are the states where fixed-delay transitions are not active or become newly active, i.e. states where the ``big steps'' that are simulated by single steps in $\contMdp(\timeouts)$ start.
	For each $n\in\Nset$, let $W_n$ and $W'_n$ denote the cost accumulated in first $n$ steps of $\contMdp(\timeouts)$ and first $n$ ``big steps'' of $\fdC(\timeouts)$, respectively, such that cost stops being accumulated when $\goalStates$ is reached. 
	
	By a straightforward induction on $n$, it is easy to show that
	\begin{align*}
		\probm_{\contMdp(\timeouts)}
		\left[ R_{s_0\cdots s_n} \right]
		\;&= \;
		\probm_{\fdC(\timeouts)}
		\left[ R_{s_0\cdots s_n} \right]
		\qquad\qquad \text{for all sequences of states $s_0\cdots s_n$, and}
		\\
		\expected^{\initstate}_{\contMdp(\timeouts)}
		\left[ W_n \right]
		\;&= \;
		\expected^{\initstate}_{\fdC(\timeouts)}
		\left[ W'_n \right]
		.
	\end{align*}
	
	\noindent
	This implies that, denoting by $N$ and $N'$ the set of all runs of $\fdC$ and $\contMdp$ that do not reach the target, 
	$$\probm_{\fdC(\timeouts)}[N] > 0 \qquad \Longleftrightarrow \qquad \probm_{\contMdp(\timeouts)}[N'] > 0.$$
	Furthermore, we get the result as 
	\begin{align*}
		\expred_{\contMdp(\timeouts)}
		&\;=\;
		\lim_{n\to\infty}
		\expected^{\initstate}_{\contMdp(\timeouts)}
		\left[ W_n \right] \\
		&\;=\;
		\lim_{n\to\infty}
		\expected^{\initstate}_{\fdC(\timeouts)}
		\left[ W'_n \right] \\
		&\;=\;
		\expred_{\fdC(\timeouts)}. 
	\end{align*}
	\qed
\end{proof}

\subsection{Proof of Lemma~\ref{lem:perturbation-error}}

In the following we denote by $\unorm{\mathbf{x}}$ a \emph{uniform norm} of a vector $\mathbf{x}$, i.e. the maximal absolute value of a component of $\mathbf{x}$. This norm induces a natural matrix norm $\unorm{\cdot}$, where for a real matrix $X$ the number $\unorm{X}$ represents the maximal absolute row sum of $X$.

\begin{reflemma}{lem:perturbation-error}
	Let $\alpha\in[0,1]$ and let, $\timeouts'$ be a delay function that is $\alpha$-bounded by another delay function $\timeouts$. If $\alpha \leq \frac{1}{2\cdot\CostBound{\mcost,\timeouts}\cdot|\brambory|}$, then
$$\expred_{\contMdp(\timeouts')}   
\; \leq \; 
\expred_{\contMdp(\timeouts)} + 2\cdot \alpha\cdot \CostBound{\msteps,\timeouts} \cdot (1 + \CostBound{\mcost,\timeouts}\cdot |\brambory|).$$ 
\end{reflemma}
\begin{proof}
 Let $\timeouts$, $\timeouts'$ be delay functions satisfying the assumptions of the lemma. 
 If both functions incur an infinite expected cost, the lemma clearly holds. Otherwise, due to Lemma~\ref{lem:finite} both functions incur finite cost.
 We start the proof by expressing the values $\expred_{\mdp(\timeouts')}^{\mcost}$ and $\expred_{\mdp(\timeouts)}^{\mcost}$ as solutions of certain systems of linear equations. To achieve this, we consider discrete time absorbing Markov chains $\A_{\timeouts}$ and $\A_{\timeouts'}$ obtained by fixing strategies $\timeouts$ and $\timeouts'$ in $\mdp$, respectively, and replacing all transitions  outgoing from states in $G$ with self-loops on these states. Note that due to both functions incurring finite cost the set $G$ is exactly the set of absorbing states of both $\A_{\timeouts}$ and $\A_{\timeouts'}$.

Let $\transientni=\brambory\setminus G$ be the set of transient states of $\A_{\timeouts}$ (and hence also of $\A_{\timeouts'}$) and let $Q_{\timeouts},Q_{\timeouts'}\colon \transientni\times \transientni \rightarrow [0,1]$ be the substochastic matrices encoding transitions in transient parts of $\A_{\timeouts}$ and $\A_{\timeouts'}$, respectively. That is, $Q_{\timeouts}$ is an $|\transientni|\times|\transientni|$ substochastic matrix with $Q_{\timeouts}(s,t)=\mtran(s,\timeouts(s))(t)$, for all $s,t\in \transientni$, and similarly for $Q_{\timeouts'}$. Next, we denote by $f_{\timeouts},f_{\timeouts'}\colon \transientni \rightarrow \Rset_{>0}$ the $|\transientni|$-dimensional vectors of one-step costs incurred with strategies $\timeouts$ and $\timeouts'$ in $\mdp$ and $\matchmdp$, respectively. That is, $f_{\timeouts}(s)=\mcost(s,\timeouts(s))$ for all $s\in \transientni$, and similarly for $f_{\timeouts'}$. By our assumptions, we have $Q_{\timeouts'}=Q_{\timeouts}+\Delta Q$, where $\Delta Q$ is a matrix whose each entry has absolute value bounded by $\alpha$, and $f_{\timeouts'} = f_{\timeouts} + \Delta f$, where $\Delta f$ is a vector whose each entry is $\leq \alpha$.

From standard results on Markov models with total accumulated reward criterion (e.g.~\cite[Theorem 7.1.3]{Puterman:book}) we have that the $|S''|$-dimensional vector $(\expred_{\mdp[s](\timeouts)}^{\mcost})_{s\in S''}$ is equal to the unique solution $\mathbf{x}$ of the following system of linear equations:
\[
(I-Q_{\timeouts})\cdot\mathbf{x} = f_{\timeouts}.
\]

\noindent
(The uniqueness of the solution comes from $\A_{\timeouts}$ being absorbing, which means that $(I-Q_{\timeouts})$ is invertible~\cite{KS:book}).
Similarly, the vector $(\expred_{\mdp[s](\timeouts')}^{\mcost})_{s\in S''}$ is the unique solution $\mathbf{y}$ of the following system of equations:
\[
(I-Q_{\timeouts} - \Delta Q)\cdot\mathbf{y} = f_{\timeouts} + \Delta f.
\]

\noindent
Hence, it suffices to obtain a bound on the maximal positive component of $\mathbf{x}-\mathbf{y}$. 
We first replace all negative components of $\Delta f$ with zeroes and thus obtain a new vector $\Delta f'$. We claim that the unique solution $\mathbf{y'}$ to the system
\[
(I-Q_{\timeouts} - \Delta Q)\cdot\mathbf{y'} = f_{\timeouts} + \Delta f'.
\]
satisfies $\mathbf{y'}\geq \mathbf{y}$. Indeed, this is because $\mathbf{y'} = (I-Q_{\timeouts} - \Delta Q)^{-1}\cdot (f_{\timeouts} + \Delta f')\geq  (I-Q_{\timeouts} - \Delta Q)^{-1}\cdot (f_{\timeouts} + \Delta f) = \mathbf{y}$, the middle inequality following that the matrix $(I-Q_{\timeouts} - \Delta Q)^{-1}$, so called fundamental matrix of a Markov chain $\A_{\timeouts'}$, is always non-negative. So to prove the lemma it suffices to give a bound on $\unorm{\mathbf{x}-\mathbf{y'}}$.
Note that $\unorm{\Delta f'}\leq \alpha$.

We use standard results on stability of perturbed systems of linear equations. From, e.g. proof of Theorem~3 in~\cite{IK:book} it follows that 
\begin{equation}
\label{eq:perturbation-1}
\unorm{\mathbf{x}-\mathbf{y'}} \leq \frac{\unorm{(I-Q_{\timeouts})^{-1}}}{1-\unorm{(I-Q_{\timeouts})^{-1}}\cdot \unorm{\Delta Q}}\cdot \left(\unorm{\Delta f'} + \unorm{\mathbf{x}}\cdot \unorm{\Delta Q} \right).
\end{equation}

\noindent
We know that $\unorm{\Delta Q} \leq \alpha\cdot |\brambory|$ and $\unorm{\Delta f'} \leq \alpha$, and from the definition of $\mathbf{x}$ we have that $\unorm{\mathbf{x}}= \CostBound{\mcost,\timeouts}$. Moreover, as noted above, $(I-Q_{\timeouts})^{-1}$ is the {fundamental matrix} of the absorbing Markov chain $\A_{\timeouts}$, i.e. a matrix such that the entry $(I-Q_{\timeouts})^{-1}(s,t)$ is equal to the expected number of visits to a transient state $t$ when starting in a transient state $s$ of $\A_{\timeouts}$. In particular, sum of each row of $(I-Q_{\timeouts})^{-1}$ is bounded by $\CostBound{\msteps,\timeouts}$, and thus $\unorm{(I-Q_{\timeouts})^{-1}}\leq \CostBound{\msteps,\timeouts}$. Plugging these (in)equalities into~\eqref{eq:perturbation-1} we get

\begin{align*}
\unorm{\mathbf{x}-\mathbf{y}} &\leq \frac{\CostBound{\msteps,\timeouts}}{1-\CostBound{\msteps,\timeouts}\cdot \alpha\cdot |\brambory|}\cdot  \left(\alpha + \CostBound{\mcost,\timeouts}\cdot \alpha \cdot |\states'| \right) 
\\
&\leq 2\cdot \CostBound{\msteps,\timeouts} \cdot (\alpha +\CostBound{\mcost,\timeouts}\cdot \alpha\cdot |\brambory|),
\end{align*}

\noindent
the last inequality following from our assumption that $\alpha\leq 1/(2\cdot \CostBound{\msteps,\timeouts}\cdot |\brambory|)$. The lemma easily follows. \qed
\end{proof}

\newcommand{\matrixExponential}{P''}
\newcommand{\upperDelayBoundOne}{\frac{\CostBound{\mcost,\timeouts}}{\minPst^{|\allact|} \cdot \minRew}}
\newcommand{\upperDelayBoundTwo}{zatialNic}
\newcommand{\errorBoundCuttingUp}{nic}

\subsection{Proof of Lemma~\ref{lem:mesh-error}} \label{subs:mesh-error}

\begin{reflemma}{lem:mesh-error}
There are positive numbers $\discconst,\cutconst\in \exp(\size{\fdC}^{\mathcal{O}(1)})$ computable in time polynomial in $\size{\fdC}$ such that the following holds for any $\alpha\in[0,1]$ and any delay function $\timeouts$: If we put $$\delta \; := \; \alpha/\discconst \quad \text{and} \quad \dmax \; := \; |\log(\alpha)|\cdot \cutconst \cdot\CostBound{\mcost,\timeouts},$$ then $\paramspace(\delta,\dmax)$ contains a delay function which is $\alpha$-bounded by $\timeouts$.
\end{reflemma}

The proof of the lemma proceeds in 2 phases. First we show how to set $\delta$ in such a way that for every delay function $\timeouts$ there is a function using only multiples of $\delta$ which is $\alpha/2$-bounded by $\timeouts$. In the second phase we show how to choose $\dmax$ such that for each delay function $\timeouts$ (and in particular each delay function that uses only multiples of $\delta$) there is a function $\alpha/2$-bounded by $\timeouts$ that does not use action greater than $\dmax$.

We will need the following definition.

\begin{definition}
\label{def:subordinated}
For a fdCTMC structure $(\states, \lambda, \prob, \allact, \trans, \initstate)$, state $\sta \in \statesM$, set of states $\statesM \subseteq \states$ and delay $\delayBound > 0$ we define a fdCTMC $\fdC'[\sta](d)$, where $\fdC' = (\states \cup \{\overline{\sta} \}, \lambda, \prob', \allact, \trans', \overline{\sta})$, $\overline{\sta}$ is new initial state, that was before not in $\states$, and for all $\sta, \sta ' \in \states$ holds
$$\prob'(\sta', \sta'') = 
  \begin{cases}
	\prob(\sta, \sta'')& \text{if $\sta' = \overline{\sta}$,} \\  
	\prob(\sta', \sta'')& \text{if $\sta' \not\in \statesM$,} \\
	1& \text{if $\sta' = \sta'' \in \statesM$,} \\
	0& \text{otherwise,}
  \end{cases}$$
and
$$\trans'(\sta', \sta'') = 
  \begin{cases}
	\trans(\sta, \sta'')& \text{if $\sta' = \overline{\sta}$,} \\
	\trans(\sta', \sta'')& \text{if $\sta' \not\in \statesM$,} \\
	1& \text{if $\sta' = \sta'' \in \statesM$,} \\
	0& \text{otherwise.}
  \end{cases}$$
\end{definition}

\begin{lemma}
There is a number $D_1 \in \exp(\size{\fdC}^{\mathcal{O}(1)})$ computable in time polynomial in $\size{\fdC}$ such that for every $\alpha$ and every delay function $\timeouts$ the following holds: If we put $\delta=\alpha/D_1$, then $\{\timeouts \mid \forall s\in\brambory \; \exists k\in \Nset: \timeouts(s) = k\delta \}$ contains a delay function which is $(\alpha/2)$-bounded by $\timeouts$.
\end{lemma}
\begin{proof}
Let us fix arbitrary strategy $\timeouts$ of $\mdp$. We show that for each state $\sta \in \statesM$ there is strategy $\matchtimeouts \in \{\timeouts \mid \forall s\in\brambory \; \exists k\in \Nset: \timeouts(s) = k\delta \}$ such that for all $\sta, \sta' \in \statesM$ it holds that
$$| \mtran(s,\timeouts(\sta))(s') - \mtran(s,\matchtimeouts(\sta))(s') | \leq \lambda \cdot \delta.$$
Please observe that to compute $\mtran(s,d)(s')$ we can perform the transient analysis of $\fdC'[\sta](d)$ in time $d$ and then perform in zero time fixed delay transitions. The $\fdC'[\sta](d)$ performs as continuous time Markov chain (CTMC) until time $d$, thus we can use methods to compute the transient analysis of CTMC in time $d$. One of them is to use matrix exponential $\matrixExponential(t)= e^{tQ}$ \cite{Norris:book} for a $Q$-matrix for $\fdC'[\sta](\timeouts(\sta))$. It is known that $\matrixExponential(t)_{s,s'}$ is the probability of moving from state $s$ to $s'$ exactly in time $t$. Moreover for $t \geq 0$ it holds that
$$\frac{d \; \matrixExponential(t)}{d t} = \matrixExponential(t) \cdot Q.$$
Since the sum of positive entries in each line in $Q$ is at most the uniformization rate $\lambda$ and that $\matrixExponential(t)$ for every $t \geq 0$ is correct transition kernel it holds that for all $\sta, \sta' \in \statesM$ and $t \geq 0$
$$(\matrixExponential(t) \cdot Q)_{\sta,\sta'} \leq \lambda.$$
Thus the transient probability at time $t > 0$ in $\fdC'[\sta](\timeouts(\sta))$ can change at most $\lambda \cdot \delta$ if we move time $\timeouts(\sta)$ by $\delta$. Since $\mtran$ is also transition kernel it does only convex combination of transient probability at time $\timeouts(\sta) \pm \delta $ thus it cannot increase the error. Altogether we have that for all $\sta, \sta' \in \statesM$
\begin{equation}
\label{eq:disc-err-1}
| \mtran(s,\timeouts(\sta))(s') - \mtran(s,\matchtimeouts(\sta))(s') | \leq \lambda \cdot \delta.
\end{equation}

Hence, the strategy $\timeouts^\star$ we seek can be obtained by simply rounding the components $\timeouts$ to the nearest positive integer multiple of $\delta$.

Now we show that for each state $\sta \in \statesM$ it holds
$$| \mcost(\sta,\timeouts(\sta)) - \mcost(s,\matchtimeouts(\sta))| \leq 2 \cdot \lambda \cdot \delta \cdot \maxRew + \delta \cdot \maxRew \leq 2 \cdot (\lambda +1) \cdot \delta \cdot \maxRew.$$
Recall that $\mcost(s,\delays)=\expred_{\fdC[s](d)}^{\costFdC[\brambory]}$, i.e. $\mcost(s,\delays)$ is the expected cost accumulated until hitting state from $\statesM$. This cost can be divided into expected cost accumulated by rate costs and by expected cost accumulated by impulse costs. 
First we bound the change on expected rate cost until hitting $\statesM$. Let $O$ be random variable assigning run the time until hitting state from $\statesM$. Assume that the delay $\timeouts(\sta)$ was increased by $\delta$ (if it was decreased the argument is symmetric- we only start with $\matchtimeouts(\sta)$). Then the time until hitting $\statesM$ by each run is increased at most by $\delta$. Thus the maximal change in expected time until hitting $\statesM$ is $\delta$. The maximum rate cost in every state is bounded by $\maxRew$. Thus the maximal change in expected rate cost accumulated until time $\timeouts(\sta)$ if we increase it by $\delta$ is $\delta \cdot \maxRew$.

Similarly, if the delay of $\timeouts(\sta)$ is increased by $\delta$ (if it was decreased the argument is symmetric- we just start with $\matchtimeouts(\sta)$) we give every run a chance to do some more steps, what can increase the impulse cost. The increase in expected impulse cost until hitting state from $\statesM$ is bounded by the change in expected number of exponential steps from $\sta''$ in time $\timeouts(\sta)$ plus change in expected number of fixed-delay steps in time $\timeouts(\sta)$ multiplied by the maximum impulse cost. The expected number of exponential steps with rate $\lambda$ in time $\timeouts(\sta)$ is mean of Poisson distribution with parameter $\lambda \cdot  \timeouts(\sta)$, what equals $\lambda \cdot \timeouts(\sta)$. Thus if the $\timeouts(\sta)$ was increased by $\delta$ the increase in expected number exponential steps is $\lambda \cdot \delta$. Please observe, that only one fixed-delay step can happen until state from $\statesM$ is reached, because all fixed-delay transitions go to states in $\statesM$. The fixed-delay transition fires in time $\timeouts(\sta)$ if the run is still in state from $\allact$. If the time is increased by $\delta$ the probability of run to be in states $\allact$ in time $\timeouts(\sta)$ can change at most by $\lambda \cdot \delta$ (from bounded change in transient probability proved in first paragraph of this proof). Thus the maximal increase in expected impulse cost until reaching $\statesM$ from $\sta$ is $\maxRew \cdot (\lambda \cdot \delta + \lambda \cdot \delta)$.
Altogether we have that 
\begin{equation}
\label{eq:disc-err-2}
| \mcost(\sta,\timeouts(\sta)) - \mcost(s,\matchtimeouts(\sta))| \leq 2 \cdot \lambda \cdot \delta \cdot \maxRew + \delta \cdot \maxRew \leq 2 \cdot (\lambda +1) \cdot \delta \cdot \maxRew.
\end{equation}

From~\eqref{eq:disc-err-1} and~\eqref{eq:disc-err-2} it follows that to get the lemma it suffices to put $D_1=\max\{\lambda,2(\lambda+1)\cdot\maxRew\}$.

\end{proof}

\begin{lemma}
There is a number $D_2 \in \exp(\size{\fdC}^{\mathcal{O}(1)})$ computable in time polynomial in $\size{\fdC}$ such that for every $\alpha$ and every delay function $\timeouts$ the following holds: If we put $\dmax := |\log(\alpha)|\cdot \cutconst \cdot\CostBound{\mcost,\timeouts}$, then $\{\timeouts \mid \forall s\in\brambory \; \timeouts(s)  \leq \dmax \}$ contains a delay function which is $(\alpha/2)$-bounded by $\timeouts$.
\end{lemma}
\begin{proof}
Let us fix arbitrary strategy $\timeouts$ of $\mdp$ and $\sta \in \statesM$. We use the $\fdC'[\sta](\timeouts(\sta))$ from Definition \ref{def:subordinated} (that can be used to compute $\mtran(s,\timeouts(\sta))$ using transient analysis of continuous time Markov chains in time $\timeouts(\sta)$) to generate a directed graph $Gr(\states \cup \{\overline{\sta} \},tra)$, where $(\sta,\sta') \in tra $ iff $\prob'(\sta,\sta') > 0$. Let $st \subseteq \states \cup \{\overline{\sta} \} $ be set of states that are reachable in $Gr$ from $\overline{\sta}$. We define directed graph $gr=(st,tra')$ to be induced graph by $st$ from  $Gr$. We distinguish two cases:
\begin{itemize}
	\item $\fdC'[\sta](\timeouts(\sta))$ contains bottom strongly connected component (BSCC) $O \subseteq \states \cup \{\overline{\sta} \}$ in $Gr$, such that $O \cap \statesM = \emptyset$, and
	\item all BSCCs in $Gr$ are self-loops and contain only states from $\statesM$.
\end{itemize}
Observe that other cases cannot happen from definition of $\fdC'[\sta](\timeouts(\sta))$. First we take care of first item. We show that it  must hold 
\begin{equation}\timeouts(\sta) \leq \upperDelayBoundOne,\label{eq:odrezani-1}\end{equation} which gives us a simple bound on $\dmax$. The minimal probability to reach $O$ is the $\minPst^{|\allact|}$ because the minimal path to $O$ in $Gr$ can have length at most $|\allact|$ and the minimal positive branching probability is $\minPst$. Observe that if $O$ is reached by a run it will stay there until the fixed-delay transition is fired. The expected rate cost accumulated by paths that reach $O$ is at least $\timeouts(\sta) \cdot \minRew $. Hence we gradually get
\begin{align*}
\minPst^{|\allact|} \cdot \minRew \cdot \timeouts(\sta) &\leq \CostBound{\mcost,\timeouts} \\
\timeouts(\sta) &\leq \upperDelayBoundOne,
\end{align*}
as required.

From now we assume that $gr$ of $\fdC'[\sta](\timeouts(\sta))$ for state $\sta$ does contain only BSCCs reachable from $\sta$ that are self-loops and all contain only states in $\statesM$. We need to obtain a suitable bound on $\dmax$ for this case. Obviously by restricting strategies that do not have large delays we can only decrease the expected one step cost. So it suffices to construct $\dmax$ such that for each $\timeouts$ there is a strategy $\matchtimeouts \in\{\timeouts \mid \forall s\in\brambory \; \timeouts(s)  \leq \dmax \}$ such that for all $\sta' \in \statesM$ it holds that
$$| \mtran(s,\timeouts(\sta))(s') - \mtran(s,\matchtimeouts(\sta))(s') | \leq \alpha/2.$$

Obviously it holds that $| \mtran(s,\timeouts(\sta))(s') - \mtran(s,\matchtimeouts(\sta))(s') | \leq | 1 - \sum_{\sta' \in \statesM} \mtran(s,\matchtimeouts(\sta))(s') | $. Observe that $| 1 - \sum_{\sta' \in \statesM} \mtran(s,\matchtimeouts(\sta))(s') |$ denotes the probability that a state from $\statesM$ was not reached yet, i.e. that run in $\fdC'[\sta](\timeouts(\sta))$ is still in transient part of $\fdC'[\sta](\timeouts(\sta))$. Obviously as $\matchtimeouts(\sta)$ approaches infinity $| 1 - \sum_{\sta' \in \statesM} \mtran(s,\matchtimeouts(\sta))(s') |$ goes to zero. We will overestimate $| 1 - \sum_{\sta' \in \statesM} \mtran(s,\matchtimeouts(\sta))(s') |$ by overestimating the probability that any BSCC has not yet been reached until time $t >0$ in $\fdC'[\sta](\timeouts(\sta))$. We divide $t$ into steps of length $|\allact|/\lambda$. Using uniformization method we get:
\begin{align*}
\Big | 1 - \sum_{\sta' \in \statesM} \mtran(s,\frac{|\allact|}{\lambda})(s') \Big | &= \sum_{\sta' \in \statesM} \e^{-\lambda \frac{|\allact|}{\lambda}} \sum_{i=0}^{\infty} \prob'^{i}(\sta,\sta') \cdot \frac{(\lambda \frac{|\allact|}{\lambda})^{i}}{i!} \\
&= \sum_{\sta' \in \statesM} \e^{-|\allact|} \sum_{i=0}^{\infty} \prob'^{i}(\sta,\sta') \cdot \frac{|\allact|^{i}}{i!} \\
&\geq \sum_{\sta' \in \statesM} \e^{-|\allact|} \prob'^{|\allact|}(\sta,\sta') \cdot \frac{|\allact|^{|\allact|}}{|\allact|!} \\
&\geq \sum_{\sta' \in \statesM} \e^{-|\allact|} \prob'^{|\allact|}(\sta,\sta') \cdot \frac{|\allact|^{|\allact|}}{|\allact|^{|\allact|}} \\
&\geq \e^{-|\allact|} \minPst^{|\allact|} = \left(\frac{\minPst}{\e}\right)^{|\allact|}.
\end{align*}
The probability that in time $\matchtimeouts(\sta) = l \cdot |\allact|/\lambda$ that state from $\statesM$ has still not been reached is
\begin{align*}
\Big | 1 - \sum_{\sta' \in \statesM} \mtran(s,\matchtimeouts(\sta))(s') \Big | &\leq \left( 1- \left(\frac{\minPst}{\e}\right)^{|\allact|} \right) ^{l} \\
&= \left(1- \left(\frac{\minPst}{\e}\right)^{|\allact|} \right) ^{\frac{\matchtimeouts(\sta) \cdot \lambda}{|\allact|}} 
\end{align*}

To ensure that the right-hand side of the above inequality is $\leq \alpha/2$, it suffices to ensure that
\begin{equation}\label{eq:odrezani-2}\matchtimeouts(\sta) \geq \frac{|\log(\alpha/2)|\cdot |\states|\cdot e^{|\states|}}{\lambda\cdot \minPst^{|\states|}}.\end{equation}
Combining~\eqref{eq:odrezani-1} and~\eqref{eq:odrezani-2} we get that it suffices to put $D_2 = (|\states|\cdot e^{|\states|})/(\minPst^{|\states|}\cdot \min\{1,\lambda\}\cdot \min\{1,\minRew\})$.
\end{proof}

\newcommand{\revisitGoodStatesTemp}{\#^{\goodStatesTemp \cup \goalStates}}
\newcommand{\visits}{W}
\newcommand{\maxStepsGoodStates}{\overline{\visits}}
\newcommand{\stepsBound}{\frac{|\allact|}{(\minPst \cdot \e^{-\lambda \lowerDelayBound})^{|\allact|}} \cdot \frac{\maxValue{\contMdp}+\eps}{\minPst \cdot \minRew \cdot (1-\e^{-\lambda \lowerDelayBound})/\lambda}}
\newcommand{\epsSmall}{\eps}
\newcommand{\expLowerDelayBoundExpression}{1 + \frac{\epsSmall \cdot \minRew}{4 \cdot \lambda \cdot (\maxValue{\contMdp} + \epsSmall + \maxRew \cdot (2 + 2/\lambda))^2}}
\newcommand{\tempA}{A(\delayBound, \timeouts, \sta)}
\newcommand{\tempB}{B(\delayBound, \epsSmall)}
\newcommand{\tempBLower}{B(\lowerDelayBound, \epsSmall)}

\subsection{Proof of Lemma~\ref{lem:bounded-step-existence}}
\label{sec:bounded-step-existence}
\begin{reflemma}{lem:bounded-step-existence}
There is a positive number $\constFactor\in \exp(\size{\fdC}^{\mathcal{O}(1)})$ computable in time polynomial in $\size{\fdC}$ such that the following holds:
	for any $\eps'>0$, there is a globally $\eps'/2$-optimal delay function $\timeouts'$ with 
	\begin{equation}
	\label{eq:step-bound-app}
	\CostBound{\msteps,\timeouts'} \; \leq \; \frac{\CostBound{\mcost,\timeouts'}}{\eps'}\cdot \constFactor.
	\end{equation}
\end{reflemma}

We start the proof with a lemma which shows that for every $\eps>0$ there is a \emph{globally} $\eps$-optimal delay function in $\contMdp$.

\begin{lemma}
\label{lem:globally-e-optimal-existence}
For every $\eps>0$ there is a \emph{globally} $\eps$-optimal delay function in $\contMdp$.
\end{lemma}
\begin{proof}
Fix $\eps>0$. For every state $s$ there is trivially a delay function $\timeouts_{s}$ which is $(\eps/2)$-optimal \emph{in $s$}. Fix such an $(\eps/2)$-optimal function for each state $s$. Since there are only finitely many $s$, there is a number $\alpha>0$ such that \begin{equation}\label{eq:globaleps} 2\cdot \alpha\cdot \CostBound{\msteps,\timeouts_s} \cdot (1 + \CostBound{\mcost,\timeouts_s}\cdot |\brambory|) \leq \eps/2\end{equation} for all $s$. Moreover, from Lemma~\ref{lem:mesh-error} it follows that for each $s$ there is sufficiently small $\delta_s$ and sufficiently large $\dmax_s$ such that $\paramspace(\delta_s,\dmax_s)$ contains a function that is $\alpha$-bounded by $\timeouts_s$. Let $\delta=\min_{s\in \allact}\delta_s$ and $\dmax=\max_{s\in \allact}\dmax_s$. Then for each $s$, the set $\paramspace(\delta,\dmax)$ contains a delay function that is $\alpha$-bounded by $\timeouts_{s}$. From~\eqref{eq:globaleps} and Lemma~\ref{lem:perturbation-error} it follows that in the \emph{finite} DTMDP obtained from $\contMdp$ by restricting the set of actions to $\{k\delta\mid k\in \Nset,k>0,k\delta\leq \dmax\}\cup \{\infty\}$ the infimum expected cost achievable from a given state is at most $\eps/2$ away from the infimum cost achievable from this state in $\contMdp$. Moreover, in a finite DTMDP with expected total accumulated cost objective there is always a memoryless delay function\footnote{I.e. a delay function whose decision is based only on the current vertex, as defined in Section~\ref{sec-prelims}.} that is optimal in every vertex~\cite{Puterman:book}. Hence, $\paramspace(\delta,\dmax)$ contains a delay function that is globally $\eps$-optimal in $\mdp$.
\qed
\end{proof}

Recall that $\statesM=$ is the set of states of DTMDP $\contMdp$, i.e. the set of those states of fdCTMC structure $\fdC$ in which the fixed-delay is newly set or switched off, together with all goal states from $G$. An \emph{fd-skeleton} of $\contMdp$ is the directed graph $(\statesM,E)$ such that $(s,t)\in E$ if and only if $s\not \in G$ and $\trans(s,t)>0$. A \emph{sink} of the DTMDP $\contMdp$ is a bottom strongly connected component of its fd-skeleton, i.e. a set of states $\Sink\subseteq \statesM$ such that a) the subgraph of $(\statesM,E)$ induced by $\Sink$ is strongly connected, and b) whenever $(s,t)\in E$ and $s\in \Sink$, then also $t\in \Sink$.

Let $\timeouts$ be any delay function in $\contMdp$. We say that a sink $\Sink$ is \emph{bad} for $\timeouts$ if the following two conditions hold:
\begin{itemize}
\item $\impRewFix(s,t)=0$ for all $s,t \in \Sink$, and
\item $\timeouts(s) < \lowerDelayBound$ for all $s\in \Sink$, where $\lowerDelayBound$ is a suitable number defined below.
\end{itemize}
\noindent
A sink is \emph{good} for $\timeouts$ if it is not bad for $\timeouts$. 

Let $\lowerDelayBound$ be a number such that it holds
\begin{equation}\label{eq:dmin-def}\e^{\lambda \lowerDelayBound} = \expLowerDelayBoundExpression.\end{equation}
Note that using Taylor series for $\ln(1+z)= z -z^/2 +z^3/3- \cdots$ we get
$$\lowerDelayBound \leq \frac{\epsSmall \cdot \minRew}{4 \cdot \lambda^2 \cdot (\maxValue{\contMdp} + \epsSmall + \maxRew \cdot (2 + 2/\lambda))^2}.$$

\begin{lemma}
\label{lem:good-is-small-step}
Let $\eps>0$ be arbitrary, and let $\timeouts$ be any globally $\eps$-optimal delay function in $\contMdp$ such that all sinks of $\contMdp$ are good for $\timeouts$. Then for all pairs of states $s,t$ of $\contMdp$ it holds \begin{equation}\expected_{\contMdp[s](\timeouts)}[\revisit{t}]\leq \stepsBound \label{eq:good-step-bound}.\end{equation} 
\end{lemma}
\begin{proof}
First we define \emph{good states} and show that there is lower bound on cost paid in one step from a good state. Because we have upper bound on cost from any state $\maxValue{\contMdp}+\eps$ we cannot visit too many times good states before reaching goal state. We show that we have small expected number of steps from any state to reach a good state. Finally we put all the previous results together to show the final result.

We define \emph{good states} $\goodStatesTemp \subseteq \statesM \setminus \goalStates$ for $\contMdp$ and $\timeouts$ as $\goodStatesTemp = (\statesM \setminus (\allact \cup \goalStates)) \cup \{ \sta \in \statesM \setminus \goalStates \mid \timeouts(\sta) \geq \lowerDelayBound \text{ or exists $t \in \statesM$ such that } \impRewFix(s,t)> 0 \text{ and } \trans(s,t)> 0 \}.$ Now we show that there is lower bound on one step cost from good states, i.e. for all $\sta \in \goodStatesTemp$ holds $\mcost(\sta,\timeouts(\sta)) \geq a$ for some $a>0$. In states $\statesM \setminus \allact$ the minimal expected cost accumulated in one step is $\minRew /\lambda$, since the expected time to execute exponential transition is $1/\lambda$ and minimal rate cost is $\minRew$. The expected cost accumulated in one step from state $ \sta \in \{ \sta \in \statesM \setminus \goalStates \mid \timeouts(\sta) \geq \lowerDelayBound \}$ can be bounded as follows: 
\begin{align}
\mcost(\sta,\timeouts(\sta)) &\geq \minRew \cdot \text{(expected time spent in $\sta$ until the first transition is taken)} \nonumber\\
&= \minRew \cdot \int_{0}^{\timeouts(\sta)} t \cdot \lambda \cdot \e^{-\lambda t} dt + \timeouts(\sta) \cdot \e^{-\lambda \timeouts(\sta)} \nonumber\\
&= \minRew \cdot \Big( -\timeouts(\sta) \cdot \e^{-\lambda \timeouts(\sta)}  + \frac{1-\e^{-\lambda \timeouts(\sta)}}{\lambda} + \timeouts(\sta') \cdot \e^{-\lambda \timeouts(\sta)} \Big ) \nonumber\\
&= \minRew \cdot \frac{1-\e^{-\lambda \timeouts(\sta)}}{\lambda} \nonumber \\
&\geq \minRew \cdot \frac{1-\e^{-\lambda \lowerDelayBound}}{\lambda} \label{eq:delayBound}\\
&\geq \minRew \cdot \minPst \cdot \frac{1-\e^{-\lambda \lowerDelayBound}}{\lambda}, \nonumber
\end{align}
where the inequality \eqref{eq:delayBound} follows from the fact that $\sta \in \{ \sta \in \statesM \setminus \goalStates \mid \timeouts(\sta) \geq \lowerDelayBound \}$. Finally if $\sta \in \{ \sta \in \statesM \setminus \goalStates \mid \timeouts(\sta) < \lowerDelayBound \text{ and exists $t \in \statesM$ such that } \impRewFix(s,t)> 0 \text{ and } \trans(s,t)> 0 \}$ then it holds that 
\begin{align*}
\mcost(\sta,\timeouts(\sta)) &\geq \minRew \cdot \minPst \cdot \e^{-\lambda \timeouts(\sta)} \\
&\geq \minRew \cdot \minPst \cdot \e^{-\lambda \lowerDelayBound},
\end{align*}
where the first inequality follows from fact that probability of taking the transition with nonzero impulse cost (thus at least $\minRew$) is at least minimal branching probability $\minPst$ times the probability that fixed delay transition is taken what is $\e^{-\lambda \timeouts(\sta)}$. The last inequality follows from the fact that $\sta \in \{ \sta \in \statesM \setminus \goalStates \mid \timeouts(\sta) < \lowerDelayBound \text{ and exists $t \in \statesM$ such that } \impRewFix(s,t)> 0 \text{ and } \trans(s,t)> 0 \}$. Obviously for every $\sta \in \goodStatesTemp$ holds that
\begin{align*}
\mcost(\sta,\timeouts(\sta)) &\geq min \left\lbrace \minRew \cdot \minPst \cdot  \e^{-\lambda \lowerDelayBound}, 
\minRew / \lambda, \minRew \cdot \minPst \cdot \frac{1-\e^{-\lambda \lowerDelayBound}}{\lambda} \right\rbrace \\
&=\minRew \cdot \minPst \cdot \frac{1-\e^{-\lambda \lowerDelayBound}}{\lambda} .
\end{align*}
Now provide an upper bound $\maxStepsGoodStates$ on the expected number of steps to reach a good or goal state (i.e. $\goodStatesTemp \cup \goalStates$) from any $\sta \in \statesM $:
\begin{align}
&\max_{\sta \in \statesM}\expected_{\contMdp[\sta](\timeouts)}[\revisitGoodStatesTemp_{\statesM}] \leq \nonumber \\
&= \sum_{i=1}^{\infty} i \cdot |\allact| \cdot (1- \min_{ \sta \in \statesM}\probm_{\contMdp(\timeouts)}(\text{reach $\goodStatesTemp \cup \goalStates$ from $\sta$ in $|\allact|$ steps}))^{i-1} \\
&\quad\cdot \min_{\sta \in \statesM}\probm_{\contMdp(\timeouts)}(\text{reach $\goodStatesTemp \cup \goalStates$ from $\sta$ in $|\allact|$ steps}) \nonumber \\
&\leq \sum_{i=1}^{\infty} i \cdot |\allact| \cdot (1- (\minPst \cdot \e^{-\lambda \lowerDelayBound})^{|\allact|}))^{i-1} \cdot (\minPst \cdot \e^{-\lambda \lowerDelayBound})^{|\allact|} \label{eq:probBound}  \\
&= \frac{|\allact|}{(\minPst \cdot \e^{-\lambda \lowerDelayBound})^{|\allact|}}  = \maxStepsGoodStates \nonumber,
\end{align}
where \eqref{eq:probBound} follows from the fact that we can do at most $|\allact|$ steps avoiding to reach some of the $\goodStatesTemp \cup \goalStates$ with positive probability, by moving through states $\allact$ (states not belonging to $\allact $ are good states from definition of good states). This probability is at least the minimal branching probability $\minPst$ times probability $\e^{-\lambda \timeouts(\sta)}$ that fixed delay transition was fired for $\sta \in \allact \setminus \goodStatesTemp$, what is at least $\e^{-\lambda \lowerDelayBound}$, because $\timeouts(\sta) < \lowerDelayBound$ from definition of good states. 

To get the final result we have to define few random variables. For a DTMDP $\contMdp$ and set $\goodStatesTemp \subseteq \statesM$ we define the series of random variables $(\visits^{\tempSet})_n$, where for each $n \in \Nset$ the $\visits^{\goodStatesTemp}_n:(\statesM\cdot \macts)^{\omega} \to \Nseto$ gives the number of steps between $n$th and $(n+1)$st visit of $\goodStatesTemp$ before reaching $\goalStates$, i.e.
\begin{align*}
\visits^{\goodStatesTemp}_n(\sta_0 a_1 \sta_1 a_2 \sta_2\cdots)  
&=\begin{cases}
\hit^{\goodStatesTemp}_{n} &\text{if }n = 1,\\
\hit^{\goodStatesTemp}_{n}-\hit^{\goodStatesTemp}_{n-1} &\text{if }n > 1.
\end{cases}
\end{align*} 
Obviously it holds that
$$\revisit{\statesM} = \sum_{n=1}^{\infty} \visits^{\goodStatesTemp}_n$$ 
and for each $n \in \Nset$ and $\sta \in \statesM$ it holds that
\begin{align*}
\expected_{\contMdp[\sta](\timeouts)}[\visits^{\goodStatesTemp}_{n+1}] = \sum_{\sta' \in \goodStatesTemp} \probm_{\contMdp(\timeouts)}(\jumpChain^{\mdp}_{\hit^{\{ \goodStatesTemp \} }_n} = \sta') \cdot \expected_{\contMdp[\sta](\timeouts)}[\visits^{\goodStatesTemp}_{n+1} \mid \jumpChain^{\mdp}_{\hit^{\{ \goodStatesTemp \} }_n} = \sta'] \leq \maxStepsGoodStates \cdot \probm_{\contMdp(\timeouts)}(\jumpChain^{\mdp}_{\hit^{\{ \goodStatesTemp \} }_n} \in \goodStatesTemp), 
\end{align*}
where $\probm_{\contMdp(\timeouts)}(\jumpChain^{\mdp}_{\hit^{\{ \goodStatesTemp \} }_i} \in \goodStatesTemp)$ is probability of at least $n$ times reaching a state from $\goodStatesTemp$ before reaching any of the goal states $\goalStates$. We define a random variable $V$ denoting number of visits of $\goodStatesTemp$ before hitting $\goalStates$, i.e.
$$V = \sum_{i=0}^{\hit_{1}^{\emptyset}}1_{\jumpChain_{i} \in \goodStatesTemp},$$
where $1$ is indicator function. Obviously it holds that
\begin{align*}
\expected_{\contMdp[\sta](\timeouts)}[\visits^{\goodStatesTemp}_{n+1}] \leq \maxStepsGoodStates \cdot \probm_{\contMdp(\timeouts)}(\jumpChain^{\mdp}_{\hit^{\{ \goodStatesTemp \} }_n} \in \goodStatesTemp) = \maxStepsGoodStates \cdot \probm^{\sta}_{\contMdp(\timeouts)}(V \geq n). 
\end{align*}
Also it holds that 
$$\expected_{\contMdp[\sta](\timeouts)}[\visits^{\goodStatesTemp}_1] = \expected_{\contMdp[\sta](\timeouts)}[\revisitGoodStatesTemp_{\statesM}] \leq \maxStepsGoodStates = \maxStepsGoodStates \cdot \probm^{\sta}_{\contMdp(\timeouts)}(V \geq 0).$$
Altogether we have
\begin{align}
\expected_{\contMdp[\sta](\timeouts)}[\revisit{\statesM}] &= \sum_{n=1}^{\infty} \expected_{\contMdp[\sta](\timeouts)}[\visits^{\goodStatesTemp}_n]  \nonumber\\
&\leq \maxStepsGoodStates \cdot \probm^{\sta}_{\contMdp(\timeouts)}(V \geq n) \nonumber \\
&= \maxStepsGoodStates \cdot \expected_{\contMdp[\sta](\timeouts)}[V] \nonumber \\
&\leq \maxStepsGoodStates \cdot \frac{\maxValue{\contMdp}+\eps}{\minPst \cdot \minRew \cdot (1-\e^{-\lambda \lowerDelayBound})/\lambda} \label{eq:boundByValue} \\
&\leq \stepsBound, \nonumber
\end{align}
where \eqref{eq:boundByValue} follows from the fact that states $\goodStatesTemp$ cannot be visited too often because in each visit at least $\minPst \cdot \minRew \cdot (1-\e^{-\lambda \lowerDelayBound})/\lambda$ of cost is accumulated and the overall bound on cost is $\maxValue{\contMdp}+\eps$. Finally trivially for all $\sta,t \in \statesM$ holds $\expected_{\contMdp[\sta](\timeouts)}[\revisit{t}] \leq \expected_{\contMdp[\sta](\timeouts)}[\revisit{\statesM}]$.
\end{proof}
The previous lemma shows that to prove Lemma~\ref{lem:bounded-step-existence} it suffices to show that there is an $\eps$-optimal delay function for which all sinks of $\contMdp$ are good. Indeed, plugging the definition of $\dmin$ into~\eqref{eq:good-step-bound} we get that such a delay function $\timeouts$ satisfies 
\begin{align}
\CostBound{\msteps,\timeouts}&\leq \frac{\CostBound{\mcost,\timeouts}}{\eps} \cdot\frac{(\maxValue{\contMdp}+ \maxRew\cdot(2+2/\lambda)\cdot \lambda)^2 \cdot 4\lambda^2\cdot e^{2\lambda\dmin}}{ \minRew^2\cdot \minPst} \nonumber\\
&\leq \frac{\CostBound{\mcost,\timeouts}}{\eps}\cdot\frac{(\maxValue{\contMdp} +\maxRew\cdot(2+2/\lambda)\cdot \lambda)^2 \cdot 4\lambda^2\cdot 4}{\minRew^2\cdot \minPst},\label{eq:good-bound}
\end{align}
the last inequality following from the fact that $e^{\lambda\dmin}\leq 2$ (the numerator in~\eqref{eq:dmin-def} is at most $\minRew$ while the denominator is at least $\maxRew\cdot (4/\lambda)\cdot 4\lambda \geq \maxRew \geq \minRew$). The number $N$ from Lemma~\ref{lem:bounded-step-existence} can be easily obtained from the right-hand side of~\eqref{eq:good-bound} and from Lemma~\ref{lem:bound-value}.

The crucial idea, which is summarized in the following lemma, is that any $\epsSmall$-optimal delay vector can be modified (by ``inflating'' one of its components to $\lowerDelayBound$) in such a way that the total number of bad sinks decreases without significantly increasing the expected total cost.

\begin{lemma}
\label{lem:delay-increase}
Let $\epsSmall>0$ be arbitrary, and let $\timeouts$ be any globally $\epsSmall$-optimal delay vector in $\contMdp$ with $k>0$ bad sinks. Then there is an globally $(2\epsSmall)$-optimal delay vector $\timeoutsInfl$ in $\contMdp$ with $k-1$ bad sinks.
\end{lemma}
Before we prove the previous lemma, let us note that it already implies that for each $\eps >0$ there is globally $\eps$-optimal delay function such that all sinks are good. 
 Indeed, a stable globally $\epsSmall$-optimal delay function $\timeoutsStable$ can be constructed by taking any globally $\epsSmall/2^{|\allact|}$-optimal delay function and iteratively using Lemma~\ref{lem:delay-increase} to remove all of its at most $|\statesM|$ bad sinks.
 
We devote the rest of this section to the proof of Lemma~\ref{lem:delay-increase}, which is perhaps the most technically involved part in the proof of Theorem~\ref{thm:unconstrained}. Fix an arbitrary globally $\epsSmall$-optimal delay function $\timeouts$ in $\contMdp$ and let $\Sink$ be any sink bad for $\timeouts$. For any state $s\in \Sink$ we define the value 

\begin{align*}
\sinkVal{\sta}{\contMdp(\timeouts)} =& 1/\lambda \cdot \rateRew(\sta) +  \sum_{\sta' \in \statesM} \prob(\sta,\sta') \cdot (\expred_{\contMdp[\sta'](\timeouts)} + \impRewExp(\sta,\sta'))\\
&+ \sum_{\sta'' \in \allact} \sum_{\sta' \in \statesM} \prob(\sta,\sta'') \trans(\sta'',\sta') \cdot (\expred_{\contMdp[\sta'](\timeouts)} + \impRewExp(\sta,\sta'') + \impRewFix(\sta'',\sta)).
\end{align*}
Intuitively the value is expected cost paid from $\sta$ if exactly one exponential transition occurs in $\sta$ before reaching $\statesM$. We denote the probability mass function of Poisson distribution with parameter $\lambda > 0$ as $\pmf_{\poiss(\lambda)}:\Nseto \to [0,1]$, where $\pmf_{\poiss(\lambda)}(i) = \e^{-\lambda} \cdot \frac{\lambda^i}{i!}$. Furthermore we denote $\pmf_{\poiss(\lambda)}(\geq i) = \sum_{j=i}^{\infty} \pmf_{\poiss(\lambda)}(j)$. 
\begin{definition}
Let $\fdC(\timeouts)$ be a fdCTMC. We define a \emph{jump chain} for $\fdC(\timeouts)$ as a series of random variables $(\jumpChain^{\fdC(\timeouts)})_n$ where for all $n \in \Nseto$, $\jumpChain^{\fdC(\timeouts)}_n:\Omega \to \states$ and for a run $\ctrun = (\sta_0,\delays_0) t_0 \cdots$
$$\jumpChain^{\fdC(\timeouts)}_n(\ctrun)= \sta_n.$$
Similarly we define define a \emph{time chain} for $\fdC(\timeouts)$ as a series of random variables $(\timeChain^{\fdC(\timeouts)})_n$ where for all $n \in \Nseto$, $\timeChain^{\fdC(\timeouts)}_n:\Omega \to \Rsetpo$ and for a run $\ctrun = (\sta_0,\delays_0) t_0 \cdots$
$$\timeChain^{\fdC(\timeouts)}_n(\ctrun)= t_n.$$
\end{definition}

The next technical lemma gives tight upper and lower bounds on $\expred_{\contMdp[\sta](\timeouts)}$ of $\sta$ in bad sink using its $\sinkVal{\sta}{\contMdp(\timeouts)}$.

\begin{lemma}
\label{lem:sink-value-bounds}
Let $\epsSmall>0$ be arbitrary, and let $\timeouts$ be any globally $\epsSmall$-optimal delay function in $\contMdp$. For each sink $\Sink$ bad for $\timeouts$, state $\sta \in \Sink$ and for every $\delayBound \in \{\lowerDelayBound,\timeouts(\sta) \}$ it holds that
$$\tempA \leq \expred_{\contMdp[\sta](\timeouts[\sta/\delayBound])} \leq \tempA + \tempB,$$
where 

\begin{align*}
\tempA = \pmf_{\poiss(\lambda \delayBound)}(0) \cdot \sum_{\sta' \in \Sink} \trans(\sta,\sta') \cdot \expred_{\contMdp[\sta'](\timeouts)} + \pmf_{\poiss(\lambda \delayBound)}(1) \cdot \sinkVal{\sta}{\contMdp(\timeouts)}\\
\end{align*}
and 
$$\tempB = \upperSinkValueBound.$$
\end{lemma}

\begin{proof}
We fix a bad sink $\Sink$. For each state $\sta \in \Sink$ holds 

$$\expred_{\contMdp[\sta](\timeouts[\sta/\delayBound])} = \sum_{\sta' \in \statesM} \mtran(\sta,\delayBound)(\sta') \cdot \expred_{\contMdp[\sta'](\timeouts)} + \mcost(\sta,\delayBound).$$
We start evaluating the right hand side. The most efficient computation of the right hand side can be done by employing uniformization approach. We build a fdCTMC $\fdC'[\sta](\delayBound)$ according to Definition~\ref{def:subordinated}.
By computing the transient analysis of $\fdC'[\sta](\delayBound)$ at time $\delayBound$ we can efficiently compute the $\mtran(\sta,\delayBound)(\sta')$ for each state $\sta' \in \statesM$:

\begin{align}
&\sum_{\sta' \in \statesM} \mtran(\sta,\delayBound)(\sta') \cdot \expred_{\contMdp[\sta'](\timeouts)} = \label{eq:sink-value-boundsA} \\
&= \sum_{\sta' \in \statesM} \sum_{i=0}^{\infty} \probm_{\fdC'[\sta](\delayBound)}(\text{exactly $i$ exponential steps with rate $\lambda$ happen in time $\delayBound$}) \nonumber \\
 &\cdot \probm_{\fdC'[\sta](\delayBound)}(\text{reach $\sta'$ from $\sta \mid$ $i$ exp. steps happened}) \cdot \expred_{\contMdp[\sta'](\timeouts)} \nonumber \\
&= \sum_{\sta' \in \statesM} \sum_{i=0}^{\infty} \pmf_{\poiss(\lambda \delayBound)}(i) \cdot \probm_{\fdC'[\sta](\delayBound)}(\text{reach $\sta'$ from $\sta \mid$ $i$ exp. steps happened}) \cdot \expred_{\contMdp[\sta'](\timeouts)} \nonumber \\
&= \sum_{i=0}^{\infty} \pmf_{\poiss(\lambda \delayBound)}(i) \cdot \sum_{\sta' \in \statesM} \probm_{\fdC'[\sta](\delayBound)}(\text{reach $\sta'$ from $\sta \mid$ $i$ exp. steps happened}) \cdot \expred_{\contMdp[\sta'](\timeouts)} \nonumber \\
&= \sum_{i=0}^{1} \pmf_{\poiss(\lambda \delayBound)}(i) \cdot \sum_{\sta' \in \statesM} \probm_{\fdC'[\sta](\delayBound)}(\text{reach $\sta'$ from $\sta \mid$ $i$ exp. steps happened}) \cdot \expred_{\contMdp[\sta'](\timeouts)} \nonumber \\
 &\quad+ \sum_{i=2}^{\infty} \pmf_{\poiss(\lambda \delayBound)}(i) \cdot \sum_{\sta' \in \statesM}  \probm_{\fdC'[\sta](\delayBound)}(\text{reach $\sta'$ from $\sta \mid$ $i$ exp. steps happened}) \cdot \expred_{\contMdp[\sta'](\timeouts)} \nonumber \\
&= \pmf_{\poiss(\lambda \delayBound)}(0) \cdot \sum_{\sta' \in \states} \trans(\sta,\sta') \expred_{\contMdp[\sta'](\timeouts)} \label{eq:sink-value-boundsB}\\
 &\quad+ \pmf_{\poiss(\lambda \delayBound)}(1) \cdot \Big ( \sum_{\sta' \in \statesM} \prob(\sta,\sta') \cdot \expred_{\contMdp[\sta'](\timeouts)} + \sum_{\sta' \in \allact} \sum_{\sta'' \in \statesM} \prob(\sta,\sta'') \cdot \trans(\sta'',\sta') \cdot \expred_{\contMdp[\sta'](\timeouts)} \Big )\label{eq:sink-value-boundsC} \\
 &\quad+ \sum_{i=2}^{\infty} \pmf_{\poiss(\lambda \delayBound)}(i) \cdot \sum_{\sta' \in \statesM}  \probm_{\fdC'[\sta](\delayBound)}(\text{reach $\sta'$ from $\sta \mid$ $i$ exp. steps happened}) \cdot \expred_{\contMdp[\sta'](\timeouts)}.\label{eq:sink-value-boundsD}
\end{align}
We can now simply bound the $\sum_{\sta' \in \statesM} \mtran(\sta,\delayBound)(\sta') \cdot \expred_{\contMdp[\sta'](\timeouts)}$ by approximating the expression \eqref{eq:sink-value-boundsD}: 

\begin{align}
0 &\leq \sum_{i=2}^{\infty} \pmf_{\poiss(\lambda \delayBound)}(i) \cdot \sum_{\sta' \in \statesM}  \probm_{\fdC'[\sta](\delayBound)}(\text{reach $\sta'$ from $\sta \mid$ $i$ exp. steps happened}) \cdot \expred_{\contMdp[\sta'](\timeouts)} \leq \label{eq:sink-value-boundsE}\\
 &\leq \sum_{i=2}^{\infty} \pmf_{\poiss(\lambda \delayBound)}(i) \cdot \sum_{\sta' \in \statesM}  \probm_{\fdC'[\sta](\delayBound)}(\text{reach $\sta'$ from $\sta \mid$ $i$ exp. steps happened}) \cdot (\maxValue{\contMdp} + \epsSmall) \nonumber \\
 &= \sum_{i=2}^{\infty} \pmf_{\poiss(\lambda \delayBound)}(i) \cdot (\maxValue{\contMdp} + \epsSmall) \nonumber \\
 &= \pmf_{\poiss(\lambda \delayBound)}(\geq 2) \cdot (\maxValue{\contMdp} + \epsSmall), \label{eq:sink-value-boundsF}
\end{align}
where the last inequality follows from global $\epsSmall$-optimality of $\timeouts$.

Let us now evaluate the $\mcost(s,\delayBound)$. We first define random variables $\costFdC R[\statesM]$ and $\costFdC I[\statesM]$ that assign to each run $\ctrun = (\sta_0,\delays_0) t_0 \cdots$ the \emph{total rate and impulse cost before reaching $\statesM$} (in at least one transition), respectively, given by
$$
\costFdC R[\statesM](\ctrun)
= \begin{cases}
\sum_{i=0}^{n-1} t_i\cdot\rateRew(s_i) 
 & \text{for minimal $n>0$ such that $(s_n,d_n)\in\statesM$,} \\
\infty & \text{if there is no such $n$,}
\end{cases}
$$
and
$$
\costFdC I[\statesM](\ctrun)
= \begin{cases}
\sum_{i=0}^{n-1} \impRew_i(\ctrun)
 & \text{for minimal $n>0$ such that $(s_n,d_n)\in\statesM$,} \\
\infty & \text{if there is no such $n$,}
\end{cases}
$$
where $\impRew_i(\ctrun)$ equals $\impRewExp(s_i,s_{i+1})$ for an exp-delay transition, i.e. when $t_i < \delays_i$, and
equals $\impRewFix(s_i,s_{i+1})$
for a fixed-delay transition, i.e. when $t_i = \delays_i$. We can divide $\mcost(s,\delayBound)$ in two parts: 

\begin{align}
\mcost(\sta,\delayBound) = \expred_{\fdC[s,d]}^{\costFdC[S']} = \expred_{\fdC'[\sta](\delayBound)}^{\costFdC I[S']} + \expred_{\fdC'[\sta](\delayBound)}^{\costFdC R[S']} \label{eq:sink-value-boundsG}
\end{align}
We will similarly as above use uniformization to underestimate the expected impulse cost: 

\begin{align}
\expred_{\fdC'[\sta](\delayBound)}^{\costFdC I[S']} &= \sum_{i=0}^{\infty} \pmf_{\poiss(\lambda \delayBound)}(i) \cdot \sum_{\sta' \in \statesM} \probm_{\fdC'[\sta](\delayBound)}(\text{reach $\sta'$ from $\sta \mid$ $i$ exp. steps happened}) \nonumber\\
 &\quad\quad\cdot \expected_{\fdC'[\sta](\delayBound)}[\costFdC I[S'] \mid \text{reach $\sta'$ from $\sta$, $i$ exp. steps happened}] \nonumber\\
\expred_{\fdC'[\sta](\delayBound)}^{\costFdC I[S']} &= \sum_{i=0}^{1} \pmf_{\poiss(\lambda \delayBound)}(i) \cdot \sum_{\sta' \in \statesM} \probm_{\fdC'[\sta](\delayBound)}(\text{reach $\sta'$ from $\sta \mid$ $i$ exp. steps happened}) \label{eq:sink-value-boundsJ}\\
 &\quad\quad\cdot \expected_{\fdC'[\sta](\delayBound)}[\costFdC I[S'] \mid \text{reach $\sta'$ from $\sta$, $i$ exp. steps happened}] \label{eq:sink-value-boundsK}\\
  &\quad+ \sum_{i=2}^{\infty} \pmf_{\poiss(\lambda \delayBound)}(i) \cdot \sum_{\sta' \in \statesM} \probm_{\fdC'[\sta](\delayBound)}(\text{reach $\sta'$ from $\sta \mid$ $i$ exp. steps happened}) \label{eq:sink-value-boundsL}\\
 &\quad\quad\cdot \expected_{\fdC'[\sta](\delayBound)}[\costFdC I[S'] \mid \text{reach $\sta'$ from $\sta$, $i$ exp. steps happened}]. \label{eq:sink-value-boundsM}
\end{align}
Now we evaluate \eqref{eq:sink-value-boundsJ} and \eqref{eq:sink-value-boundsK}:

\begin{align}
&\sum_{i=0}^{1} \pmf_{\poiss(\lambda \delayBound)}(i) \cdot \sum_{\sta' \in \statesM} \probm_{\fdC'[\sta](\delayBound)}(\text{reach $\sta'$ from $\sta \mid$ $i$ exp. steps happened}) \label{eq:sink-value-boundsN} \\
 &\quad\cdot \expected_{\fdC'[\sta](\delayBound)}[\costFdC I[S'] \mid \text{reach $\sta'$ from $\sta$, $i$ exp. steps happened}] = \label{eq:sink-value-boundsO}\\
&= \impulseCosts + \pmf_{\poiss(\lambda \delayBound)}(0) \cdot \sum_{\sta'' \in \statesM} \trans(\sta,\sta') \cdot \impRewFix(\sta,\sta') \nonumber \\
&= \impulseCosts, \label{eq:sink-value-boundsP}
\end{align}
where 
\begin{align}
\impulseCosts = \pmf_{\poiss(\lambda \delayBound)}(1) \cdot \Big ( &\sum_{\sta' \in \statesM} \prob(\sta,\sta') \cdot \impRewExp(\sta,\sta') \nonumber\\ 
&+ \sum_{\sta' \in \allact} \sum_{\sta'' \in \statesM} \prob(\sta,\sta'') \trans(\sta'',\sta') \cdot (\impRewExp(\sta,\sta'') + \impRewFix(\sta'',\sta')) \Big ).\label{eq:sink-value-boundsQ}
\end{align}
The last equation follows from the fact that all fixed-delay transitions from $\sta$ go to states in $\Sink$ and from zero impulse cost of such transitions because $\sink$ is a bad sink. 
 
We now over and underestimate terms \eqref{eq:sink-value-boundsL} and \eqref{eq:sink-value-boundsM}.
\begin{align}
0 &\leq \sum_{i=2}^{\infty} \pmf_{\poiss(\lambda \delayBound)}(i) \cdot \sum_{\sta' \in \statesM} \probm_{\fdC'[\sta](\delayBound)}(\text{reach $\sta'$ from $\sta \mid$ $i$ exp. steps happened}) \label{eq:sink-value-boundsR}\\
 &\quad\quad\cdot \expected_{\fdC'[\sta](\delayBound)}[\costFdC I[S'] \mid \text{reach $\sta'$ from $\sta$, $i$ exp. steps happened}] \label{eq:sink-value-boundsS}\\
&\leq \sum_{i=2}^{\infty} \pmf_{\poiss(\lambda \delayBound)}(i) \cdot \sum_{\sta' \in \statesM} \probm_{\fdC'[\sta](\delayBound)}(\text{reach $\sta'$ from $\sta \mid$ $i$ exp. steps happened}) \cdot 2 \cdot i \cdot \maxRew \nonumber \\
&\leq 2 \cdot  \maxRew  \cdot \sum_{i=2}^{\infty} \pmf_{\poiss(\lambda \delayBound)}(i) \cdot i \nonumber \\
&= 2 \cdot \maxRew \cdot \Big ( \sum_{i=0}^{\infty} \pmf_{\poiss(\lambda \delayBound)}(i) \cdot i - \sum_{i=0}^{1} \pmf_{\poiss(\lambda \delayBound)}(i) \cdot i \Big ) \nonumber \\
&= 2 \cdot \maxRew \cdot ( \lambda \delayBound -  \pmf_{\poiss(\lambda \delayBound)}(1)).\label{eq:sink-value-boundsT}   
\end{align}
Using equations \eqref{eq:sink-value-boundsJ}, \eqref{eq:sink-value-boundsK}, \eqref{eq:sink-value-boundsL}, \eqref{eq:sink-value-boundsM}, \eqref{eq:sink-value-boundsN}, \eqref{eq:sink-value-boundsO}, \eqref{eq:sink-value-boundsP}, \eqref{eq:sink-value-boundsQ}, \eqref{eq:sink-value-boundsR}, \eqref{eq:sink-value-boundsS}, \eqref{eq:sink-value-boundsT}  we can under and overestimate the expected impulse cost accumulated on paths from $\sta$ to $\statesM$ by:

\begin{align}
\impulseCosts \leq \expred_{\fdC'[\sta](\delayBound)}^{\costFdC I[S']} \leq \impulseCosts + 2 \cdot \maxRew \cdot ( \lambda \delayBound -  \pmf_{\poiss(\lambda \delayBound)}(1)).\label{eq:sink-value-boundsU}
\end{align}
 
We evaluate the rate cost part of $\mcost(s,\delayBound)$.  
 
\begin{align*}
\expred_{\fdC'[\sta](\delayBound)}^{\costFdC R[S']} &= \sum_{\sta' \in \states} \text{(expected time spent in $\sta'$ until hitting \statesM)} \cdot \rateRew(\sta') \\
&= \text{(expected time spent in $\sta$ until the first transition is taken)} \cdot \rateRew(\sta) \\
 &\quad+ \sum_{\sta' \in \states} \text{(expected time spent in $\sta'$ until hitting $\statesM$ except the 1st trans.)} \cdot \rateRew(\sta'),\\
\end{align*}
where 
\begin{align*}
&\text{(expected time spent in $\sta$ until the first transition is taken)} = \\
&= \int_{0}^{\delayBound} t \cdot \lambda \cdot \e^{-\lambda t} dt + \delayBound \cdot \e^{-\lambda \delayBound} \\
&= -\delayBound \cdot \e^{-\lambda \delayBound}  + \frac{1-\e^{-\lambda \delayBound}}{\lambda} + \delayBound \cdot \e^{-\lambda \delayBound} \\
&= \frac{1-\e^{-\lambda \delayBound}}{\lambda} \\
&= \frac{1-\pmf_{\poiss(\lambda \delayBound)}(0)}{\lambda} \\
&= \frac{\pmf_{\poiss(\lambda \delayBound)}(\geq 1)}{\lambda} \\
&= \frac{\pmf_{\poiss(\lambda \delayBound)}(1)}{\lambda} + \frac{\pmf_{\poiss(\lambda \delayBound)}(\geq 2)}{\lambda} \\
\end{align*}
Thus 
\begin{align*}
\expred_{\fdC'[\sta](\delayBound)}^{\costFdC R[S']} &= \text{(expected time spent in $\sta$ until the first transition is taken)} \cdot \rateRew(\sta) \\
 &\quad+ \sum_{\sta' \in \states} \text{(expected time spent in $\sta'$ until hitting $\statesM$ except the 1st trans.)} \cdot \rateRew(\sta'),\\
&= \Big ( \frac{\pmf_{\poiss(\lambda \delayBound)}(1)}{\lambda} + \frac{\pmf_{\poiss(\lambda \delayBound)}(\geq 2)}{\lambda} \Big ) \cdot \rateRew(\sta)\\
 &\quad+ \sum_{\sta' \in \states} \text{(expected time spent in $\sta'$ until hitting $\statesM$ except the 1st trans.)} \cdot \rateRew(\sta'),\\
&\leq \Big ( \frac{\pmf_{\poiss(\lambda \delayBound)}(1)}{\lambda} + \frac{\pmf_{\poiss(\lambda \delayBound)}(\geq 2)}{\lambda} \Big ) \cdot \rateRew(\sta)\\
 &\quad+ \sum_{\sta' \in \states} \text{(expected time spent in $\sta'$ until hitting $\statesM$ except the 1st trans.)} \cdot \maxRew,\\
&\leq \Big ( \frac{\pmf_{\poiss(\lambda \delayBound)}(1)}{\lambda} + \frac{\pmf_{\poiss(\lambda \delayBound)}(\geq 2)}{\lambda} \Big ) \cdot \rateRew(\sta) + \Big ( \delayBound - \frac{\pmf_{\poiss(\lambda \delayBound)}(\geq 2)}{\lambda} \Big ) \cdot \maxRew,\\
\end{align*}
Using above results we can easily bound from above and below the (expected rate cost accumulated on paths from $\sta$ to \statesM):

\begin{align}
&\frac{\pmf_{\poiss(\lambda \delayBound)}(1)}{\lambda} \cdot \rateRew(\sta) \leq \expred_{\fdC'[\sta](\delayBound)}^{\costFdC R[S']} \leq \label{eq:sink-value-boundsV}\\
&\leq \frac{\pmf_{\poiss(\lambda \delayBound)}(1)}{\lambda}\cdot \rateRew(\sta) + \frac{\pmf_{\poiss(\lambda \delayBound)}(\geq 2)}{\lambda} \cdot \rateRew(\sta) + \Big ( \delayBound - \frac{\pmf_{\poiss(\lambda \delayBound)}(\geq 1)}{\lambda} \Big) \cdot \maxRew. \label{eq:sink-value-boundsW}
\end{align}

Now we have all the results needed to prove the lemma. To get the lower bound we combine and evaluate \eqref{eq:sink-value-boundsA}, \eqref{eq:sink-value-boundsB}, \eqref{eq:sink-value-boundsC}, \eqref{eq:sink-value-boundsD}, \eqref{eq:sink-value-boundsE}, \eqref{eq:sink-value-boundsG}, \eqref{eq:sink-value-boundsQ}, \eqref{eq:sink-value-boundsU}, \eqref{eq:sink-value-boundsV}:

\begin{align*}
&\expred_{\contMdp[\sta](\timeouts[\sta/\delayBound])} =\\
=& \sum_{\sta' \in \statesM} \mtran(\sta,\delayBound)(\sta') \cdot \expred_{\contMdp[\sta'](\timeouts)} + \mcost(\sta,\delayBound)\\
\geq& \pmf_{\poiss(\lambda \delayBound)}(0) \cdot \sum_{\sta' \in \states} \trans(\sta,\sta') \expred_{\contMdp[\sta'](\timeouts)} \\
 &+ \pmf_{\poiss(\lambda \delayBound)}(1) \cdot \Big ( \sum_{\sta' \in \statesM} \prob(\sta,\sta') \cdot \expred_{\contMdp[\sta'](\timeouts)} + \sum_{\sta' \in \allact} \sum_{\sta'' \in \statesM} \prob(\sta,\sta'') \cdot \trans(\sta'',\sta') \cdot \expred_{\contMdp[\sta'](\timeouts)} \Big ) \\
 &+ \pmf_{\poiss(\lambda \delayBound)}(1) \\
 &\quad\quad\cdot \Big ( \sum_{\sta' \in \statesM} \prob(\sta,\sta') \cdot \impRewExp(\sta,\sta') + \sum_{\sta' \in \allact} \sum_{\sta'' \in \statesM} \prob(\sta,\sta'') \trans(\sta'',\sta') \cdot (\impRewExp(\sta,\sta'') + \impRewFix(\sta'',\sta')) \Big ) \\
 &+ \frac{\pmf_{\poiss(\lambda \delayBound)}(1)}{\lambda} \cdot \rateRew(\sta) \\
=& \pmf_{\poiss(\lambda \delayBound)}(0) \cdot \sum_{\sta' \in \states} \trans(\sta,\sta') \expred_{\contMdp[\sta'](\timeouts)} + \pmf_{\poiss(\lambda \delayBound)}(1) \cdot \sinkVal{\sta}{\contMdp(\timeouts)} \\
=& \pmf_{\poiss(\lambda \delayBound)}(0) \cdot \sum_{\sta' \in \Sink} \trans(\sta,\sta') \expred_{\contMdp[\sta'](\timeouts)} + \pmf_{\poiss(\lambda \delayBound)}(1) \cdot \sinkVal{\sta}{\contMdp(\timeouts)} = \tempA\\
\end{align*}
The second equality from bottom follows from definition of bad sink, i.e. all fixed-delay transitions from $\sta$ go to states in $\Sink$. For the upper bound we get the same reductions, but we need to add additional summands, see \eqref{eq:sink-value-boundsA}, \eqref{eq:sink-value-boundsB}, \eqref{eq:sink-value-boundsC}, \eqref{eq:sink-value-boundsD}, \eqref{eq:sink-value-boundsE}, \eqref{eq:sink-value-boundsF}, \eqref{eq:sink-value-boundsG}, \eqref{eq:sink-value-boundsQ}, \eqref{eq:sink-value-boundsU}, \eqref{eq:sink-value-boundsV}, \eqref{eq:sink-value-boundsW}. We get the final result:

\begin{align*}
\expred_{\contMdp[\sta](\timeouts[\sta/\delayBound])} &\leq \pmf_{\poiss(\lambda \delayBound)}(0) \cdot \sum_{\sta' \in \Sink} \trans(\sta,\sta') \expred_{\contMdp[\sta'](\timeouts)}+ \pmf_{\poiss(\lambda \delayBound)}(1) \cdot \sinkVal{\sta}{\contMdp(\timeouts)} \\
&\quad+ \pmf_{\poiss(\lambda \delayBound)}(\geq 2) \cdot (\maxValue{\contMdp} + \epsSmall) + 2 \cdot \maxRew \cdot ( \lambda \delayBound -  \pmf_{\poiss(\lambda \delayBound)}(1)) \\ 
&\quad+ \frac{\pmf_{\poiss(\lambda \delayBound)}(\geq 2)}{\lambda} \cdot \rateRew(\sta) + \Big ( \delayBound - \frac{\pmf_{\poiss(\lambda \delayBound)}(\geq 1)}{\lambda} \Big) \cdot \maxRew\\
&\leq \tempA + \pmf_{\poiss(\lambda \delayBound)}(\geq 2) \cdot \Big( \maxValue{\contMdp} + \epsSmall + \frac{\maxRew}{\lambda} \Big )\\
&\quad+ ( \lambda \delayBound -  \pmf_{\poiss(\lambda \delayBound)}(1)) \cdot \maxRew \cdot \Big (2 +\frac{1}{\lambda} \Big ) \\
&\leq \tempA + \lambda \delayBound \cdot \pmf_{\poiss(\lambda \delayBound)}(\geq 1) \cdot \Big (\maxValue{\contMdp} + \epsSmall + \maxRew \cdot \Big (2 +\frac{2}{\lambda} \Big) \Big) \\
&= \tempA + \tempB,
\end{align*}
where the last inequality follows from
$$\lambda \delayBound -  \pmf_{\poiss(\lambda \delayBound)}(1) = \lambda \delayBound (1- \e^{-\lambda \delayBound}) = \lambda \delayBound \cdot \pmf_{\poiss(\lambda \delayBound)}(\geq 1) = \e^{-\lambda \delayBound} \cdot \sum_{i=2}^{\infty} \frac{(\lambda \delayBound)^i}{(i-1)!} \geq \pmf_{\poiss(\lambda \delayBound)}(\geq2).$$
\end{proof}

\begin{lemma}
\label{lem:value-optimality}
Let $\timeouts$ be globally $\epsSmall$-optimal delay function in $\contMdp$. Let $\Sink$ be a bad sink for $\timeouts$ and let $\sta \in \Sink$ be such that $\sinkVal{\sta}{\contMdp(\timeouts)} = \min_{\sta'' \in \Sink} \sinkVal{\sta''}{\contMdp(\timeouts)}$. Then for each $\sta' \in \Sink$ it holds that
$$\expred_{\contMdp[\sta'](\timeouts)} \geq \frac{\pmf_{\poiss(\lambda \lowerDelayBound)}(1)}{(1-\pmf_{\poiss(\lambda \lowerDelayBound)}(0))} \cdot \sinkVal{\sta}{\contMdp(\timeouts)}, $$
and
$$\expred_{\contMdp[\sta'](\timeouts)} \geq  \sinkValue_{\contMdp(\timeouts)}^{\sta} - \frac{\pmf_{\poiss(\lambda \lowerDelayBound)}(\geq 2) \cdot (\maxValue{\contMdp}+\epsSmall)}{\pmf_{\poiss(\lambda \lowerDelayBound)}(1)}.$$
\end{lemma}

\begin{proof}

We fix a sink $\Sink$ bad for $\timeouts$ and state $\sta' \in \Sink$ with minimal value of $\expred_{\contMdp[\sta'](\timeouts)}$. To obtain the lower bound on $\expred_{\contMdp[\sta'](\timeouts)}$ we use Lemma~\ref{lem:sink-value-bounds}.

\begin{align}
\expred_{\contMdp[\sta'](\timeouts)} &\geq \pmf_{\poiss(\lambda \timeouts(\sta'))}(0) \cdot \sum_{\sta'' \in \Sink} \trans(\sta',\sta'') \expred_{\contMdp[\sta''](\timeouts)} + \pmf_{\poiss(\lambda \timeouts(\sta'))}(1) \cdot \sinkVal{\sta'}{\contMdp(\timeouts)} \nonumber \\
&\geq \pmf_{\poiss(\lambda \timeouts(\sta'))}(0) \cdot \sum_{\sta'' \in \Sink} \trans(\sta',\sta'') \expred_{\contMdp[\sta'](\timeouts)} + \pmf_{\poiss(\lambda \timeouts(\sta'))}(1) \cdot \sinkVal{\sta'}{\contMdp(\timeouts)} \label{eq:valOptA}\\
&= \pmf_{\poiss(\lambda \timeouts(\sta'))}(0) \cdot \expred_{\contMdp[\sta'](\timeouts)} + \pmf_{\poiss(\lambda \timeouts(\sta'))}(1) \cdot \sinkVal{\sta'}{\contMdp(\timeouts)} \nonumber \\
&\geq \pmf_{\poiss(\lambda \timeouts(\sta'))}(0) \cdot \expred_{\contMdp[\sta'](\timeouts)} + \pmf_{\poiss(\lambda \timeouts(\sta'))}(1) \cdot \sinkVal{\sta}{\contMdp(\timeouts)}, \label{eq:valOptC}
\end{align}
where the inequality \eqref{eq:valOptA} follows from minimality of $\expred_{\contMdp[\sta'](\timeouts)}$ among all states of $\sta' \in \Sink$ and the inequality \eqref{eq:valOptC} follows from the minimality of $\sta \in \Sink$ with respect to $\sinkVal{\sta}{\contMdp(\timeouts)}$. We proceed with simplification of inequality \eqref{eq:valOptC}:
\begin{align*}
\expred_{\contMdp[\sta'](\timeouts)} &\geq \pmf_{\poiss(\lambda \timeouts(\sta'))}(0) \cdot \expred_{\contMdp[\sta'](\timeouts)} + \pmf_{\poiss(\lambda \timeouts(\sta'))}(1) \cdot \sinkVal{\sta}{\contMdp(\timeouts)}\\
(1-\pmf_{\poiss(\lambda \timeouts(\sta'))}(0)) \cdot \expred_{\contMdp[\sta'](\timeouts)} &\geq  \pmf_{\poiss(\lambda \timeouts(\sta'))}(1) \cdot \sinkVal{\sta}{\contMdp(\timeouts)}\\
\expred_{\contMdp[\sta'](\timeouts)} &\geq  \frac{\pmf_{\poiss(\lambda \timeouts(\sta'))}(1)}{(1-\pmf_{\poiss(\lambda \timeouts(\sta'))}(0))} \cdot \sinkVal{\sta}{\contMdp(\timeouts)} \\
&\geq \frac{\pmf_{\poiss(\lambda \lowerDelayBound)}(1)}{(1-\pmf_{\poiss(\lambda \lowerDelayBound)}(0))} \cdot \sinkVal{\sta}{\contMdp(\timeouts)},  \\
\end{align*}
where the last inequality can be justified as follows: $\sta'$ belongs to the bad sink $\Sink$, thus from definition of a bad sink it holds that $0 < \timeouts(\sta')\leq \lowerDelayBound$. Since function 
 $$\frac{\pmf_{\poiss(\lambda \timeouts(\sta'))}(1)}{(1-\pmf_{\poiss(\lambda \timeouts(\sta'))}(0))} = \frac{\lambda \timeouts(\sta') \cdot \e^{-\lambda \timeouts(\sta')}}{\e^{-\lambda \timeouts(\sta')} \cdot \sum_{i=1}^{\infty} \frac{(\lambda \timeouts(\sta'))^{i}}{i!}} = \frac{1}{\sum_{i=1}^{\infty} \frac{(\lambda \timeouts(\sta'))^{i-1}}{i!}}$$
is monotone, and from $0 < \timeouts(\sta')\leq \lowerDelayBound$ we get the result. We picked the state $\sta'$ in $\Sink$ with minimal $\expred_{\contMdp[\sta'](\timeouts)}$ thus for all $\sta'' \in \Sink$ it holds that
$$\expred_{\contMdp[\sta''](\timeouts)} \geq \expred_{\contMdp[\sta'](\timeouts)} \geq \frac{\pmf_{\poiss(\lambda \lowerDelayBound)}(1)}{(1-\pmf_{\poiss(\lambda \lowerDelayBound)}(0))} \cdot \sinkVal{\sta}{\contMdp(\timeouts)}.$$

Now we simplify again the inequality \eqref{eq:valOptC} to get the second bound:
\begin{align*}
\sum_{i=0}^{\infty} \pmf_{\poiss(\lambda \timeouts(\sta'))}(i) \cdot \expred_{\contMdp[\sta'](\timeouts)} &\geq \pmf_{\poiss(\lambda \timeouts(\sta'))}(0) \\
&\quad\cdot \expred_{\contMdp[\sta'](\timeouts)} + \pmf_{\poiss(\lambda \timeouts(\sta'))}(1) \cdot \sinkVal{\sta}{\contMdp(\timeouts)} \\
\pmf_{\poiss(\lambda \timeouts(\sta'))}(1) \cdot \expred_{\contMdp[\sta'](\timeouts)} + \pmf_{\poiss(\lambda \timeouts(\sta'))}(\geq 2) \cdot \expred_{\contMdp[\sta'](\timeouts)} &\geq  \pmf_{\poiss(\lambda \timeouts(\sta'))}(1) \cdot \sinkVal{\sta}{\contMdp(\timeouts)} \\
\expred_{\contMdp[\sta'](\timeouts)} + \frac{\pmf_{\poiss(\lambda \timeouts(\sta'))}(\geq 2)}{\pmf_{\poiss(\lambda \timeouts(\sta'))}(1)} \cdot \expred_{\contMdp[\sta'](\timeouts)} &\geq  \sinkVal{\sta}{\contMdp(\timeouts)} \\
\expred_{\contMdp[\sta'](\timeouts)} + \frac{\pmf_{\poiss(\lambda \timeouts(\sta'))}(\geq 2) \cdot (\maxValue{\contMdp}+\epsSmall)}{\pmf_{\poiss(\lambda \timeouts(\sta'))}(1)} &\geq  \sinkVal{\sta}{\contMdp(\timeouts)} \\
\expred_{\contMdp[\sta'](\timeouts)} &\geq \sinkVal{\sta}{\contMdp(\timeouts)} -\frac{\pmf_{\poiss(\lambda \lowerDelayBound)}(\geq 2) \cdot (\maxValue{\contMdp}+\epsSmall)}{\pmf_{\poiss(\lambda \lowerDelayBound)}(1)},
\end{align*}
where the last inequality follows again from $0 < \timeouts(\sta')\leq \lowerDelayBound$ and from monotonicity of the following function
 $$\frac{\pmf_{\poiss(\lambda \timeouts(\sta'))}(\geq 2)}{\pmf_{\poiss(\lambda \timeouts(\sta'))}(1)} = \frac{\e^{-\lambda \timeouts(\sta')} \cdot \sum_{i=2}^{\infty} \frac{(\lambda \timeouts(\sta'))^{i}}{i!}}{\lambda \timeouts(\sta') \cdot \e^{-\lambda \timeouts(\sta')}} = \sum_{i=2}^{\infty} \frac{(\lambda \timeouts(\sta'))^{i-1}}{i!}.$$
Since we picked the state $\sta'$ in $\Sink$ with minimal $\expred_{\contMdp[\sta'](\timeouts)}$ it holds that, for all $\sta'' \in \Sink$
$$\expred_{\contMdp[\sta''](\timeouts)} \geq \expred_{\contMdp[\sta'](\timeouts)} \geq \sinkVal{\sta}{\contMdp(\timeouts)} -\frac{\pmf_{\poiss(\lambda \lowerDelayBound)}(\geq 2) \cdot (\maxValue{\contMdp}+\epsSmall)}{\pmf_{\poiss(\lambda \lowerDelayBound)}(1)}.$$
\end{proof}

\begin{lemma}
\label{lem:one-step-error}
Let $\epsSmall>0$ be arbitrary, and let $\timeouts$ be any globally $\epsSmall$-optimal delay function in $\contMdp$. Next, let $\Sink \subseteq \statesM$ be sink bad for $\timeouts$, and let $\sta \in \Sink$ be state with minimal $\sinkVal{\sta}{\contMdp(\timeouts)}$ among all states in $\Sink$. Then for each state $\sta' \in \statesM $ it holds that 
$$\expred_{\contMdp[\sta'](\timeouts[s/\lowerDelayBound])} -\expred_{\contMdp[\sta'](\timeouts)} \leq  \tempBLower,$$
 where $\tempBLower$ is as in Lemma~\ref{lem:sink-value-bounds}, i.e.
$$\tempBLower= \lambda \lowerDelayBound \cdot \pmf_{\poiss(\lambda \lowerDelayBound)}(\geq 1) \cdot \Big (\maxValue{\contMdp} + \eps + \maxRew \cdot \Big (2 +\frac{2}{\lambda} \Big) \Big).$$
\end{lemma}
 
\begin{proof}
We first prove the lemma for state $\sta$. Observe that since we start in state $\sta$, then according to delay function $\timeouts[s/\lowerDelayBound]$ we make a change of delay only in the fist visit of $\sta$ and then we proceed with original delay function. Thus, using Lemma~\ref{lem:sink-value-bounds} we start evaluating the left hand side:

\begin{align*}
&\expred_{\contMdp[\sta](\timeouts[s/\lowerDelayBound])} -\expred_{\contMdp[\sta](\timeouts)} \leq \\
&\leq \pmf_{\poiss(\lambda \lowerDelayBound)}(0) \cdot \sum_{\sta' \in \Sink} \trans(\sta,\sta') \expred_{\contMdp[\sta'](\timeouts)} + \pmf_{\poiss(\lambda \lowerDelayBound)}(1) \cdot \sinkVal{\sta}{\contMdp(\timeouts)} + \tempBLower -\expred_{\contMdp[\sta](\timeouts)}\\
&\leq \pmf_{\poiss(\lambda \lowerDelayBound)}(0) \cdot \sum_{\sta' \in \Sink} \trans(\sta,\sta') \expred_{\contMdp[\sta'](\timeouts)} + \pmf_{\poiss(\lambda \lowerDelayBound)}(1) \cdot \sinkVal{\sta}{\contMdp(\timeouts)} + \tempBLower\\
&\quad- \pmf_{\poiss(\lambda \timeouts(\sta))}(0) \cdot \sum_{\sta' \in \Sink} \trans(\sta,\sta') \expred_{\contMdp[\sta'](\timeouts)} + \pmf_{\poiss(\lambda \timeouts(\sta))}(1) \cdot \sinkVal{\sta}{\contMdp(\timeouts)}\\
&= \big (\pmf_{\poiss(\lambda \lowerDelayBound)}(0) - \pmf_{\poiss(\lambda \timeouts(\sta))}(0) \big ) \cdot \sum_{\sta' \in \Sink} \trans(\sta,\sta') \expred_{\contMdp[\sta'](\timeouts)} \\
 &\quad+ \big (\pmf_{\poiss(\lambda \lowerDelayBound)}(1) - \pmf_{\poiss(\lambda \timeouts(\sta))}(1) \big ) \cdot \sinkVal{\sta}{\contMdp(\timeouts)} + \tempBLower.
\end{align*}
From the definition of a bad sink it holds that $\lowerDelayBound > \timeouts(\sta)$ and thus $\pmf_{\poiss(\lambda \lowerDelayBound)}(0) - \pmf_{\poiss(\lambda \timeouts(\sta))}(0) < 0.$ From Lemma~\ref{lem:value-optimality} it further follows that for all $\sta' \in \Sink$

$$\expred_{\mdp[\sta'](\timeouts)} \geq \frac{\pmf_{\poiss(\lambda \lowerDelayBound)}(1)}{(1-\pmf_{\poiss(\lambda \lowerDelayBound)}(0))} \cdot \sinkVal{\sta}{\contMdp(\timeouts)}.$$ 
Thus
\begin{align*}
&\expred_{\mdp[\sta](\timeouts[s/\lowerDelayBound])} -\expred_{\mdp[\sta](\timeouts)} \leq \\
&\leq \big (\pmf_{\poiss(\lambda \lowerDelayBound)}(0) - \pmf_{\poiss(\lambda \timeouts(\sta))}(0) \big ) \cdot \sum_{\sta' \in \Sink} \trans(\sta,\sta') \expred_{\contMdp[\sta'](\timeouts)} \\
 &\quad+ \big (\pmf_{\poiss(\lambda \lowerDelayBound)}(1) - \pmf_{\poiss(\lambda \timeouts(\sta))}(1) \big ) \cdot \sinkVal{\sta}{\contMdp(\timeouts)} + \tempBLower\\
&\leq \big (\pmf_{\poiss(\lambda \lowerDelayBound)}(0) - \pmf_{\poiss(\lambda \timeouts(\sta))}(0) \big ) \cdot \frac{\pmf_{\poiss(\lambda \lowerDelayBound)}(1)}{(1-\pmf_{\poiss(\lambda \lowerDelayBound)}(0))} \cdot \sinkVal{\sta}{\contMdp(\timeouts)} \\
 &\quad+ \big (\pmf_{\poiss(\lambda \lowerDelayBound)}(1) - \pmf_{\poiss(\lambda \timeouts(\sta))}(1) \big ) \cdot \sinkVal{\sta}{\contMdp(\timeouts)} + \tempBLower\\
&= \Big ( \big (\pmf_{\poiss(\lambda \lowerDelayBound)}(0) - \pmf_{\poiss(\lambda \timeouts(\sta))}(0) \big ) \cdot \frac{\pmf_{\poiss(\lambda \lowerDelayBound)}(1)}{(1-\pmf_{\poiss(\lambda \lowerDelayBound)}(0))} + \big (\pmf_{\poiss(\lambda \lowerDelayBound)}(1) - \pmf_{\poiss(\lambda \timeouts(\sta))}(1) \big ) \Big ) \cdot \sinkVal{\sta}{\contMdp(\timeouts)} \\
&\quad+\tempBLower.\\
\end{align*}
To get the lemma it now suffices to show that the following holds
$$\Big ( \big (\pmf_{\poiss(\lambda \lowerDelayBound)}(0) - \pmf_{\poiss(\lambda \timeouts(\sta))}(0) \big ) \cdot \frac{\pmf_{\poiss(\lambda \lowerDelayBound)}(1)}{(1-\pmf_{\poiss(\lambda \lowerDelayBound)}(0))} + \big (\pmf_{\poiss(\lambda \lowerDelayBound)}(1) - \pmf_{\poiss(\lambda \timeouts(\sta))}(1) \big ) \Big ) \leq 0.$$
For better readability, we use abbreviations $\alpha(i)= \pmf_{\poiss(\lambda \lowerDelayBound)}(i)$ and $\beta(i) = \pmf_{\poiss(\lambda \timeouts(\sta))}(1)$ for $i \in \{ 0,1 \}$. Now we are ready to show that the above inequality holds:

\begin{align*}
(\alpha(0)-\beta(0))\cdot \frac{\alpha(1)}{1-\alpha(0)} + \alpha(1) - \beta(1) &< 0 \\
\alpha(0)\alpha(1) - \alpha(1)\beta(0) + \alpha(1) -\beta(1) -\alpha(0)\alpha(1) + \alpha(0)\beta(1) &\leq 0 \\
- \alpha(1)\beta(0) + \alpha(1) -\beta(1) + \alpha(0)\beta(1) &\leq 0 \\
\alpha(1) - \alpha(1)\beta(0) &\leq \beta(1) - \alpha(0)\beta(1) \\
\alpha(1) \cdot(1- \beta(0)) &\leq \beta(1) \cdot (1 - \alpha(0)) \\
\frac{\alpha(1)}{1 - \alpha(0)} &\leq \frac{\beta(1)}{1- \beta(0)} \\
\frac{\pmf_{\poiss(\lambda \lowerDelayBound)}(1)}{1 - \pmf_{\poiss(\lambda \lowerDelayBound)}(0)} &\leq \frac{\pmf_{\poiss(\lambda \timeouts(\sta))}(1)}{1- \pmf_{\poiss(\lambda \timeouts(\sta))}(0)},\\
\frac{\lambda \lowerDelayBound \cdot \e^{-\lambda \lowerDelayBound}}{\e^{-\lambda \lowerDelayBound} \cdot \sum_{i=1}^{\infty} \frac{(\lambda \lowerDelayBound)^i}{i!}} &\leq \frac{\lambda \timeouts(\sta) \cdot \e^{-\lambda \timeouts(\sta)}}{\e^{-\lambda \timeouts(\sta)} \cdot \sum_{i=1}^{\infty} \frac{(\lambda \timeouts(\sta))^i}{i!}} \\
\frac{1}{\sum_{i=1}^{\infty} \frac{(\lambda \lowerDelayBound)^{i-1}}{i!}} &\leq \frac{1}{\sum_{i=1}^{\infty} \frac{(\lambda \timeouts(\sta))^{i-1}}{i!}} \\
\end{align*}
where the last inequality follows from definition of bad sink that $0 <\timeouts(\sta) < \lowerDelayBound$.

We finished the proof for state $\sta$. It remains to show that the increase of value for any other state $\sta' \in \statesM \setminus \{\sta\}$ is smaller or equal to $\tempBLower$. For every state $\sta' \in \statesM \setminus \{\sta\}$ it holds

\begin{align*}
\expred_{\contMdp[\sta'](\timeouts)} &= \probm_{\contMdp(\timeouts)}(\jumpChain^{\mdp}_{\hit_1^{\{\sta\}}}=\sta) \cdot \expected_{\contMdp[\sta'](\timeouts)}[\costFdC \mid \jumpChain^{\mdp}_{\hit_1^{\{\sta\}}}=\sta] \\
 &\quad+ (1-\probm_{\contMdp(\timeouts)}(\jumpChain^{\mdp}_{\hit_1^{\{\sta\}}}=\sta)) \cdot \expected_{\contMdp[\sta'](\timeouts)}[\costFdC \mid \jumpChain^{\mdp}_{\hit_1^{\{\sta\}}}\neq\sta] \\
&= \probm_{\contMdp(\timeouts)}(\jumpChain^{\mdp}_{\hit_1^{\{\sta\}}}=\sta) \\
&\quad\cdot \big ( \text{(expected cost until reaching $\sta$ in $\contMdp(\timeouts) \mid \jumpChain^{\mdp}_{\hit_1^{\{\sta\}}}=\sta$)} + \expred_{\contMdp[\sta](\timeouts)} \big ) \\
 &\quad+ (1-\probm_{\contMdp(\timeouts)}(\jumpChain^{\mdp}_{\hit_1^{\{\sta\}}}=\sta)) \cdot \expected_{\contMdp[\sta'](\timeouts)}[\costFdC \mid \jumpChain^{\mdp}_{\hit_1^{\{\sta\}}} \neq \sta] \\
\end{align*}
and similarly
\begin{align*}
&\expred_{\contMdp[\sta'](\timeouts[s/\lowerDelayBound])} = \\ &=\probm_{\contMdp(\timeouts)}(\jumpChain^{\mdp}_{\hit_1^{\{\sta\}}}=\sta) \cdot \big ( \text{(expected cost until reaching $\sta$ in $\contMdp(\timeouts)\mid \jumpChain^{\mdp}_{\hit_1^{\{\sta\}}}=\sta$)} + \expred_{\contMdp[\sta](\timeouts[s/\lowerDelayBound])} \big ) \\
 &\quad+ (1-\probm_{\contMdp(\timeouts)}(\jumpChain^{\mdp}_{\hit_1^{\{\sta\}}}=\sta)) \cdot \expected_{\contMdp[\sta'](\timeouts[s/\lowerDelayBound])}[\costFdC \mid \jumpChain^{\mdp}_{\hit_1^{\{\sta\}}}\neq\sta].\\
\end{align*}
The last equality follows from the fact that $\contMdp(\timeouts[s/\lowerDelayBound])$ behaves exactly in the same way as $\contMdp(\timeouts)$ until reaching $\sta$. Thus for each state $\sta' \in \statesM \setminus \{\sta\}$ we are getting
$$\expred_{\contMdp[\sta'](\timeouts[s/\lowerDelayBound])} -\expred_{\contMdp[\sta'](\timeouts)} \leq \probm_{\contMdp(\timeouts)}(\jumpChain^{\mdp}_{\hit_1^{\{\sta\}}}=\sta) \cdot \tempBLower \leq \tempBLower,$$ 
what completes the proof.
\end{proof} 

By Lemma~\ref{lem:one-step-error} we have that the increase in expected cost is bounded for each state if we increase delay of state $\sta$ with minimal value in first visit. We show that we have positive probability that we reach target before returning to $\sta$. First we will prove basic lemma what we will need on the way. 

\begin{lemma} \label{lem:expFuncBound}
For all $x \in \Rsetpo$ it holds that $\e^{x} - 1 \leq x \cdot \e^x$.
\end{lemma}
\begin{proof}
The derivation of $\e^x - 1$ is $x \cdot \e^x$. The derivation of $x \cdot \e^x$ is $x^2 \cdot \e^x + \e^x$. Observe that for all $x \in \Rsetpo$ it holds that
\begin{align*}
x \cdot \e^x &\leq x^2 \cdot \e^x + \e^x \\
x &\leq x^2 + 1,
\end{align*}
what implies that $\e^x - 1$ increases slower than $x \cdot \e^x$ on positive real numbers. Since both $\e^x - 1$ and $x \cdot \e^x$ equal to zero at $x =0$, it must hold that $\e^x - 1 \leq x \cdot \e^x$ for all $x \in \Rsetpo$. 
\end{proof}

\begin{lemma}
\label{lem:possitive-leaving-probability}
Let $\epsSmall >0$ be arbitrary, and let $\timeouts$ be any globally $\epsSmall$-optimal delay function for $\contMdp$, $\Sink \subseteq \statesM$ is sink bad for $\timeouts$, and $\sta \in \Sink$ is state with minimal $\sinkVal{\sta}{\contMdp(\timeouts)}$ among states in $\Sink$. It holds that
$$p_{\sta} \geq \pmf_{\poiss(\lambda \lowerDelayBound)}(1) \cdot \frac{\minRew/(2\cdot \lambda)}{\minRew/\lambda + \maxValue{\contMdp}+\epsSmall} ,$$
where $p_{\sta}$ is the probability that goal state is reached before returning to $\sta$ if run starts in $\sta$, i.e. 
$$p_{\sta} = \probm_{\contMdp(\timeouts)}(\jumpChain^{\mdp(\timeouts)}_n \in \mlocs', \forall l. i<l<j. \jumpChain^{\mdp(\timeouts)}_l \not = \sta \mid \jumpChain^{\mdp(\timeouts)}_i = \sta).$$
\end{lemma}

\begin{proof}
First we fix sink $\Sink$ bad for $\timeouts$ and $\sta \in \sink$ that has minimal $\sinkValue_{\contMdp(\timeouts)}$ in $\Sink$. Assume we have fdCTMC $\fdC'(\timeouts)$, where $\fdC'$ is same to fdCTMC from which the $\contMdp$ was created except that the initial state is $\sta$. Obviously it holds that
$$p_{\sta} = \probm_{\fdC'(\timeouts)}(\jumpChain^{\fdC'(\timeouts)}_n \in \goalStates, \forall l. 0<l<j. \jumpChain^{\fdC'(\timeouts)}_l \not = \sta).$$
Thus in the remaining part of this proof we will use mostly the $\fdC'(\timeouts)$. First we provide necessary definitions. The next property denotes the probability of reaching $\goalStates$ before $\Sink$ conditioned by the fact that exactly one exponential step was done from state $\sta$  in $\fdC'(\timeouts)$:
$$p_{\sta}'= \probm_{\fdC'(\timeouts)}(\jumpChain^{\fdC'(\timeouts)}_n \in \goalStates, \forall l. 0<l<j. \jumpChain^{\fdC'(\timeouts)}_l \not \in \Sink \mid \timeChain^{\fdC'(\timeouts)}_1<\timeouts(\sta), \jumpChain^{\fdC'(\timeouts)}_1 + \jumpChain^{\fdC'(\timeouts)}_2 \geq \timeouts(\sta)).$$
We denote by $\costFdCA$ and $\costFdCB$ the random variables assigning to each run $\ctrun = (\sta_0,\delays_0) t_0 \cdots$ the \emph{total cost after first state from $\statesM$ and before reaching $\goalStates$} and \emph{total cost after first state from $\statesM$ and before reaching state from $\Sink$ from which an exponential transition is taken}, given by
\begin{align*}
\costFdCA(\ctrun)
= \begin{cases}
\sum_{i=j}^{n-1} \left(
t_i\cdot\rateRew(s_i) +
\impRew_i(\ctrun)
\right)
 & \text{for minimal $n>0$ such that $s_n\in\goalStates$,} \\
 &\text{$\jumpChain^{\fdC'(\timeouts)}_j \in \statesM$, and $\forall l. 0<l<j. \jumpChain^{\fdC'(\timeouts)}_l \not \in \statesM $,} \\
\infty & \text{if there is no such $n$,}
\end{cases}
\end{align*}
and
\begin{align*}
&\costFdCB(\ctrun) = \\
&= \begin{cases}
\sum_{i=j}^{n-1} \left(
t_i\cdot\rateRew(s_i) +
\impRew_i(\ctrun)
\right)
 & \text{for minimal $n>0$ such that $\sta_n\in\Sink$, $\timeChain^{\fdC'(\timeouts)}_n<\timeouts(\sta_n)$,} \\
 &\text{$\jumpChain^{\fdC'(\timeouts)}_j \in \statesM$, and $\forall l. 0<l<j. \jumpChain^{\fdC'(\timeouts)}_l \not \in \statesM $,} \\
\infty & \text{if there is no such $n$,}
\end{cases}
\end{align*}
respectively, where $\impRew_i(\ctrun)$ equals $\impRewExp(s_i,s_{i+1})$ for an exp-delay transition, i.e. when $t_i < \delays_i$, and
equals $\impRewFix(s_i,s_{i+1})$
for a fixed-delay transition, i.e. when $t_i = \delays_i$.

Now we evaluate the following basic property:

\begin{align*}
&\sinkValue_{\contMdp(\timeouts)}^{\sta} \geq \\
&\geq p_{\sta}' \cdot \expected_{\fdC'(\timeouts)}[\costFdCA \mid \text{reach $\goalStates$ before  $\Sink$ and $\timeChain^{\fdC'(\timeouts)}_1<\timeouts(\sta), \jumpChain^{\fdC'(\timeouts)}_1 + \jumpChain^{\fdC'(\timeouts)}_2 \geq \timeouts(\sta)$}]\\
 &\quad+ (1-p_{\sta}') \cdot \expected_{\fdC'(\timeouts)}[\costFdCA \mid \text{reach $\Sink$ before  $\goalStates$ and $\timeChain^{\fdC'(\timeouts)}_1<\timeouts(\sta), \jumpChain^{\fdC'(\timeouts)}_1 + \jumpChain^{\fdC'(\timeouts)}_2 \geq \timeouts(\sta)$}]\\ 
&\geq (1-p_{\sta}') \cdot \expected_{\fdC'(\timeouts)}[\costFdCA \mid \text{reach $\Sink$ before  $\goalStates$ and $\timeChain^{\fdC'(\timeouts)}_1<\timeouts(\sta), \jumpChain^{\fdC'(\timeouts)}_1 + \jumpChain^{\fdC'(\timeouts)}_2 \geq \timeouts(\sta)$}]\\ 
&= (1-p_{\sta}') \cdot \Big (\tempCost + \sum_{\sta' \in \Sink} \pi(\sta') \cdot \expred_{\contMdp[\sta'](\timeouts)} \Big) \\
&= \tempCost + \sum_{\sta' \in \Sink} \pi(\sta') \cdot \expred_{\contMdp[\sta'](\timeouts)} -p_{\sta}' \cdot \Big (\tempCost + \sum_{\sta' \in \Sink} \pi(\sta') \cdot \expred_{\contMdp[\sta'](\timeouts)} \Big ) \\
&\geq \tempCost + \sum_{\sta' \in \Sink} \pi(\sta') \cdot \Big ( \sinkValue_{\contMdp(\timeouts)}^{\sta} - \frac{\pmf_{\poiss(\lambda \lowerDelayBound)}(\geq 2) \cdot (\maxValue{\contMdp}+\epsSmall)}{\pmf_{\poiss(\lambda \lowerDelayBound)}(1)} \Big) \\
&\quad- p_{\sta}' \cdot \Big ( \tempCost + \sum_{\sta' \in \Sink} \pi(\sta') \cdot (\maxValue{\contMdp}+\epsSmall) \Big) \\
&= \tempCost + \sinkValue_{\contMdp(\timeouts)}^{\sta} - \frac{\pmf_{\poiss(\lambda \lowerDelayBound)}(\geq 2) \cdot (\maxValue{\contMdp}+\epsSmall)}{\pmf_{\poiss(\lambda \lowerDelayBound)}(1)} - p_{\sta}' \cdot ( \tempCost + \maxValue{\contMdp}+\epsSmall ) 
\end{align*}
where 
\begin{align*}
\tempCost = \expected_{\fdC'(\timeouts)}[\costFdCB \mid \text{reach $\Sink$ before  $\goalStates$ and $\timeChain^{\fdC'(\timeouts)}_1<\timeouts(\sta), \jumpChain^{\fdC'(\timeouts)}_1 + \jumpChain^{\fdC'(\timeouts)}_2 \geq \timeouts(\sta)$}],
\end{align*}
and for $\sta' \in \Sink$ 

\begin{align*}
\pi(\sta') = \probm_{\fdC'(\timeouts)} &(\text{leaving from state $\sta'$ after returning to } \Sink \\
&\quad\mid \text{reach $\Sink$ before  $\goalStates$ and $\timeChain^{\fdC'(\timeouts)}_1<\timeouts(\sta), \jumpChain^{\fdC'(\timeouts)}_1 + \jumpChain^{\fdC'(\timeouts)}_2 \geq \timeouts(\sta)$}).
\end{align*}
The last inequality follows from Lemma~\ref{lem:value-optimality}. By further reduction we are getting:
\begin{align*}
\sinkValue_{\contMdp(\timeouts)}^{\sta} &\geq \tempCost + \sinkValue_{\contMdp(\timeouts)}^{\sta} - \frac{\pmf_{\poiss(\lambda \lowerDelayBound)}(\geq 2) \cdot (\maxValue{\contMdp}+\epsSmall)}{\pmf_{\poiss(\lambda \lowerDelayBound)}(1)} - p_{\sta}' \cdot ( \tempCost + \maxValue{\contMdp}+\epsSmall ) \\
p_{\sta}' \cdot ( \tempCost + \maxValue{\contMdp}+\epsSmall ) &\geq \tempCost - \frac{\pmf_{\poiss(\lambda \lowerDelayBound)}(\geq 2) \cdot (\maxValue{\contMdp}+\epsSmall)}{\pmf_{\poiss(\lambda \lowerDelayBound)}(1)} \\ 
p_{\sta}' &\geq \frac{\tempCost - \frac{\pmf_{\poiss(\lambda \lowerDelayBound)}(\geq 2) \cdot (\maxValue{\contMdp}+\epsSmall)}{\pmf_{\poiss(\lambda \lowerDelayBound)}(1)}}{\tempCost + \maxValue{\contMdp}+\epsSmall }.
\end{align*}

Observe that $\tempCost$ is actually a cost accumulated until getting to state in $\Sink$ from which an exponential transition is taken. We know that exponential transition is memoryless, thus the expected time until exponential transition occurs is $1/\lambda$. Thus it holds that $\tempCost \geq \minRew/\lambda - \maxRew \cdot \lowerDelayBound$. From conditional probability and plugging in the bound for $\tempCost$ we get:
\begin{align}
p_{\sta} &\geq \probm_{\fdC'(\timeouts)}(\timeChain^{\fdC'(\timeouts)}_1<\timeouts(\sta), \jumpChain^{\fdC'(\timeouts)}_1 + \jumpChain^{\fdC'(\timeouts)}_2 \geq \timeouts(\sta)) \cdot p_{\sta}' \nonumber \\
&\geq \pmf_{\poiss(\lambda \lowerDelayBound)}(1) \cdot \frac{\minRew/\lambda - \maxRew \cdot \lowerDelayBound - \frac{\pmf_{\poiss(\lambda \lowerDelayBound)}(\geq 2) \cdot (\maxValue{\contMdp}+\epsSmall)}{\pmf_{\poiss(\lambda \lowerDelayBound)}(1)}}{\minRew/\lambda - \maxRew \cdot \lowerDelayBound + \maxValue{\contMdp}+\epsSmall } \label{eq:probBoundA}
\end{align}
We now simplify the equation \eqref{eq:probBoundA} by inserting the value for $\lowerDelayBound$ in it. We first evaluate following:
\begin{align}
\frac{\minRew} {4 \cdot \lambda} - \maxRew \cdot \lowerDelayBound &= \frac{\minRew} {4 \cdot \lambda} - \maxRew \cdot \frac{\epsSmall \cdot \minRew}{4 \cdot \lambda^2 \cdot (\maxValue{\contMdp} + \epsSmall + \maxRew \cdot (2 + 2/\lambda))^2} \nonumber \\
&\geq \frac{\minRew} {4 \cdot \lambda} - \maxRew \cdot \frac{\epsSmall \cdot \minRew}{4 \cdot \lambda^2 \cdot (\epsSmall +  \maxRew /\lambda)^2} \nonumber \\
&\geq \frac{\minRew} {4 \cdot \lambda} - \maxRew \cdot \frac{\epsSmall \cdot \minRew}{4 \cdot \lambda^2 \cdot (2 \cdot \epsSmall \cdot \maxRew /\lambda)} \nonumber \\
&\geq 0,\label{eq:probBoundAA}
\end{align} 
Now we bound the term $\frac{\pmf_{\poiss(\lambda \lowerDelayBound)}(\geq 2) \cdot (\maxValue{\contMdp}+\epsSmall)}{\pmf_{\poiss(\lambda \lowerDelayBound)}(1)}$ from equation \eqref{eq:probBoundA}:
\begin{align}
\frac{\pmf_{\poiss(\lambda \lowerDelayBound)}(\geq 2) \cdot (\maxValue{\contMdp}+\epsSmall)}{\pmf_{\poiss(\lambda \lowerDelayBound)}(1)} &= \frac{(1- \e^{-\lambda \lowerDelayBound} - \lambda \lowerDelayBound \cdot \e^{-\lambda \lowerDelayBound}) \cdot (\maxValue{\contMdp}+\epsSmall)}{\lambda \lowerDelayBound \cdot \e^{-\lambda \lowerDelayBound}} \nonumber \\
&= \frac{(\e^{\lambda \lowerDelayBound} - 1 - \lambda \lowerDelayBound ) \cdot (\maxValue{\contMdp}+\epsSmall)}{\lambda \lowerDelayBound } \label{eq:probBoundB} \\
&\leq \frac{(\lambda \lowerDelayBound \cdot \e^{\lambda \lowerDelayBound} - \lambda \lowerDelayBound ) \cdot (\maxValue{\contMdp}+\epsSmall)}{\lambda \lowerDelayBound } \nonumber \\
&= (\e^{\lambda \lowerDelayBound} - 1 ) \cdot (\maxValue{\contMdp}+\epsSmall) \label{eq:probBoundC}
\end{align}
We use Lemma~\ref{lem:expFuncBound} in the inequality \eqref{eq:probBoundB}. In inequality~\eqref{eq:probBoundC} we plug in the definition of $\lowerDelayBound$:
\begin{align}
\frac{\pmf_{\poiss(\lambda \lowerDelayBound)}(\geq 2) \cdot (\maxValue{\contMdp}+\epsSmall)}{\pmf_{\poiss(\lambda \lowerDelayBound)}(1)} &\leq (\e^{\lambda \lowerDelayBound} - 1 ) \cdot (\maxValue{\contMdp}+\epsSmall) \nonumber \\
&= \frac{\epsSmall \cdot \minRew}{4 \cdot \lambda \cdot (\maxValue{\contMdp} + \epsSmall + \maxRew \cdot (2 + 2/\lambda))^2} \cdot (\maxValue{\contMdp}+\epsSmall) \nonumber\\
&\leq \frac{\epsSmall \cdot \minRew}{4 \cdot \lambda \cdot (\maxValue{\contMdp} + \epsSmall)^2} \cdot (\maxValue{\contMdp}+\epsSmall) \nonumber\\
&= \frac{\epsSmall \cdot \minRew}{4 \cdot \lambda \cdot (\maxValue{\contMdp} + \epsSmall)} \nonumber\\
&\leq \frac{\minRew}{4 \cdot \lambda },\label{eq:probBoundD}
\end{align}
Finally we plug inequalities \eqref{eq:probBoundAA} and \eqref{eq:probBoundD} into inequality \eqref{eq:probBoundA} to finish the proof
\begin{align*}
p_{\sta} &\geq \pmf_{\poiss(\lambda \lowerDelayBound)}(1) \cdot \frac{\minRew/\lambda - \maxRew \cdot \lowerDelayBound - \frac{\pmf_{\poiss(\lambda \lowerDelayBound)}(\geq 2) \cdot (\maxValue{\contMdp}+\epsSmall)}{\pmf_{\poiss(\lambda \lowerDelayBound)}(1)}}{\minRew/\lambda - \maxRew \cdot \lowerDelayBound + \maxValue{\contMdp}+\epsSmall } \\
&\geq \pmf_{\poiss(\lambda \lowerDelayBound)}(1) \cdot \frac{\minRew/(2\cdot \lambda)}{\minRew/\lambda - \maxRew \cdot \lowerDelayBound + \maxValue{\contMdp}+\epsSmall } \\
&\geq \pmf_{\poiss(\lambda \lowerDelayBound)}(1) \cdot \frac{\minRew/(2\cdot \lambda)}{\minRew/\lambda + \maxValue{\contMdp}+\epsSmall}.
\end{align*} 
\end{proof}

Now we have all necessary to prove Lemma~\ref{lem:delay-increase}. 
\begin{reflemma}{lem:delay-increase}
Let $\epsSmall>0$ be arbitrary, and let $\timeouts$ be any globally $\epsSmall$-optimal delay function in $\contMdp$ with $k>0$ bad sinks. Then there is an globally  $2\epsSmall$-optimal delay function $\timeoutsInfl$ in $\contMdp$ with $k-1$ bad sinks.
\end{reflemma}
\begin{proof}
First we fix sink $\Sink$ bad for $\timeouts$ and $\sta \in \sink$ that has minimal $\sinkValue_{\contMdp(\timeouts)}$ in $\Sink$. We will get rid of $\Sink$ by inflating delay of $\sta$ to $\lowerDelayBound$. But we will not do it in one turn but by gradually changing $\timeouts$ to $\timeouts[s/\lowerDelayBound/1], \timeouts[s/\lowerDelayBound/2], \ldots$ and showing that we can bound the change in value of any state. First we will provide the necessary random variable: 
For each $n \in \Nset $ we denote by $\costFdCC_n^{\sta}$ the random variable assigning to each run $\sta_0,a_1 \sta_1 \cdots$ the \emph{total cost before $n$th reach of $\sta$}, given by
\begin{align*}
\costFdCC_n(\vtx_0 a_1 \vtx_1 a_2 \vtx_2\cdots ) = \sum_{i=0}^{\hit^{ \{\sta \} }_n} \mcost(\vtx_i,a_{i+1}).
\end{align*}

Now for each $\sta' \in \statesM$ and $i \in \Nset$ holds
\begin{align*}
\expred_{\contMdp[\sta'](\timeouts[s/\lowerDelayBound/i])} &= \probm_{\contMdp[\sta']\timeouts[s/\lowerDelayBound/i])}(\jumpChain^{\mdp}_{\hit^{\{ \sta \} }_i} = \sta) \cdot \expect_{\contMdp[\sta']\timeouts[s/\lowerDelayBound/i])}[\costFdC \mid \jumpChain^{\mdp}_{\hit^{\{ \sta \} }_i} = \sta]\\
&\quad+ \probm_{\contMdp[\sta']\timeouts[s/\lowerDelayBound/i])}(\jumpChain^{\mdp}_{\hit^{\{ \sta \} }_i} \neq \sta) \cdot \expect_{\contMdp[\sta']\timeouts[s/\lowerDelayBound/i])}[\costFdC \mid \jumpChain^{\mdp}_{\hit^{\{ \sta \} }_i} \neq \sta] \\
&= \probm_{\contMdp[\sta']\timeouts[s/\lowerDelayBound/i])}(\jumpChain^{\mdp}_{\hit^{\{ \sta \} }_i} = \sta) \cdot 
(\expect_{\contMdp[\sta']\timeouts[s/\lowerDelayBound/i])}[\costFdCC_i \mid \jumpChain^{\mdp}_{\hit^{\{ \sta \} }_i} \neq \sta]+ \expred_{\contMdp[\sta](\timeouts)})\\
&\quad+ \probm_{\contMdp[\sta']\timeouts[s/\lowerDelayBound/i])}(\jumpChain^{\mdp}_{\hit^{\{ \sta \} }_i} \neq \sta) \cdot \expect_{\contMdp[\sta']\timeouts[s/\lowerDelayBound/i])}[\costFdC \mid \jumpChain^{\mdp}_{\hit^{\{ \sta \} }_i} \neq \sta].
\end{align*}
Similarly we get
\begin{align*}
\expred_{\contMdp[\sta']\timeouts[s/\lowerDelayBound/i+1])} &= \probm_{\contMdp[\sta']\timeouts[s/\lowerDelayBound/i])}(\jumpChain^{\mdp}_{\hit^{\{ \sta \} }_i} = \sta) \cdot \expect_{\contMdp[\sta']\timeouts[s/\lowerDelayBound/i+1])}[\costFdC \mid \jumpChain^{\mdp}_{\hit^{\{ \sta \} }_i} = \sta]\\
&\quad+ \probm_{\contMdp[\sta']\timeouts[s/\lowerDelayBound/i])}(\jumpChain^{\mdp}_{\hit^{\{ \sta \} }_i} \neq \sta) \cdot \expect_{\contMdp[\sta']\timeouts[s/\lowerDelayBound/i+1])}[\costFdC \mid \jumpChain^{\mdp}_{\hit^{\{ \sta \} }_i} \neq \sta] \\
&= \probm_{\contMdp[\sta']\timeouts[s/\lowerDelayBound/i])}(\jumpChain^{\mdp}_{\hit^{\{ \sta \} }_i} = \sta) \cdot \expect_{\contMdp[\sta']\timeouts[s/\lowerDelayBound/i+1])}[\costFdC \mid \jumpChain^{\mdp}_{\hit^{\{ \sta \} }_i} = \sta]\\
&\quad+ \probm_{\contMdp[\sta']\timeouts[s/\lowerDelayBound/i])}(\jumpChain^{\mdp}_{\hit^{\{ \sta \} }_i} \neq \sta) \cdot \expect_{\contMdp[\sta']\timeouts[s/\lowerDelayBound/i])}[\costFdC \mid \jumpChain^{\mdp}_{\hit^{\{ \sta \} }_i} \neq \sta] \\
&= \probm_{\contMdp[\sta']\timeouts[s/\lowerDelayBound/i])}(\jumpChain^{\mdp}_{\hit^{\{ \sta \} }_i} = \sta) \\
&\quad\cdot (\expect_{\contMdp[\sta']\timeouts[s/\lowerDelayBound/i+1])}[\costFdCC_i \mid \jumpChain^{\mdp}_{\hit^{\{ \sta \} }_i} \neq \sta]+ \expred_{\contMdp[\sta]\timeouts[s/\lowerDelayBound/1])})\\
&\quad+ \probm_{\contMdp[\sta']\timeouts[s/\lowerDelayBound/i])}(\jumpChain^{\mdp}_{\hit^{\{ \sta \} }_i} \neq \sta) \cdot \expect_{\contMdp[\sta']\timeouts[s/\lowerDelayBound/i])}[\costFdC \mid \jumpChain^{\mdp}_{\hit^{\{ \sta \} }_i} \neq \sta]\\
&= \probm_{\contMdp[\sta']\timeouts[s/\lowerDelayBound/i])}(\jumpChain^{\mdp}_{\hit^{\{ \sta \} }_i} = \sta) \\
&\quad\cdot (\expect_{\contMdp[\sta']\timeouts[s/\lowerDelayBound/i])}[\costFdCC_i \mid \jumpChain^{\mdp}_{\hit^{\{ \sta \} }_i} \neq \sta]+ \expred_{\contMdp[\sta]\timeouts[s/\lowerDelayBound/1])})\\
&\quad+ \probm_{\contMdp[\sta']\timeouts[s/\lowerDelayBound/i])}(\jumpChain^{\mdp}_{\hit^{\{ \sta \} }_i} \neq \sta) \cdot \expect_{\contMdp[\sta']\timeouts[s/\lowerDelayBound/i])}[\costFdC \mid \jumpChain^{\mdp}_{\hit^{\{ \sta \} }_i} \neq \sta].
\end{align*}
Now subtracting to find the increase of expected cost we get
\begin{align}
\expred_{\contMdp[\sta']\timeouts[s/\lowerDelayBound/i+1])}^{\sta'} - \expred_{\contMdp[\sta']\timeouts[s/\lowerDelayBound/i])}^{\sta'} &= \probm_{\contMdp[\sta']\timeouts[s/\lowerDelayBound/i])}(\jumpChain^{\mdp}_{\hit^{\{ \sta \} }_i} = \sta) \cdot (\expred_{\contMdp[\sta](\timeouts)} - \expred_{\contMdp[\sta](\timeouts[s/\lowerDelayBound/1])})\label{eq:nonZeroPartA}\\ 
&= (1-p_{\sta})^i \cdot (\expred_{\contMdp[\sta](\timeouts)} - \expred_{\contMdp[\sta](\timeouts[s/\lowerDelayBound/1])}), \label{eq:nonZeroPartB}
\end{align}
where in last inequality we used Lemma~\ref{lem:possitive-leaving-probability}. Obviously for $i=0$ we get
\begin{align}
\expred_{\contMdp[\sta']\timeouts[s/\lowerDelayBound/i+1])} - \expred_{\contMdp[\sta']\timeouts[s/\lowerDelayBound/i])} &= \expred_{\contMdp[\sta']\timeouts[s/\lowerDelayBound/1])} - \expred_{\contMdp[\sta'](\timeouts)} \nonumber \\
&= (1-p_{\sta})^0 \cdot (\expred_{\contMdp[\sta](\timeouts)} - \expred_{\contMdp[\sta]\timeouts[s/\lowerDelayBound/1])}). \label{eq:zeroPart}
\end{align}
Finally putting together \eqref{eq:nonZeroPartA}, \eqref{eq:nonZeroPartB}, \eqref{eq:zeroPart}, Lemma~\ref{lem:possitive-leaving-probability}, Lemma~\ref{lem:one-step-error} and summing up the geometric series, we get the final result
\begin{align*}
\expred_{\contMdp[\sta'](\timeouts[s/\lowerDelayBound/\infty])} - \expred_{\contMdp[\sta'](\timeouts)} &= \sum_{i=0}^{\infty} (1-p_{\sta})^i \cdot (\expred_{\contMdp[\sta](\timeouts)} - \expred_{\contMdp[\sta](\timeouts[s/\lowerDelayBound/1])})\\
&= 1/p_{\sta} \cdot (\expred_{\contMdp[\sta](\timeouts)} - \expred_{\contMdp[\sta](\timeouts[s/\lowerDelayBound/1])})\\
&\leq \frac{(\minRew/\lambda + \maxValue{\contMdp} + \epsSmall)\cdot \lambda}{2 \cdot \minRew \cdot \pmf_{\poiss(\lambda \lowerDelayBound)}(1)} \\
&\quad\cdot \lambda \lowerDelayBound \cdot \pmf_{\poiss(\lambda \lowerDelayBound)}(\geq 1) \cdot \Big (\maxValue{\contMdp} + \eps + \maxRew \cdot \Big (2 +\frac{2}{\lambda} \Big) \Big) \\
&\leq \frac{(\maxRew/\lambda + \maxValue{\contMdp} + \epsSmall)}{2 \cdot \minRew \cdot \lambda \lowerDelayBound \cdot \e^{-\lambda \lowerDelayBound}} \\
&\quad\cdot \lambda^2 \lowerDelayBound \cdot (1-\e^{-\lambda \lowerDelayBound }) \cdot \Big (\maxValue{\contMdp} + \eps + \maxRew \cdot \Big (2 +\frac{2}{\lambda} \Big) \Big) \\
&\leq \frac{(\e^{\lambda \lowerDelayBound}-1) \cdot \lambda \cdot \Big (\maxValue{\contMdp} + \eps + \maxRew \cdot \Big (2 +\frac{2}{\lambda} \Big) \Big)^2}{2 \cdot \minRew}
\end{align*} 
Finally we plug in the definition of $\lowerDelayBound$
$$\e^{\lambda \cdot \lowerDelayBound} = 1 + \frac{\epsSmall \cdot \minRew}{4 \cdot \lambda \cdot (\maxValue{\contMdp} + \epsSmall + \maxRew \cdot (2 + 2/\lambda))^2}$$
to the last inequality and we get:
\begin{align*}
\expred_{\contMdp(\timeouts[s/\lowerDelayBound/\infty])}^{\sta'} - \expred_{\contMdp[\sta'](\timeouts)} &\leq \frac{(\e^{\lambda \lowerDelayBound}-1) \cdot \lambda \cdot \Big (\maxValue{\contMdp} + \eps + \maxRew \cdot \Big (2 +\frac{2}{\lambda} \Big) \Big)^2}{2 \cdot \minRew} \\
&\leq \frac{\epsSmall}{8}.
\end{align*}
\end{proof}
\subsection{Proof of Proposition~\ref{prop:mesh-delays-ok}}

\begin{refproposition}{prop:mesh-delays-ok}
For $\constFactor$ from Lemma~\ref{lem:bounded-step-existence}, $\discconst$
and $\cutconst$ from Lemma~\ref{lem:mesh-error}, it holds that
$$\left\lvert \; \Value{\contMdp} - \Value{\contMdp,\paramspace(\delta,\dmax)} \; \right\rvert 
\;\; \leq \;\; \frac{\eps}{2}$$
\[ \textrm{where }~~
\delta \; :=\; \frac{\alpha}{\discconst}, ~~ \dmax\; := \; |\log(\alpha)|\cdot \cutconst\cdot (\maxValue{\contMdp}+\eps), ~~ \alpha \; := \; \frac{\eps^2}{64\constFactor\cdot|\states'|\cdot (1+\maxValue{\contMdp})^2}.
\]
\end{refproposition}
\begin{proof}
Let $\timeouts'$ be the globally $\eps/4$-optimal delay function from Lemma~\ref{lem:bounded-step-existence} (obtained by putting $\eps'=\eps/2$). From the choice of parameters $\delta$ and $\dmax$ and from Lemma~\ref{lem:mesh-error} it follows that $\paramspace(\delta,\dmax)$ contains a function $\timeouts$ that is $\alpha$-bounded by $\timeouts'$. From Lemma~\ref{lem:perturbation-error} it follows that the cost incurred by $\timeouts$ (from any initial state) differs from the cost incurred by $\timeouts'$ by at most
\begin{align*}
& 2\alpha \frac{N\cdot \CostBound{\mcost,\timeouts'}}{\eps}(1+\CostBound{\mcost,\timeouts'}|\states'|) \\
&\quad\leq \frac{8\eps\cdot(1+\eps+\maxValue{\contMdp} )}{64(1+\maxValue{\contMdp})}\leq \frac{\eps}{8}\left(1+\eps\right) \leq \eps/4.
\end{align*}
Hence, $\timeouts\in \paramspace(\delta,\dmax)$ is $\eps/2$-optimal in every state $\contMdp$, from which the Proposition follows. (Note that $\Value{\contMdp,\paramspace(\delta,\dmax)}\geq \Value{\contMdp}$.)
\end{proof}
\subsection{Proof of Lemma~\ref{lem:bound-value}}

\begin{reflemma}{lem:bound-value}
There is a number $\maxvalconst\in \exp(\size{\fdC}^{\mathcal{O}(1)})$ computable in time polynomial in $\size{\fdC}$ such that $\maxValue{\contMdp} \leq \maxvalconst$.
\end{reflemma}

\begin{proof}
We simply choose a appropriate delay function and compute the overestimation of expected cost until reaching goal states from any state $\sta \in \statesM$. We pick a delay function $\timeouts$ such that for each $\sta \in \statesM$ it holds that $\timeouts(\sta) = \frac{|\allact|}{\lambda}$.

We first bound the $\mcost(\sta,\timeouts(\sta))$ for each state $\sta \in \statesM$. Obviously for $\sta \in \statesM \setminus \allact$ the expected impulse cost paid in one step is $\maxRew$ and expected rate cost paid in one step is $\maxRew/\lambda$, thus together $\mcost(\sta,\timeouts(\sta)) \leq (1+1/\lambda) \cdot \maxRew$. For state $\sta \in \statesM \cap \allact$ the expected number of exponential steps until reaching $\statesM$ in time $\timeouts(\sta) = \frac{|\allact|}{\lambda}$ is $\timeouts(\sta) \cdot \lambda = |\allact|$ (employing mean of Poisson distribution). It is possible that one extra step is taken due to fixed-delay transition, thus the expected impulse cost is bounded by $(|\allact|+1) \cdot \maxRew$. Finally the expected rate cost paid until state from $\statesM$ is reached from any $\statesM \cap \allact$ can be bounded from above by $\timeouts(\sta) \cdot \maxRew = \frac{|\allact|}{\lambda} \cdot \maxRew$, what is maximum time until $\statesM$ are reached $\timeouts(\sta)= \frac{|\allact|}{\lambda}$ (from definition of $\mcost$ after firing of fixed-delay transition a state from $\statesM$ is reached) times maximum rate cost bounded by $\maxRew$. Altogether we have:
\begin{align*}
\mcost(\sta,\timeouts(\sta)) &\leq \max \{ (1+1/\lambda) \cdot \maxRew, |\allact|/\lambda \cdot \maxRew + (|\allact|+1) \cdot \maxRew \} \\
&\leq (|\allact|/\lambda + |\allact| + 1) \cdot \maxRew.
\end{align*}

Now we bound the minimal branching probability of $\mtran(s,|\allact|/\lambda)$. We use again uniformization method to compute the transient analysis of $\fdC'[\sta](\timeouts(\sta))$ to evaluate $\mtran(s,|\allact|/\lambda)$. For all $\sta \in \statesM \cap \allact$ and $\sta' \in \statesM$ that is reachable from $\overline{\sta}$ in $\fdC'[\sta](\timeouts(\sta))$ ($\overline{\sta}$ is artificial initial state in $\fdC'[\sta](\timeouts(\sta))$ corresponding to state $\sta$) it holds that
\begin{align}
\mtran(s,|\allact|/\lambda)(s') &= \e^{-|\allact|} \sum_{i=0}^{\infty} \prob'^{i}(\overline{\sta},\sta') \cdot \frac{|\allact|^{i}}{i!} \nonumber \\
&\geq \e^{-|\allact|} \sum_{i=0}^{|\allact|} \prob'^{i}(\overline{\sta},\sta') \cdot \frac{|\allact|^{i}}{i!} \nonumber \\
&\geq \e^{-|\allact|} \min_{i \in \{0, \ldots, |\allact| \} } \minPst^{|\allact|} \cdot \frac{|\allact|^{i}}{i^i} \label{eq:minimalNumber-2} \\
&\geq \e^{-|\allact|} \minPst^{|\allact|} \nonumber \\
&=(\minPst /\e)^{|\allact|} \nonumber,
\end{align}
where \eqref{eq:minimalNumber-2} holds because the minimal positive branching probability in $\fdC'[\sta](\timeouts(\sta))$ is $\minPst$ and each state, if it is reachable in $\fdC'[\sta](\timeouts(\sta))$ from $\overline{\sta}$ it has to be reached by path of length at most $|\allact|$. Thus the minimum discrete probability to reach any reachable state is at least $\minPst^{|\allact|}$. This probability is at least once multiplied by a factor $\frac{|\allact|^{i}}{i^i}$ for appropriate $i$ thus we pick the smallest factor to get the correct lower bound on probability of moving from $\statesM \cap \allact$. The probability of moving from $\statesM \setminus \allact$ is simply $\minPst$ from definition of $\mtran(s,|\allact|/\lambda)$. Since $(\minPst /\e)^{|\allact|} \leq \minPst$ it holds that for all $\sta, \sta' \in \statesM$ such that $\mtran(s,|\allact|/\lambda)(\sta') > 0$ implies 
$$\mtran(s,|\allact|/\lambda)(\sta') \geq (\minPst /\e)^{|\allact|}.$$
Now we provide a lower bound on probability of reaching $\goalStates$ from any $\sta \in \statesM$:
\begin{align*}
\probm_{\mdp(\timeouts)}(\text{reach $\goalStates$ from $\sta$ in $|\states'|$ steps}) &\geq (\min \; \{ \mtran(\sta',|\allact|/\lambda)(\sta'') >0 \mid \sta',\sta'' \in \statesM \}\;)^{|\states'|} \\
&\geq \Big( (\minPst /\e)^{|\allact|} \Big) ^{|\states'|} \\
&\geq (\minPst /\e)^{|\states|^2}.
\end{align*}
Finally we bound the expected cost using Bernoulli trials, i.e. for all $\sta \in \statesM$ it holds that
\begin{align}
\expred_{\contMdp[s](\timeouts)}^{\mcost} &\leq \sum_{i=i}^{\infty} i \cdot |\statesM| \cdot \max_{\sta'' \in \statesM} \mcost(\sta'',\timeouts(\sta'')) \cdot (1 - pr)^{i-1} \cdot pr \nonumber \\
&= \frac{|\statesM| \cdot \max_{\sta'' \in \statesM} \mcost(\sta'',\timeouts(\sta''))} {pr}\nonumber\\
&\leq \frac{|\statesM| \cdot (|\allact|+1 + |\allact|/\lambda)\cdot \maxRew}{(\minPst/\e)^{|\states|^2}},\label{eq:valbound-final-unc}
\end{align}
where $pr = \min_{\sta'' \in \statesM}\probm_{\mdp(\timeouts)}(\text{reach $\goalStates$ from $\sta$ in $|\states \setminus \allact|$ steps})$. The number $\maxvalconst$ can be taken to be the right hand side of~\eqref{eq:valbound-final-unc}.
\end{proof}

\subsection{Proof of Proposition~\ref{prop:rounding-error}}
\begin{refproposition}{prop:rounding-error}
 Let $\eps>0$ and fix $\kappa=(\eps\cdot\delta\cdot \minRew)/({2\cdot |\states'|\cdot(1+\maxValue{\contMdp})^2}),$ where $\minRew$ is a minimal cost rate in $\fdC$.
	\label{prop:mdp-formulation-22}
	Then it holds
	$$ \left\lvert \; \Value{\contMdp,\paramspace(\delta,\dmax)} - \Value{\matchmdp} \; \right\rvert 
	\;\; \leq \;\; \frac{\eps}{2}.$$
\end{refproposition}

Let $\timeouts$ be any delay function in $\paramspace(\delta,\dmax)$. Note that $\timeouts$ can be taken as a delay function both in $\contMdp$ and $\matchmdp$. To distinguish these two contexts we denote this function by $\timeouts_\kappa$ when viewing it as a strategy in $\mdp_\kappa$. These ``two'' functions $\kappa$-bound one another, i.e. for all states $s$ and $t$ it holds
\begin{enumerate} 
	\item  $|\mtran_\kappa(s,\timeouts_\kappa(s))(t)-\mtran(s,\timeouts(s))(t)|\leq \kappa$ and
	\item $|\mcost_\kappa(s,\timeouts_\kappa(s))-\mcost(s,\timeouts(s))|\leq \kappa$,
\end{enumerate}
with the qualitative properties of transitions being preserved. Hence, we can use exactly the same arguments as in the proof of Lemma~\ref{lem:perturbation-error} to show that the difference between costs incurred by these strategies is bounded by these two expressions:
\begin{align}
E_1 &= 2\kappa \cdot\CostBound{\msteps,\timeouts}(1+\CostBound{\mcost,\timeouts}\cdot|\states'|)\label{eq:final-1}\\
E_2 &= 2\kappa \cdot \CostBound{\msteps,\timeouts_\kappa}(1+\CostBound{\mcost_\kappa,\timeouts_\kappa}\cdot|\states'|).\label{eq:final-2}
\end{align}

Now let $\timeouts\in\paramspace(\delta,\dmax)$ be any $(\eps/8)$-optimal strategy in $\mdp$. Then $\CostBound{\mcost,\timeouts}\leq \maxValue{\mdp}+\eps/8$. Moreover, since minimal delay used in $\timeouts$ is equal to $\delta$, the minimal one-step cost incurred by $\timeouts$ is equal to $((1-e^{-\lambda\delta})/\lambda)\cdot \minRew$ (this is because the average time elapsed before either firing an exponential transition or the timeout $\delta$ running out is $(1-e^{-\lambda\delta})/\lambda$). Hence, on average, $\timeouts$ performs at most $\frac{(\maxValue{\mdp}+\eps)\cdot \lambda}{(1-e^{-\lambda\delta})\minRew}\leq \frac{2(\maxValue{\mdp}+\eps) }{\delta\minRew} $ steps. Plugging this into~\eqref{eq:final-1} yields

\begin{align*}
E_1 &\leq \frac{2\eps (1+\eps/8 + \maxValue{\mdp})^2}{2(1 + \maxValue{\mdp})^2} \leq \eps/4.
\end{align*}

It follows that $  \Value{\matchmdp} - \Value{\contMdp,\paramspace(\delta,\dmax)}    \leq  \frac{\eps}{8}+\frac{\eps}{4}\leq\frac{\eps}{2}$. It now suffices to show that $   \Value{\contMdp,\paramspace(\delta,\dmax)}-\Value{\matchmdp} \leq \frac{\eps}{2}$.

To this end, let $\timeouts_\kappa$ be any $(\eps/8)$-optimal strategy in $\mdp_\kappa$. Then $\CostBound{\mcost_\kappa,\timeouts_\kappa}\leq \maxValue{\mdp_\kappa}+\eps/8$. Moreover, since $\mdp_\kappa \in \paramspace(\delta,\dmax)$, using the same arguments as above (and the fact that rounding could only increase the one-step cost incurred) we get the same lower bound on one-step cost incurred by $\timeouts_\kappa$ and thus also the fact that $\CostBound{\msteps,\timeouts_\kappa}\leq  \frac{2(\maxValue{\mdp}+\eps) }{\delta\minRew} $. Plugging this into~\eqref{eq:final-2} and using the same computation as above, we get the desired inequality. \qed
\section{Proofs for Section~\ref{sec:results-multi}}
\label{app:bound-sec-4}

\begin{refproposition}{prop:multi-entry-aprox}
There is a number $\POconst \in \exp((\size{\fdC}\cdot\size{\dmin}\cdot \dmax)^{\mathcal{O}(1)})$ such that for $\delta= \eps/\POconst$ and $\kappa=(\eps\cdot \delta)/\POconst$ it holds	$\left\lvert\Value{\contMdp,\paramspacePOShort}-\Value{\matchmdp_D} \right\rvert < \varepsilon.$
\end{refproposition}

Similarly to Section~\ref{sec:results-single} we first bound the error caused by discretization and then by rounding:

\begin{lemma}\label{lem:mesh-delays-ok}
Let 
\[
\alpha \; := \; \frac{\eps}{8 \cdot \CostBound{\mcost,\timeouts} \cdot \CostBound{\msteps,\timeouts} \cdot |\statesM|}, \qquad \delta \; :=\; \alpha/\discconst,
\]
where $\discconst$ and is from Lemma~\ref{lem:mesh-error}. Then 
$$ \left\lvert \; \Value{\contMdp} - \Value{\contMdp,\paramspace(\delta,\dmin, \dmax, \equiv)} \; \right\rvert 
\;\; \leq \;\; \frac{\eps}{2}.$$
\end{lemma}

\begin{lemma}\label{lem:rounding}
 Let $\eps>0$ and fix $\kappa=(\eps/8 \cdot \CostBound{\mcost,\timeouts} \cdot \CostBound{\msteps,\timeouts} \cdot |\statesM|)^2$. Then it holds that
	$$ \left\lvert \; \Value{\contMdp,\paramspace(\delta,\dmin, \dmax, \equiv)} - \Value{\matchmdp} \; \right\rvert 
	\;\; \leq \;\; \frac{\eps}{2}.$$
\end{lemma}
Note that from Lemma~\ref{lem:mesh-delays-ok} and Lemma~\ref{lem:rounding} directly follows Proposition~\ref{prop:multi-entry-aprox}.

\subsection{Proof of Lemma~\ref{lem:mesh-delays-ok}}
To prove Lemma~\ref{lem:mesh-delays-ok} we similarly to proof of Proposition~\ref{prop:mesh-delays-ok} we will need bounds on expected number of steps and expected cost until reaching target. We will prove the following lemmas.

\begin{lemma}\label{lem:value-bound}
Let $\timeouts $ be a delay function such that for each $\sta \in \statesM$ it holds that $\dmin \leq \timeouts(\sta) \leq \dmax$. It holds that  $\CostBound{\mcost,\timeouts} \leq \left( \frac{|\allact|}{\minPst \cdot \min \{1, \lambda \cdot \delayBound \}} \right)^{|\states|} \cdot \e^{\lambda \dmax \cdot |\states \setminus \allact|} \cdot (\max\{ 1/\lambda, \; \dmax \cdot \lambda + \dmax \} + 1) \cdot |\statesM| \cdot \maxRew$. 
\end{lemma}
We first prove that there is a lower bound to reach a target from any state.

\begin{lemma} \label{lem:bounded-prob}
Let $\contMdp$ be DTMDP and $\dmin, \dmax >0$ be lower, upper be bound on delay function, respectively. Let $\timeouts$ be any delay function such that for each state $\sta \in \statesM$ it holds that $\dmin \leq \timeouts(\sta) \leq \dmax$. 
Then for all $\sta \in \statesM$ it holds that
$$\probm_{\mdp(\timeouts)}(\text{reach $\goalStates$ from $\sta$ in $|\states \setminus \allact|$ steps}) \geq \left(  \Big ( \frac{\minPst}{|\allact|} \Big )^{|\allact|} \cdot \e^{-\lambda \dmax} \cdot \min \{ 1, (\lambda \cdot \dmin)^{|\allact|} \} \right) ^{|\states \setminus \allact|} $$
\end{lemma}

\begin{proof}
We first bound the minimal branching probability of $\mtran(s,|\allact|/\lambda)$. We use again uniformization method to compute the transient analysis of $\fdC'[\sta](\timeouts(\sta))$ to evaluate $\mtran(s,|\allact|/\lambda)$. For all $\sta \in \statesM \cap \allact$ and $\sta' \in \statesM$ that is reachable from $\overline{\sta}$ in $\fdC'[\sta](\timeouts(\sta))$ ($\overline{\sta}$ is artificial initial state in $\fdC'[\sta](\timeouts(\sta))$ corresponding to state $\sta$) it holds that
\begin{align}
\mtran(s,\timeouts(\sta))(s') &= \e^{-\lambda \timeouts(\sta)} \sum_{i=0}^{\infty} \prob'^{i}(\overline{\sta},\sta') \cdot \frac{(\lambda \cdot \timeouts(\sta))^{i}}{i!} \nonumber \\
&\geq \e^{-\lambda \timeouts(\sta)} \sum_{i=0}^{|\allact|} \prob'^{i}(\overline{\sta},\sta') \cdot \frac{(\lambda \cdot \timeouts(\sta))^{i}}{i!} \nonumber \\
&\geq \e^{-\lambda \dmax} \sum_{i=0}^{|\allact|} \prob'^{i}(\overline{\sta},\sta') \cdot \frac{(\lambda \cdot \timeouts(\sta))^{i}}{i!} \nonumber \\
&\geq \e^{-\lambda \dmax} \min_{i \in \{0, \ldots, |\allact| \} } \minPst^{|\allact|} \cdot \frac{(\lambda \cdot \timeouts(\sta))^{i}}{i^i} \label{eq:minimalNumber} \\
&\geq \e^{-\lambda \dmax} \min_{i \in \{0, \ldots, |\allact| \} } \minPst^{|\allact|} \cdot \frac{(\lambda \cdot \timeouts(\sta))^{i}}{|\allact|^{|\allact|}} \nonumber \\
&= \Big ( \frac{\minPst}{|\allact|} \Big )^{|\allact|} \cdot \e^{-\lambda \dmax} \cdot \min_{i \in [0,|\allact|]}(\lambda \cdot \dmin)^{i} \nonumber \\
&= \Big ( \frac{\minPst}{|\allact|} \Big )^{|\allact|} \cdot \e^{-\lambda \dmax} \cdot \min \{ 1, (\lambda \cdot \dmin)^{|\allact|} \} \nonumber,
\end{align}
where \eqref{eq:minimalNumber} holds because the minimal positive branching probability in $\fdC'[\sta](\timeouts(\sta))$ is $\minPst$ and each state, if it is reachable in $\fdC'[\sta](\timeouts(\sta))$ from $\overline{\sta}$, it has to be reached by path of length at most $|\allact|$. Thus the minimum discrete probability to reach any reachable state is at least $\minPst^{|\allact|}$. This probability is at least once multiplied by a factor $\frac{(\lambda \cdot \timeouts(\sta))^{i}}{i^i}$ for appropriate~$i$, thus we pick the smallest factor to get the correct lower bound on probability of moving from $\statesM \cap \allact$. The minimal positive probability of moving from $\statesM \setminus \allact$ is simply $\minPst$ from definition of $\mtran(s,|\allact|/\lambda)$. Since $ \Big ( \frac{\minPst}{|\allact|} \Big )^{|\allact|} \cdot \e^{-\lambda \dmax} \cdot \min \{ 1, (\lambda \cdot \dmin)^{|\allact|} \} \leq \minPst$ it holds that for all $\sta, \sta' \in \statesM$ such that $\mtran(s,|\allact|/\lambda)(\sta') > 0$ implies 
$$\mtran(s,|\allact|/\lambda)(\sta') \geq  \Big ( \frac{\minPst}{|\allact|} \Big )^{|\allact|} \cdot \e^{-\lambda \dmax} \cdot \min \{ 1, (\lambda \cdot \dmin)^{|\allact|} \}.$$
Now we provide a lower bound on probability of reaching $\goalStates$ from any $\sta \in \statesM$:
\begin{align*}
\probm_{\mdp(\timeouts)}(\text{reach $\goalStates$ from $\sta$ in $|\states \setminus \allact|$ steps}) &\geq (\min \; \{ \mtran(\sta',|\allact|/\lambda)(\sta'') >0 \mid \sta',\sta'' \in \statesM \}\;)^{|\states \setminus \allact|} \\
&\geq \left(  \Big ( \frac{\minPst}{|\allact|} \Big )^{|\allact|} \cdot \e^{-\lambda \dmax} \cdot \min \{ 1, (\lambda \cdot \dmin)^{|\allact|} \} \right) ^{|\states \setminus \allact|}.
\end{align*}
\end{proof}

Now we are ready to prove Lemma~\ref{lem:bound-value}.

\begin{proof}
We choose an arbitrary delay function $\timeouts$, such that for all $\sta \in \statesM$ holds $\timeouts(\sta) =\delayBound$ for some $\dmin \leq \delayBound \leq \dmax$. We compute the overestimation of expected cost until reaching goal states from any state $\sta \in \statesM$.  During the computation we substitute $\dmin$ or $\dmax$ for $\delayBound$ to get the correct upper bound. 

We first bound the $\mcost(\sta,\timeouts(\sta))$ for each state $\sta \in \statesM$. Obviously for $\sta \in \statesM \setminus \allact$ the expected impulse cost paid in one step is $\maxRew$ and the expected rate cost paid in one step is $\maxRew/\lambda$, thus together $\mcost(\sta,\timeouts(\sta)) \leq (1+1/\lambda) \cdot \maxRew$. For state $\sta \in \statesM \cap \allact$ the expected number of exponential steps until reaching $\statesM$ in time $\timeouts(\sta) = \delayBound$ is $ \lambda \cdot \delayBound \leq \lambda \cdot \dmax$ (employing mean of Poisson distribution). It is possible that one extra step is taken due to fixed-delay transition, thus the expected impulse cost is bounded by $(\lambda \cdot \dmax +1 ) \cdot \maxRew$. Finally the expected rate cost paid until state from $\statesM$ is reached from any $\statesM \cap \allact$ can be bounded from above by $\timeouts(\sta) \cdot \maxRew \leq \dmax \cdot \maxRew$, what comprises of the maximum time $\dmax$ until $\statesM$ are reached (from definition of $\mcost$ after firing of fixed-delay transition a state from $\statesM$ is reached), and maximum rate cost bounded by $\maxRew$. Altogether we have:
\begin{align*}
\mcost(\sta,\timeouts(\sta)) &\leq \max \{ (1+1/\lambda) \cdot \maxRew, (\dmax \cdot \lambda+1) \cdot \maxRew + \dmax \cdot \maxRew\} \\
&\leq (\max\{\ 1/\lambda, \dmax \cdot \lambda + \dmax \} + 1) \cdot \maxRew.
\end{align*}

Now we bound the expected cost using Bernoulli trials, i.e. for all $\sta \in \statesM$ it holds that
\begin{align*}
&\expred_{\contMdp[s](\timeouts)}^{\mcost} \leq \\
&\leq \sum_{i=i}^{\infty} i \cdot |\statesM| \cdot \max_{\sta'' \in \statesM} \mcost(\sta'',\timeouts(\sta'')) \cdot (1 - pr)^{i-1} \cdot pr \\
&= \frac{|\statesM| \cdot \max_{\sta'' \in \statesM} \mcost(\sta'',\timeouts(\sta''))} {pr}\\
&\leq \frac{|\statesM| \cdot (\max\{\ 1/\lambda, \dmax \cdot \lambda + \dmax \} + 1) \cdot \maxRew}{\left(  \Big ( \frac{\minPst}{|\allact|} \Big )^{|\allact|} \cdot \e^{-\lambda \dmax} \cdot \min \{ 1, (\lambda \cdot \dmin)^{|\allact|} \} \right) ^{|\states \setminus \allact|}}\\
&= \left( \frac{|\allact|}{\minPst \cdot \min \{1, \lambda \cdot \dmin \}} \right)^{|\states|} \cdot \e^{\lambda \dmax \cdot |\states \setminus \allact|} \cdot (\max\{ 1/\lambda, \; \dmax \cdot \lambda + \dmax \} + 1) \cdot |\statesM| \cdot \maxRew,
\end{align*}
where using Lemma~\ref{lem:bounded-prob}
\begin{align*}
pr &= \min_{\sta'' \in \statesM}\probm_{\mdp(\timeouts)}(\text{reach $\goalStates$ from $\sta''$ in $|\states \setminus \allact|$ steps}) \\
&\geq \left(  \Big ( \frac{\minPst}{|\allact|} \Big )^{|\allact|} \cdot \e^{-\lambda \dmax} \cdot \min \{ 1, (\lambda \cdot \dmin)^{|\allact|} \} \right) ^{|\states \setminus \allact|}.
\end{align*}

\end{proof}

\begin{lemma}\label{lem:cost-bound}
Let $\timeouts $ be a delay function such that for each $\sta \in \statesM$ it holds that $\dmin \leq \timeouts(\sta) \leq \dmax$. It holds that  $\CostBound{\msteps,\timeouts} \leq \left( \frac{|\allact|}{\minPst \cdot \min \{1, \lambda \cdot \dmin \}} \right)^{|\states|} \cdot \e^{\lambda \dmax \cdot |\states \setminus \allact|} \cdot |\statesM|$.
\end{lemma}

\begin{proof}
We use Bernoulli trials and Lemma~\ref{lem:bounded-prob} to bound the expected number of steps from any state to a goal state, i.e. for all $\sta \in \statesM$ it holds that
\begin{align*}
&\expred_{\contMdp[s](\timeouts)}^{\msteps} \leq \\
&\leq \sum_{i=i}^{\infty} i \cdot |\statesM| \cdot (1 - \min_{\sta'' \in \statesM}\probm_{\mdp(\timeouts)}(\text{reach $\goalStates$ from $\sta''$ in $|\states \setminus \allact|$ steps}))^{i-1} \\
&\quad\cdot \min_{\sta'' \in \statesM}\probm_{\mdp(\timeouts)}(\text{reach $\goalStates$ from $\sta''$ in $|\states \setminus \allact|$ steps}) \\
&= \frac{|\statesM|} {\min_{\sta'' \in \statesM}\probm_{\mdp(\timeouts)}(\text{reach $\goalStates$ from $\sta''$ in $|\states \setminus \allact|$ steps})}\\
&\leq \frac{|\statesM|}{\left(  \Big ( \frac{\minPst}{|\allact|} \Big )^{|\allact|} \cdot \e^{-\lambda \dmax} \cdot \min \{ 1, (\lambda \cdot \dmin)^{|\allact|} \} \right) ^{|\states \setminus \allact|}}\\
&= \left( \frac{|\allact|}{\minPst \cdot \min \{1, \lambda \cdot \dmin \}} \right)^{|\states|} \cdot \e^{\lambda \dmax \cdot |\states \setminus \allact|} \cdot |\statesM|.
\end{align*}
\end{proof}

Now we have everything needed to prove Lemma~\ref{lem:mesh-delays-ok}

\begin{reflemma}	{lem:mesh-delays-ok}
Let 
\[
\alpha \; := \; \frac{\eps}{8 \cdot \CostBound{\mcost,\timeouts} \cdot \CostBound{\msteps,\timeouts} \cdot |\statesM|}, \qquad \delta \; :=\; \alpha/\discconst,
\]
where $\discconst$ and is from Lemma~\ref{lem:mesh-error}. Then 
$$ \left\lvert \; \Value{\contMdp} - \Value{\contMdp,\paramspace(\delta,\dmax)} \; \right\rvert 
\;\; \leq \;\; \frac{\eps}{2}.$$
\end{reflemma}

\begin{proof}
Please observe that Lemma~\ref{lem:perturbation-error} and Lemma~\ref{lem:mesh-error} from Section~\ref{sec:results-single} hold also for bounded optimization under partial observation. Moreover Lemma~\ref{lem:mesh-error} is more general than we need here as we can use upper given bound $\dmax$. Now we just insert equations for $\alpha$, $\CostBound{\mcost,\timeouts}$, and $\CostBound{\msteps,\timeouts}$ from Lemma~\ref{lem:value-bound} and Lemma~\ref{lem:cost-bound} to the the statement of Lemma~\ref{lem:perturbation-error}:
\begin{align*}
&2\cdot \alpha\cdot \CostBound{\msteps,\timeouts} \cdot (1 + \CostBound{\mcost,\timeouts}\cdot |\brambory|) \leq \\
&\leq 2 \cdot \frac{\eps}{8 \cdot \CostBound{\mcost,\timeouts} \cdot \CostBound{\msteps,\timeouts} \cdot |\statesM|} \cdot \CostBound{\msteps,\timeouts} \cdot (1 + \CostBound{\mcost,\timeouts}\cdot |\brambory|) \\
&\leq \frac{\eps}{4 \cdot \CostBound{\mcost,\timeouts} \cdot |\statesM|} \cdot 2 \cdot \CostBound{\mcost,\timeouts}\cdot |\brambory| \\
&\leq \frac{\eps}{2}.
\end{align*}
Note that for $\POconst$ from Proposition~\ref{prop:multi-entry-aprox} holds
\begin{align*}
\POconst &= 8 \cdot \CostBound{\mcost,\timeouts} \cdot \CostBound{\msteps,\timeouts} \cdot |\statesM| \\
&= 8 \cdot \left( \left( \frac{|\allact|}{\minPst \cdot \min \{1, \lambda \cdot \delayBound \}} \right)^{|\states|} \cdot \e^{\lambda \dmax \cdot |\states \setminus \allact|} \right)^2 \cdot |\statesM|^3 \cdot (\max\{ 1/\lambda, \; \dmax \cdot \lambda + \dmax \} + 1) \cdot \maxRew.
\end{align*}
Obviously $\POconst \in \exp((\size{\fdC}\cdot\size{\dmin}\cdot \dmax)^{\mathcal{O}(1)}).$
\end{proof}

\subsection{Proof of Lemma~\ref{lem:rounding}}

\begin{reflemma}{lem:rounding}
 Let $\eps>0$ and fix $\kappa=(\eps/8 \cdot \CostBound{\mcost,\timeouts} \cdot \CostBound{\msteps,\timeouts} \cdot |\statesM|)^2$. Then it holds that
	$$ \left\lvert \; \Value{\contMdp,\paramspace(\delta,\dmin, \dmax, \equiv)} - \Value{\matchmdp} \; \right\rvert 
	\;\; \leq \;\; \frac{\eps}{2}.$$
\end{reflemma}

\begin{proof}
We proceed similarly to proof of Proposition~\ref{prop:rounding-error}.  The only difference is that from Lemma~\ref{lem:value-bound} and Lemma~\ref{lem:cost-bound} we have bounds for arbitrary delay functions $\timeouts, \timeouts_\kappa \in \paramspace(\dmin, \dmax, \equiv)$, thus the argumentation is much simpler. We show only equation \eqref{eq:final-1} (for equation \eqref{eq:final-2} the argumentation is symmetric):
\begin{align*}
E_1 &= 2 \cdot \kappa \cdot\CostBound{\msteps,\timeouts}(1+\CostBound{\mcost,\timeouts}\cdot|\states'|)\\
&\leq 2 \cdot \left( \frac{\eps}{8 \cdot \CostBound{\mcost,\timeouts} \cdot \CostBound{\msteps,\timeouts} \cdot |\statesM|} \right)^2 \cdot\CostBound{\msteps,\timeouts}(1+\CostBound{\mcost,\timeouts}\cdot|\states'|)\\
&\leq 2 \cdot \frac{\eps}{16 \cdot \CostBound{\mcost,\timeouts} \cdot \CostBound{\msteps,\timeouts} \cdot |\statesM|} \cdot \CostBound{\msteps,\timeouts} \cdot (1 + \CostBound{\mcost,\timeouts}\cdot |\brambory|) \\
&\leq \frac{\eps}{8 \cdot \CostBound{\mcost,\timeouts} \cdot |\statesM|} \cdot  2 \cdot \CostBound{\mcost,\timeouts}\cdot |\brambory| \\
&\leq \frac{\eps}{4}.
\end{align*}
Note that for $\POconst$ from Proposition~\ref{prop:multi-entry-aprox} holds
\begin{align*}
\POconst &= 8 \cdot \CostBound{\mcost,\timeouts} \cdot \CostBound{\msteps,\timeouts} \cdot |\statesM| \\
&= 8 \cdot \left( \left( \frac{|\allact|}{\minPst \cdot \min \{1, \lambda \cdot \delayBound \}} \right)^{|\states|} \cdot \e^{\lambda \dmax \cdot |\states \setminus \allact|} \right)^2 \cdot |\statesM|^3 \cdot (\max\{ 1/\lambda, \; \dmax \cdot \lambda + \dmax \} + 1) \cdot \maxRew.
\end{align*}
Obviously $\POconst \in \exp((\size{\fdC}\cdot\size{\dmin}\cdot \dmax)^{\mathcal{O}(1)}).$
\end{proof}

\subsection{Proof of Correctness of the NP-Algorithm in Unary Case} \label{subsec:correctness}
Let us show that the presented algorithm is correct for the problem in Theorem~\ref{thm:np-complete}.
 First, assume that $\Value{\fdC,\paramspacePOShort} < \computedValue - \eps$. The optimal (MD) strategy $\timeouts$ in $\matchmdp_\paramspacePOShort$ from the set $\paramspacePO$ hence satisfies $\expred_{\matchmdp_\paramspacePOShort(\timeouts)} < \computedValue$, due to Proposition~\ref{prop:multi-entry-aprox}. The algorithm accepts as it non-deterministically guesses $\timeouts$. Similarly if $\Value{\fdC,\paramspacePOShort} > \computedValue + \eps$, no strategy from the set $\paramspacePOShort$ guarantees in $\matchmdp_\paramspacePOShort$ expected total cost $\leq \computedValue$, hence the algorithm does not accept. This concludes the proof that the problem from Theorem~\ref{thm:np-complete} is in NP.

 \subsection{Proof of Proposition~\ref{prop:np-hardness} (NP Hardness)}
 
 \begin{refproposition}{prop:np-hardness}
 	For a formula $\varphi$ in CNF with $k$ literals, $\fdC_\varphi$ and $\costFdC_\varphi$ are constructed in time polynomial in $k$ and, furthermore,
 	$$\Value{\fdC_\varphi,\paramspacePOShort} < 17 k^2 \;\; \text{if $\varphi$ is satisfiable \;\;\; and} \;\;\;\; \Value{\fdC_\varphi,\paramspacePOShort} > 17k^2 + 1 \text{, otherwise.}$$
 \end{refproposition}
 
 \begin{proof}
 	If $\varphi$ is satisfiable, let $\nu$ denote the satisfying truth assignment. Recall that we set $\dmin = 0.01$ and $\dmax = 16k$. We set $\timeouts(\initstate) := \dmin$ and $\timeouts(s_{i,j}^0) := \dmax$ if $\nu(X) = 1$ and $\timeouts(s_{i,j}^0) := \dmin$ if $\nu(X) = 1$, where $X$ is the variable of the literal $l_{i,j}$ and arbitrarily in other states.
 	In $\fdC_\varphi(\timeouts)$ in the component for any TRUE literal $l_{i,j}$, i.e. with $\nu(l_{i,1}) = 1$, the goal is reached from $s_{i,j}^0$ with probability $>0.99$ before leaving the component. Indeed, the probability to take no exponential transition within time $0.01$ is $>0.99$ and the probability to take at least $8k$ exponential transitions within time $16k$ is $>0.99$ for any $k\in\Nset$.
 	As each clause $\varphi_i$ has at most $k$ literals and at least one TRUE literal, the expected cost incurred in the component for $\varphi_i$ is at most $(16k\cdot k)/0.99 < 17k^2$.
 	
 	As regards the other implication, let $\varphi$ be not satisfiable such that $k \geq 7$ (note that we can assume it as when fixing the number of literals, the SAT problem becomes polynomial) and let $\timeouts \in \paramspacePO$. We need to show that $\expred_{\fdC_\varphi(\timeouts)} > 17k^2+1$. Based on $\timeouts$, we construct a truth assignment $\nu$ such that $\nu(X) = 0$ iff for all literals $l_{i,j}$ with variable $X$, we have $\timeouts(s_{i,j}) \leq k$. Since $\varphi$ is unsatisfiable, there must be at least one clause $i$ that is not satisfied w.r.t. $\nu$. It suffices to show that for the initial state of the component of the $i$-th clause we have $E_i := \expred_{\fdC_\varphi[s^0_{i,1}](\timeouts)} > k\cdot(17k^2+1)$. Indeed, even for the unrealistic worst case that the expected cost in components of other clauses is $0$, the overall expected cost is still $> 17k^2+1$.
 	
 	We know that all negative literals have the delay above $k\geq 7$ and all the positive literals have the delay below $k$. Let us assume the worst case that all the delays are $k$. 
 	With delay equal $k$, let us by $p_{pos}$ and $p_{neg}$ denote the probabilities that from the initial state of a component for a positive and negative literal, respectively, the goal is reached before leaving the component. It is easy show that for all $k \geq 7$:
 	$$
 	p_{pos} := 1-e^{-k}\cdot \sum_{n=0}^{8k-1}\frac{k^n}{n!} \;\;\leq\;\; \frac{1}{17k^2 +1} \text{, and}
 	\qquad\qquad
 	p_{neg} := 
 	e^{-k}
 	\;\;\leq\;\; \frac{1}{17k^2 +1}.
 	$$
 	Based on this observation, we can under-approximate $E_i$ by 
 	$$E_i 
 	\;\geq\; 
 	\frac{k}{\frac{1}{17k^2+1}}
 	\;= \;
 	k\cdot(17k^2+1).$$ \qed
 \end{proof}

\end{document}